\documentclass[runningheads,a4paper]{llncs}

\usepackage{amstext,amsmath,amssymb,url}
\allowdisplaybreaks
\usepackage{newunicodechar, pifont, stix, float}
\usepackage{listings}
\makeatletter
\newcommand{\verbatimfont}[1]{\renewcommand{\verbatim@font}{\ttfamily#1}}
\makeatother

\usepackage{parskip}

\usepackage[page,titletoc]{appendix}
\makeatletter
\renewcommand*\l@author[2]{}
\renewcommand*\l@title[2]{}
\makeatletter

\usepackage[hide]{ed}
\usepackage{paralist}
\usepackage{basics}
\usepackage{multirow}

\usepackage[bookmarks,bookmarksopen,colorlinks,linkcolor=blue,citecolor=blue,bookmarksnumbered,urlcolor=gray]{hyperref}

\usepackage{xfrac}
\usepackage{tikz}
\usepackage{array}

\newunicodechar{❘}{\\&\quad}
\newunicodechar{❙}{\\&}
\newunicodechar{❚}{}
\newunicodechar{Π}{\ensuremath{\ \Pi\ }}
\DeclareUnicodeCharacter{80}{\ensuremath{\ @\ }}
\DeclareUnicodeCharacter{036D}{}

\newunicodechar{⟶}{\ensuremath{\rightarrow}}
\newunicodechar{⇒}{\ensuremath{\Rightarrow}}
\newunicodechar{⇔}{\ensuremath{\Leftrightarrow}}
\newunicodechar{∧}{\ensuremath{~\land~}}
\newunicodechar{∨}{\ensuremath{~\lor~}}
\newunicodechar{≐}{\ensuremath{~\doteq~}}
\newunicodechar{⊦}{\ensuremath{\vdash}}
\newunicodechar{⦂}{\ensuremath{\typecolon}}
\newunicodechar{π}{\ensuremath{\pi}}
\newunicodechar{λ}{\ensuremath{\lambda}}

\newcommand{\hol}{HOL}

\newcommand{\dhol}{DHOL}
\newcommand{\dphol}{DHOL}

\newcommand{\lambdaFun}[2]{\lambda #1\!:\!#2.~}
\newcommand{\piType}[2]{\Pi #1\!:\!#2.~}
\newcommand{\subtype}[2]{#1|_{#2}\hspace{0.4mm}}
\newcommand{\univQuant}[2]{\ensuremath{\forall \,#1\!:\!#2.~}}

\newcommand{\judg}{\Gamma\dedT}

\usepackage{etoolbox}

\newcommand{\existQuant}[2]{\ensuremath{\exists #1\!:\!#2.}}

\newcommand{\Thy}[1]{\ensuremath{#1\ \mathsf{Thy}}}
\newcommand{\Ctx}[1]{\ensuremath{#1\ \mathsf{Ctx}}}
\newcommand{\namedax}[2]{#1\!:#2}
\newcommand{\namedass}[2]{#1\!:#2}
\newcommand{\typingAxName}[1]{\ensuremath{#1^*}}
\newcommand{\typingAx}[2]{\namedax{\typingAxName{#2}}{\PredPhi{#1}{#2}}}
\newcommand{\typingAssName}[1]{\ensuremath{#1^*}}
\newcommand{\typingAss}[2]{\namedass{\typingAssName{#2}}{\PredPhi{#1}{#2}}}
\newcommand{\concatCtx}[2]{#1,\;#2}
\newcommand{\concatThy}[2]{#1,\;#2}
\newcommand{\emptyThy}{\circ}
\newcommand{\emptyCtx}{.}
\newcommand{\ctxIn}[2]{\ensuremath{#1\text{ in }#2}}
\newcommand{\thyIn}[2]{\ensuremath{#1\text{ in }#2}}
\newcommand{\Type}[1]{\ensuremath{#1\ \type}}
\newcommand{\defaultTerm}[1]{\ensuremath{w_{#1}}}

\usepackage{scalerel,stackengine}
\newcommand\reallywidetilde[1]{%
	\savestack{\tmpbox}{\stretchto{%
			\scaleto{%
				\scalerel*[\widthof{\ensuremath{#1}}]{\kern-.6pt\sim\kern-.6pt}%
				{\rule[-\textheight/2]{1ex}{\textheight}}
			}{\textheight}%
		}{0.5ex}}%
	\stackon[1pt]{#1}{\tmpbox}%
}
\newcommand\widevec[1]{%
\savestack{\tmpbox}{\stretchto{%
		\scaleto{%
			\scalerel*[\widthof{\ensuremath{#1}}]{\kern-.6pt\kern-.6pt}%
			{\rule[-\textheight/2]{1ex}{\textheight}}
		}{\textheight}%
	}{0.5ex}}%
\stackon[1pt]{#1}{\tmpbox}%
}

\newcommand{\PredPhi}[2]{\PredPhiName{#1}\ #2\ #2}
\newcommand{\PredPhiName}[1]{{#1}^*}
\newcommand{\PhiAppl}[1]{\overline{#1}}
\newcommand{\quasiImage}[1]{\reallywidetilde{#1}}
\newcommand{\termEqT}[3]{\PredPhiName{#1}\ #2\ #3}
\newcommand{\pbool}{\PredPhiName{\bool}}
\newcommand{\boolPred}[1]{\pbool\ #1\ #1}

\newcommand{\betaEtaRed}[1]{\ensuremath{#1^{\beta\eta}}}
\newcommand{\norm}[1]{\ensuremath{\mathsf{norm}\left[#1\right]}}
\newcommand{\sRed}[1]{\ensuremath{\mathsf{sRed}\left(#1\right)}}

\newcommand{\reserved}[1]{\ensuremath{\mathsf{#1}}}
\usepackage{xspace,xcolor,realboxes}
\definecolor{backcolor}{gray}{.92}

\newcommand{\lowlight}[1]{\textcolor{gray}{#1}}
\newcommand{\symbolFont}[1]{\ensuremath{\mathtt{#1}}}

\newcommand{\type}{\reserved{tp}}

\newcommand{\bool}{\reserved{bool}}
\newcommand{\T}{\reserved{true}}
\newcommand{\F}{\reserved{false}}

\newcommand{\dedH}{\ensuremath{\vdash^H}}
\newcommand{\dedT}{\ensuremath{\vdash}_T}

\newcommand{\dedPT}{\ensuremath{\vdash}_{\PsiAppl T}}
\newcommand{\ded}{\ensuremath{\vdash\,}}
\newcommand{\termEquals}[1]{\ensuremath{\,=_{#1}\,}} 

\newcommand{\termEqB}{\termEquals{\bool}}
\newcommand{\typeEquals}{\ensuremath{\,\equiv\,}}
\newcommand{\subtyping}{\ensuremath{\,<:\,}}
\newcommand{\teq}{\typeEquals} 
\newcommand{\substOp}[2]{\ensuremath{[\sfrac{#1}{#2}]}}
\newcommand{\subst}[3]{\ensuremath{#1\substOp{#2}{#3}}}
\newcommand{\NDLine}[3]{#1\ded& #2 &&\text{#3}}
\newcommand{\NDLineH}[3]{#1\dedH& #2 &&\text{#3}}
\newcommand{\NDLinePT}[3]{#1\dedPT& #2 &&\text{#3}}
\newcommand{\NDLinePTG}[2]{\PhiAppl{\Gamma}\dedPT& #1 &&\text{#2}}
\newcommand{\NDLinePTD}[2]{\NDLinePT{\Delta}{#1}{#2}}
\newcommand{\NDLinePTT}[2]{\NDLinePT{\Theta}{#1}{#2}}
\newcommand{\NDLineT}[3]{#1\dedT& #2 &&\text{#3}}

\newcommand{\NDLineTG}[2]{\Gamma\dedT& #1 &&\text{#2}}
\newcommand{\IH}{induction hypothesis}
\newcommand{\byAss}{by assumption}

\ifcsmacro{grammar}{
	\renewenvironment{grammar}{\[\begin{array}{l@{\ \bbc\quad}l@{\quad}l}}{\end{array}\]}}{
	\newenvironment{grammar}{\[\begin{array}{lll}}{\end{array}\]}
	\newcommand{\bnf}[1]{#1}
	\newcommand{\bbc}{\bnf{::=}}

	\newcommand{\opt}[1]{\bnf{[}#1\bnf{]}}
	\newcommand{\bnfbracket}[1]{\bnf{(}#1\bnf{)}}
	\newcommand{\rep}[1]{#1^{\bnf{\ast}}}

	}

\makeatletter
\newcommand{\rulelabelAppendix}[2]{%
	\protected@write \@auxout {}{\string \newlabel {#1}{{#2}{\thepage}{#2}{#1}{}} }%
	\hypertarget{#1}{}%
}
\makeatother
\newcommand{\rulelabel}[2]{}

\newcommand{\namedRule}[3]{\rul[\text{}]{#3}{#2}}
\newcommand{\rnamedRule}[4]{\rul[\text{}]{#3}{#2}}
\newcommand{\snamedRule}[3]{\rul[\text{}]{#3}{#2}}

\newcommand{\ruleRef}[1]{(\ref{#1})}

\newcommand{\PrefLabelsSuff}[4]{&\label{#1}\addtocounter{#3}{1}\tag{#4\arabic{#3}#2}}

\newcommand{\sredlabel}[1]{\PrefLabelsSuff{#1}{}{SRlabel}{SR}}

\newbool{inAppendix}
\setbool{inAppendix}{false}
\newcommand{\plabel}[1]{\ifbool{inAppendix}{\PrefLabelsSuff{#1}{}{transDefnLabelsP}{PT}}{}}

\newcommand{\Impl}{\tb\text{implies}\tb}
\newcommand{\Mand}{\;\text{and}\;}

\usepackage{comment}
\excludecomment{oldCode}
\includecomment{codeBlock}

\newcommand{\negSp}{\!\!\!\!}
\newcommand{\QNegSp}{\negSp\negSp}
\newcommand{\Quad}{\quad\ \ \!\!\,}
\newcommand{\QQNegSp}{\QNegSp\QNegSp}
\newcommand{\QQuad}{\Quad\Quad}
\newcommand{\QQQNegSp}{\QQNegSp\QQNegSp}
\newcommand{\QQQuad}{\QQuad\QQuad}
\newcommand{\QQQQuad}{\QQQuad\QQQuad}
\newcommand{\QQQQNegSp}{\QQQNegSp\QQQNegSp}
\newcommand{\QQQQQNegSp}{\QQQQNegSp\QQQQNegSp}

\newcommand{\qqquad}{\qquad\qquad}
\newcommand{\qqqquad}{\qqquad\qqquad}
\newcommand{\qqqqquad}{\qqqquad\qqqquad}
\newcommand{\qqqqqquad}{\qqqqquad\qqqqquad}
\newcommand{\qqqqqqquad}{\qqqqqquad\qqqqqquad}

\newcommand{\breakintertext}{\qqqqqqquad\-}

\renewcommand{\id}{\symbolFont{id}}
\newcommand{\ident}[1]{\ensuremath{\id_{#1}}}
\renewcommand{\obj}{\symbolFont{obj}}
\renewcommand{\mor}{\symbolFont{mor}}
\newcommand{\comp}{\symbolFont{comp}}

\newtheorem{thm}{Theorem}

\newtheorem{rem}[thm]{Remark}

\title{Theorem Proving in Dependently-Typed \\ Higher-Order Logic\,---\,Extended Preprint}

\author{Colin Rothgang\inst{1}\orcidID{0000-0001-9751-8989} \and
Florian Rabe\inst{2}\orcidID{0000-0003-3040-3655} \and Christoph Benzm\"uller\inst{3,1}\orcidID{0000-0002-3392-3093}}
\authorrunning{Benzm\"uller, Rabe, Rothgang}
\institute{Mathematics and Computer Science, FU Berlin, Germany \and Computer Science, University Erlangen-N\"urnberg, Germany \and AI Systems Engineering, University Bamberg}
\begin{document}
\maketitle

\begin{abstract}
Higher-order logic \hol{} offers a very simple syntax and semantics for representing and reasoning about typed data structures.
But its type system lacks advanced features where types may depend on terms.
Dependent type theory offers such a rich type system, but has rather substantial conceptual differences to HOL, as well as comparatively poor proof automation support.

We introduce a dependently-typed extension DHOL of HOL that retains the style and conceptual framework of HOL.
Moreover, we build a translation from DHOL to HOL and implement it as a preprocessor to a HOL theorem prover, thereby obtaining a theorem prover for DHOL.
\end{abstract}


\section{Introduction and Related Work}\label{sec:intro}
Theorem proving in higher-order logic (HOL) \cite{churchtypes,holsemantics} has been a long-running research strand producing multiple mature interactive provers \sloppy\cite{isabelle,hollight,hol} and automated provers \cite{tps,leo3,satallax}.
Similarly, many, mostly interactive, theorem provers are available for various versions of dependent type theory (DTT) \cite{coq,agda,lean,twelf}.
However, it is (maybe surprisingly) difficult to develop theorem provers for dependently-typed higher-order logic (DHOL).

In this paper, we use HOL to refer to a version of Church's \emph{simply}-typed $\lambda$-calculus with a base type $\bool$ for Booleans, simple function types $\to$, and equality $=_A:A\to A\to \bool$.
This already suffices to define the usual logical quantifiers and connectives.\footnote{We do not assume a choice operator or the axiom of infinity.}
Intuitively, it is straightforward to develop DHOL accordingly on top of the \emph{dependently}-typed $\lambda$-calculus, which uses a dependent function type $\piType{x}{A}B$ instead of $\to$.
However, several subtleties arise that seem deceptively minor at first but end up presenting fundamental theoretical issues.
They come up already in the elementary expression $x=_A y \impl f(x)=_{B(x)} f(y)$ for some dependent function $f:\piType{x}{A}B(x)$.


\textbf{Firstly}, the equality $f(x)\termEquals{B(x)} f(y)$ is not even well-typed because the terms $f(x):B(x)$ and $f(y):B(y)$ do not have the same type.
Intuitively, it is obvious that the type system can (and maybe should) be adjusted so that the equality $x=_A y$ between terms carries over to an equality $B(x)\teq B(y)$ between types.\footnote{Note that while term equality $=_A$ is a $\bool$-valued connective, type equality $\teq$ is not. Instead, in HOL, $\teq$ is a judgment at the same level as the typing judgment $t:A$.}
However, this means that the undecidability of equality leaks into the equality of types and thus into type-checking.

While some interactive provers successfully use undecidable type systems \cite{pvs,nuprl}, most formal systems for DTT commit to keeping type-checking decidable.
The typical approach goes back to Martin-L\"of type theory \cite{martinlof} and the calculus of constructions \cite{calcconstructions} and uses two separate equality relations, a decidable meta-level equality for use in the type-checker and a stronger undecidable one subject to theorem proving.
Moreover, it favors the propositions-as-types representation and deemphasizes or omits a type of classical Booleans.
This approach has been studied extensively \cite{coq,agda,lean} and is not the subject of this paper.

Instead, our motivation is to retain a single equality relation and classical Booleans.
This is arguably more intuitive to users, especially to those outside the DTT community such as typical HOL users or mathematicians, and it is certainly much closer to the logics of the strongest available ATP systems.
This means we have to pay the price of undecidable type-checking.
The current paper was prompted by the observation that this price may be acceptable for two reasons:
\begin{compactenum}
	\item If our ultimate interest is theorem proving, undecidability comes up anyway.
	Indeed, it is plausible that the cost of showing the well-typedness of a conjecture will be negligible compared to the cost of proving it.
	\item As the strength of ATPs for HOL increases, the practical drawbacks of undecidable type-checking decrease, which indicates revisiting the trade-off from time to time.
	Indeed, if we position DHOL close to an existing HOL ATP, it is plausible that the price will, in practice, be affordable.
\end{compactenum}\ednote{CB: It would be good to conduct experiments and find out (for large benchmark libraries) how severe the undecidability problem really is.}

\textbf{Secondly}, even if we add a rule like ``if $\vdash x=_A y$, then $\vdash B(x)\teq B(y)$'' to our type system, the above expression is still not well-typed:
Above, the equality $x=_A y$ on the left of $\impl$ is needed to show the well-typedness of the equality $f(x)=_{B(x)}f(y)$ on the right.
This intertwines theorem proving and type-checking even further.
Concretely, we need a \emph{dependent implication}, where the first argument is assumed to hold while checking the well-typedness of the second one.
Formally, this means to show $\vdash F\impl G:\bool$, we require $\vdash F:\bool$ and $F\vdash G:\bool$.
Similarly, we need a dependent conjunction.
And if we are classical, we may also opt to add a dependent disjunction $F\vee G$, where $\neg F$ is assumed in $G$.
Naturally, dependent conjunction and disjunction are not commutative anymore. This may feel disruptive, but similar behavior of connectives is well-known from short-circuit evaluation in programming languages.

The meta-logical properties of dependent connectives are straightforward.
However, interestingly, they can no longer be defined from just equality.
At least one of them (we will choose dependent implication) must be taken as an additional primitive in DHOL along with $=_A$.

\textbf{Finally}, the above generalizations require a notion of DHOL-contexts that is more complex than for HOL.
HOL-contexts can be stratified into (a) a set of variable declarations $x_i:A_i$, and (b) a set of logical assumptions $F$ possibly using the variables $x_i$.
Moreover, the former are often not explicitly listed at all and instead inferred from the remainder of the sequent.
But in DHOL, the well-formedness of an $A_i$ may now depend on previous logical assumptions.
To linearize this inter-dependency, DHOL contexts must consist of a single list alternating between variable declarations and assumptions.

\paragraph{Contribution.}
Our contribution is twofold.
Firstly, we introduce a new logic DHOL designed along the lines described above.
Moreover, we further extend DHOL with predicate subtypes $\subtype A p$ for a predicate $p:A\to\bool$ on the type $A$.
Besides dependent types, these constitute a second important source of terms occurring in types.
Because they also make typing undecidable, they are often avoided.
The most prominent exception is PVS \cite{pvs}, whose kernel essentially arises by adding predicate subtypes to HOL.
In current HOL ITPs going back to \cite{hol}, their use is usually restricted to the subtype definition principle: here a definition $b:=\subtype A p$ may occur on toplevel and is elaborated into a fresh type $b$ that is axiomatized to mimic the subtype $\subtype A p$.
Because we are committed to undecidable typing anyway, predicate subtypes fit naturally into our approach.

Secondly, we develop and implement a sound and complete translation of DHOL into HOL.
This setup allows the use of DHOL as the expressive user-facing language and HOL as the internal theorem-proving language.
We position our implementation close to an existing HOL ATP, namely the LEO-III system.
From the LEO-III perspective, DHOL serves as an additional input language that is translated into HOL by an external logic embedding tool \cite{leo_embedding,S22} in the LEO-III ecosystem.
Because LEO-III already supports such embeddings and because the TPTP syntax \cite{tptp} foresees the use of dependent types in ATPs and provides syntax for them (albeit without a normative semantics), we were able to implement the translation with no disruptions to existing workflows.

The general idea of our translation of dependent into simple type theory is not new \cite{jacobs_dtt_hol}.
In that work, Martin-L\"of-style dependent type theory is translated into Gordon's HOL ITP \cite{hol}.
This work differs critically from ours because it uses DTT in propositions-as-types style.
Our work builds DHOL with classical Booleans and equality predicate, which makes the task of proving the translation sound and complete very different.
Moreover, their work targets an interactive prover while ours targets automated ones.

\paragraph{Overview.}
In Sect.~\ref{sec:hol} we recap the HOL logic.
In Sect.~\ref{sec:dhol} we extend it to DHOL and define our translation from DHOL to HOL.
In Sect.~\ref{sec:dphol} we add subtyping and predicate subtypes.
In Sect.~\ref{sec:meta} we prove the soundness and completeness of the translation.
In Sect.~\ref{sec:impl} we describe how to use our translation and a HOL ATP to implement a theorem prover for DHOL.

\paragraph{Acknowledgments.} Chad Brown and Alexander Steen provided valuable feedback on earlier versions of this paper.

\section{Preliminaries: Higher-Order Logic}\label{sec:hol}
We introduce the syntax and rules of HOL.
Our definitions are standard except that we tweak a few details in order to later present the extension to DHOL more succinctly.
We use the following grammar for HOL:
\begin{grammar}
	T   	& \emptyThy \alt T,\, a:\type \alt T,\,c:A \alt T,\,\namedax{c}{F}          	      &\text{theories}\\
	\Gamma & \emptyCtx\alt \Gamma,x:A \alt \Gamma,\namedass{x}{F}									            &\text{contexts} \\
	A, B	& a \alt A\to B \alt \bool 									&\text{types}\\
	s,t,f,F,G & c\alt x\alt \lambdaFun{x}{A} t \alt f\ t \alt  s \termEquals{A} t \alt F\Rightarrow G &\text{terms}\\
\end{grammar}
A theory $T$ is a list of base type declarations $a:\type$, typed constant declarations $c:A$, and named axioms $\namedax{c}{F}$ asserting the formula $F$.
A context $\Gamma$ has the same form except that no type variables are allowed.
It is not strictly necessary to use named axioms and assumptions, but it makes our extensions to DHOL later on simpler.
We write $\emptyThy$ and $\emptyCtx$ for the empty theory and context, respectively.
At this point, it is possible to normalize contexts into a set of variable declarations followed by a set of assumptions because the well-formedness of a type $A$ can never depend on a variable or an assumption.
But that property will change when going to DHOL, which is why we allow $\Gamma$ to alternate between variables and assumptions.

Types $A$ are either user-declared types $a$, the built-in base type $\bool$, or function types $A\to B$.
Terms are constants $c$, variables $x$, $\lambda$-abstractions $\lambdaFun{x}{A} t$, function applications $f\ t$, or obtained from the built-in $\bool$-valued connectives $\termEquals{A}$ or $\Rightarrow$.
As usual \cite{andrews_truthproof}, this suffices to define all the usual quantifiers and connectives $\T$, $\F$, $\neg$, $\land$, $\lor$, $\forall$ and $\exists$.
This includes $\Rightarrow$, but we make it a primitive here because we will change it in DHOL.
As usual, $\subst{E}{x}{t}$ denotes the capture-avoiding substitution of the variable $x$ with the term $t$ within expression $E$.

The type and proof system uses the judgments given below.
Note that we need a meta-level judgment for the equality of types because $\typeEquals$ is not a $\bool$-valued connective.
On the contrary, the equality of terms $\vdash s\termEquals{A} t$ is a special case of the validity judgment $\vdash F$.
In HOL, $\typeEquals$ is trivial, and the judgment is redundant.
But we include it here already because it will become non-trivial in DHOL.

\begin{center}
\begin{tabular}{|c|c|c|}
	\hline 
	Name & Judgment & Intuition \\ 
	\hline
	theories & $\ded \Thy{T}$ & $T$ is well-formed theory\\
	contexts & $\dedT \Ctx{\Gamma}$ & $\Gamma$ is well-formed context\\
	types & $\Gamma\dedT A\, \type$ & $A$ is well-formed type \\ 
  	typing & $\Gamma\dedT t :A$ & $t$ is a well-formed term of type well-formed type $A$ \\
	validity & $\Gamma\dedT F$ & well-formed Boolean $F$ is provable \\ 
	equality of types & $\Gamma\dedT A\typeEquals B$ & well-formed types $A$ and $B$ are equal \\ 
	\hline 
\end{tabular}
\end{center}

The rules are given in Figure~\ref{fig:holrulesPaper}.
We assume that all names in a theory or a context are unique without making that explicit in the rules. Following common practice, we further assume that HOL types are non-empty.
See also Appendix~\ref{sec:logics} for a detailed definition description.

\begin{figure}[hpt]
	{\ifbool{inAppendix}{\scriptsize}{\small}
		Theories and contexts:
		\[
		\snamedRule{thyEmpty}{\ded\Thy{\emptyThy}}{}\tb
		\snamedRule{thyType}{\ded\Thy{\concatThy{T}{\Type{a}}}}{\ded\Thy{T}}\tb
		\snamedRule{thyConst}{\ded\Thy{\concatThy{T}{c:A}}}{\dedT \Type{A}}\tb
		\snamedRule{thyAxiom}{\ded\Thy{\concatThy{T}{\namedax{c}{F}}}}{\dedT F:\bool}
		\]
		\[
		\snamedRule{ctxEmpty}{\dedT\Ctx{\emptyCtx}}{\ded\Thy{T}}\tb
		\snamedRule{ctxVar}{\dedT\Ctx{\concatCtx{\Gamma}{x:A}}}{\Gamma\dedT \Type{A}}\tb
		\snamedRule{ctxAssume}{\dedT\Ctx{\concatCtx{\Gamma}{\namedass{x}{F}}}}{\Gamma\dedT F:\bool}
		\]\\
		Lookup in theory and context:
		\[
		\snamedRule{type}{\Gamma \dedT \Type{a}}{\thyIn{a:\type}{T} \tb \dedT\Ctx{\Gamma}}\tb
		\rnamedRule{const}{\Gamma\dedT c:A}{\thyIn{c:A'}{T} \tb\Gamma\dedT A'\typeEquals A}{const}\tb
		\snamedRule{axiom}{\Gamma\dedT F}{\thyIn{\namedax{c}{F}}{T} \tb  \dedT\Ctx{\Gamma}}
		\]
		\[
		\rnamedRule{var}{\Gamma\dedT x:A}{\ctxIn{x:A'}{\Gamma} \tb\Gamma\dedT A'\typeEquals A}{var} \tb
		\snamedRule{assume}{\Gamma\dedT F}{\ctxIn{\namedass{x}{F}}{\Gamma} \tb  \dedT\Ctx{\Gamma}}
		\]\\
		Well-formedness and equality of types:
		\[
		\snamedRule{bool}{\Gamma\dedT \Type{\bool}}{\dedT\Ctx\Gamma}\tb
		\snamedRule{arrow}{\Gamma\dedT \Type{A\to B}}{\Gamma\dedT \Type{A} \tb \Gamma\dedT \Type{B}}
		\tb
		\snamedRule{congBase}{\Gamma\dedT A\typeEquals A}{\Gamma\dedT \Type{A}} \tb
		\rnamedRule{cong$\to$}{\Gamma\dedT A\to B\typeEquals A'\to B'}{\Gamma\dedT A\typeEquals A' \tb  \Gamma\dedT B\typeEquals B'}{congTo}
		\]\\
		Typing: 
		\[
		\snamedRule{lambda}{\Gamma\dedT (\lambdaFun{x}{A} t): A\to B}{\concatCtx{\Gamma}{x:A}\dedT t:B}\tb
		\snamedRule{appl}{\Gamma\dedT f\ t:B }{\Gamma\dedT f:A\to B \tb \Gamma\dedT t:A}\tb
		\rnamedRule{$=$type}{\Gamma\dedT s\termEquals{A}t:\bool}{\Gamma\dedT s:A\tb \Gamma\dedT t:A}{eqType}
		\]\\
		Term equality: congruence, reflexivity, symmetry, $\beta$, $\eta$
		\[
		\namedRule{cong$\lambda$ (xi)}{\Gamma\dedT \lambdaFun{x}{A} t\termEquals{A\to B} \lambdaFun{x}{A'} t'}{\Gamma\dedT A\typeEquals A' \tb\concatCtx{\Gamma}{x:A}\dedT t\termEquals{B} t'}\rulelabel{congLam}{cong$\lambda$}
		\tb
		\snamedRule{congAppl}{\Gamma\dedT f\ t\termEquals{B} f'\ t'}{\Gamma\dedT t\termEquals{A} t'\tb  \Gamma\dedT f\termEquals{A\to B} f'}
		\]
		\[
		\snamedRule{refl}{\Gamma\dedT t\termEquals{A} t}{\Gamma\dedT t:A}\tb
		\snamedRule{sym}{\Gamma\dedT s\termEquals{A} t}{\Gamma\dedT t \termEquals{A}s}\tb
		\snamedRule{beta}{\Gamma\dedT (\lambdaFun{x}{A} s)\ t \termEquals{B} \subst{s}{x}{t} }{\Gamma\dedT (\lambdaFun{x}{A} s)\ t:B}\tb
		\snamedRule{eta}{\Gamma\dedT t\termEquals{A\to B} \lambdaFun{x}{A}t\ x}{\Gamma\dedT t:A\to B\ \tb\ctxIn{x\text{ not}}{\Gamma}}
		\]\\
		Rules for implication:
		\[
		\rnamedRule{$\impl$type}{\Gamma\dedT F\Rightarrow G:\bool}{\Gamma\dedT F:\bool\tb \Gamma\dedT G:\bool}{implType}\tb
		\rnamedRule{$\impl$I}{\Gamma\dedT F\Rightarrow G}{\Gamma\dedT F:\bool \tb \concatCtx{\Gamma}{\namedass{x}{F}}\dedT G}{implI}\tb
		\rnamedRule{$\impl$E}{\Gamma\dedT G}{\Gamma\dedT F\Rightarrow G\tb \Gamma\dedT F}{implE}
		\]\\
		Congruence for validity, Boolean extensionality, and non-emptiness of types:
		\[
		\rnamedRule{cong$\ded$}{\Gamma\dedT F}{\Gamma\dedT F\termEquals{\bool} F'\tb \Gamma\dedT F'}{congDed}\tb
		\snamedRule{boolExt}{\Gamma,x:\bool\dedT p\ x}{\Gamma\dedT p\ \T \tb \Gamma\dedT p\ \F}\tb
		\snamedRule{nonempty}{\Gamma\dedT F}{\Gamma\dedT F:\bool\tb \Gamma,x:A\dedT F}
		\]
	}
	\caption{\hol{} Rules}
	\ifbool{inAppendix}{\vspace*{-1cm}\label{fig:holrulesAppendix}}{\label{fig:holrulesPaper}}
\end{figure}



%

\section{Dependent Function Types}\label{sec:dhol}
\subsection{Language}

We have carefully defined HOL in such a way that only a few surgical changes are needed to define DHOL.
A consolidated definition of DHOL is given in Appendix~\ref{sec:logics}.
The \textbf{grammar} is as follows where unchanged parts are \lowlight{shaded out}:
\begin{grammar}
	T   	& \lowlight{\emptyThy \alt T,\,a:}(\piType{x}{A})^*\lowlight{\type \alt T,\,c:A \alt T,\,c:F}         				&\text{theories}\\
	\Gamma     & \lowlight{\emptyCtx\alt \Gamma,x:A \alt \Gamma,ass:F}	             										&\text{contexts} \\
	A, B	& \lowlight{a}\ t_1\ldots t_n \alt \piType{x}{A} B \alt \lowlight{\bool}
																	&\text{types}\\
	s,t,f,F,G & \lowlight{c}\alt \lowlight{x} \alt \lowlight{\lambdaFun{x}{A} t \alt \lowlight{f\ t}\alt s\termEquals{A}t \alt F\impl G} &\text{terms}\\
\end{grammar}
Concretely, base types $a$ may now take term arguments and simple function types $A\to B$ are replaced with dependent function types  $\piType{x}{A} B$.
As usual we will retain the notation $A\to B$ for the latter if $x$ does not occur free in $B$.
DHOL is a conservative extension of HOL, and we recover HOL as the fragment of DHOL in which all base types $a$ have arity $0$.

\renewcommand{\id}{\symbolFont{id}}
\renewcommand{\obj}{\symbolFont{obj}}
\renewcommand{\mor}{\symbolFont{mor}}
\newcommand{\neutl}{\symbolFont{neutL}}
\newcommand{\neutr}{\symbolFont{neutR}}
\newcommand{\assoc}{\symbolFont{assoc}}
\begin{example}[Category Theory]\label{ex:cat}
As a running example, we formalize the theory of a category in DHOL.
It declares the base type $obj$ for objects and the dependent base type $\mor\ a\ b$ for morphisms.
Further it declares the constants $id$ and $\comp$ for identity and composition, and the axioms for neutrality. We omit the associativity axiom for brevity.
	\begin{align*}
		\obj:& \type\\
		\mor:& \piType{x,y}{\obj} \type\\
		\id:& \piType{a}{\obj} \mor\ a\ a\\
		\comp:& \piType{a,b,c}{\obj} \mor\ a\ b\to \mor\ b\ c\to \mor\ a\ c\\
		\neutl:&\forall x,y:\obj.\forall m:\mor\ x\ y.\;\;  m\circ \ident{x}\termEquals{\mor\ x\ y} m\\
		\neutr:&\forall x,y:\obj.\;\forall m:\mor\ x\ y.\;\;  \ident{y}\circ m\termEquals{\mor\ x\ y} m
	\end{align*}
Here we use a few intuitive notational simplifications such as writing $\piType{x,y}{\obj}$ for binding two variables of the same type.
We also use the notations $\ident{x}$ for $\id\ x$ and $h\circ g$ for $\comp\ \_\ \_\ \_\ g\ h$ where the $\_$ denote inferable arguments of type $\obj$.
\end{example}

The \textbf{judgments} stay the same and we only make minor changes to the \textbf{rules}, which we explain in the sequel.
Firstly we replace all rules for $\to$ with the ones for $\Pi$:

{\small
\[	
\snamedRule{pi}{\lowlight{\Gamma\dedT} \piType{x}{A} B\ \lowlight{\type}}{\lowlight{\Gamma\dedT \Type{A} \quad }\concatCtx{\lowlight{\Gamma}}{x:A}\lowlight{\dedT \Type{B}}}\tb \rnamedRule{cong$\Pi$}{\lowlight{\Gamma\dedT}\piType{x}{A} B\lowlight{\typeEquals} \piType{x}{A'}B'}{\lowlight{\Gamma\dedT A\typeEquals A' \quad  \Gamma},x:A\lowlight{\dedT B\typeEquals B'}}{congPi}
\]
\[
\snamedRule{lambda'}{\lowlight{\Gamma\dedT (\lambdaFun{x}{A} t):} \piType{x}{A} B}{\lowlight{\concatCtx{\Gamma}{x:A}\dedT t:B}}\tb
\snamedRule{appl'}{\lowlight{\Gamma\dedT f\,t:}\subst{B}{x}{t}}{\lowlight{\Gamma\dedT f:}\piType{x}{A} B \tb \lowlight{\Gamma\dedT t:A}}\]
\[
\rnamedRule{cong$\lambda$'}{\lowlight{\Gamma\dedT \lambdaFun{x}{A} t}\termEquals{\piType{x}{A}B} \lowlight{\lambdaFun{x}{A'} t'}}{\lowlight{\Gamma\dedT A\typeEquals A' \ \;\concatCtx{\Gamma}{x:A} \dedT t\termEquals{B} t'}}{congLam'}
\quad	
\snamedRule{congAppl'}{\lowlight{\Gamma\dedT f\ t\termEquals{B} f'\ t'}}{\lowlight{\Gamma\dedT t\termEquals{A} t'}\quad \lowlight{\Gamma\dedT f}\termEquals{\piType{x}{A} B} \lowlight{f'}}
\]
\[	
\snamedRule{etaPi}{\lowlight{\Gamma\dedT t}\termEquals{\piType{x}{A} B} \lowlight{\lambdaFun{x}{A}t\ x}}{\lowlight{\Gamma\dedT t:}\piType{x}{A} B}
\]
}

Then we replace the rules for declaring, using, and equating base types with the ones where base types are applied to arguments:
{
\[
\snamedRule{thyType'}{\lowlight{\ded}\Thy{\concatThy{\lowlight{T}}{\lowlight{a:}\piType{x_1}{A_1}\ldots\piType{x_n}{A_n}\lowlight{\type}}}}
{\dedT \Ctx{x_1:A_1,\,\ldots,x_n:A_n}}
\]
\[
\snamedRule{type'}{\lowlight{\Gamma \dedT a}\ t_1\ \ldots\ t_n\;\lowlight{\type}}
{\begin{array}{c}\lowlight{\dedT \Ctx{\Gamma}}\ \;\lowlight{a:}\piType{x_1}{A_1}\ldots\piType{x_n}{A_n}\lowlight{\type \text{ in }T}\\
	\Gamma \dedT t_1:A_1 \ \;\ldots\ \; \Gamma \dedT t_n:\subst{A_n}{x_1}{t_1}\ldots\substOp{x_{n-1}}{t_{n-1}}\end{array}}
\]
\[
\snamedRule{congBase'}{\lowlight{\Gamma\dedT a}\ s_1\ \ldots\ s_n\lowlight{\typeEquals} \lowlight{a}\ t_1\ \ldots t_n}
{\begin{array}{c}\lowlight{\dedT \Ctx{\Gamma}}\ \;\lowlight{a:}\piType{x_1}{A_1}\ldots\piType{x_n}{A_n}\lowlight{\type\text{ in }T}\\
\Gamma \dedT s_1\termEquals{A_1}t_1 \ \ldots\ \Gamma \dedT s_n\termEquals{\subst{A_n}{x_1}{t_1}\ldots\substOp{x_{i-1}}{t_{i-1}}}t_n\end{array}}
\]
}

The last of these is the critical rule via which term equality leaks into type equality.
Thus, typing of expressions may now depend on equality assumptions and thus typing becomes undecidable.


\begin{example}[Undecidability of Typing]\label{exam:congBaseUndecidable}
Continuing Ex.~\ref{ex:cat}, consider terms $\vdash f:\mor\ u\ v$ and $\vdash g:\mor\ v'\ w$ for terms $\vdash u,v,v',w:\obj$.
Then $\vdash g\circ f:\mor\ u\ w$ holds iff $\vdash f:\mor\ u\ v'$, which holds iff $\vdash v\termEquals{\obj} v'$.
Depending on the axioms present, this may be arbitrarily difficult to prove.
\end{example}

Finally, we modify the rule for the non-emptiness of types: we allow the existence of empty dependent types and only require that for each HOL type in the image of the translation there exists one non-empty DHOL type translated to it (rather than requiring all dependent types translated to it to be non-empty). 
And we replace the typing rule for implication with the dependent one. The proof rules for implications are unchanged.
\[	
\rnamedRule{$\impl$type'}{\lowlight{\Gamma\dedT F\Rightarrow G:\bool}}{\lowlight{\Gamma\dedT F:\bool}\tb \concatCtx{\lowlight{\Gamma}}{\namedass{x}{F}}\lowlight{\dedT G:\bool}}{implType'}
\]

\begin{example}[Dependent Implication]\label{ex:depimpl}
Continuing Ex.~\ref{ex:cat}, consider the formula
\[ x:\obj,\ y:\obj\vdash  x\termEquals{\obj}y \,\impl\, \ident{x}\termEquals{\mor\ x\ x}\ident{y}\;:\;\bool\]
which expresses that equal objects have equal identity morphisms.
It is easy to prove. But it is only well-typed because the typing rule for dependent implication allows using $x\termEquals{\obj}y$ while type-checking $\ident{x}\termEquals{\mor\ x\ x}\ident{y}:\bool$, which requires deriving $\ident{y}:\mor\ x\ x$ and thus $\mor\ y\ y\typeEquals\mor\ x\ x$.
\end{example}

All the usual connectives and quantifiers can be defined in any of the usual ways now.
However, the details matter for the dependent versions of the connectives.
In particular, we choose $F\wedge G:=\neg(F\impl \neg G)$ and $F\vee G:=\neg F\impl G$ in order to obtain the dependent versions of conjunction and disjunction, in which the well-formedness of $G$ may depend on the truth or falsity of $F$, respectively.

\subsection{Translation}

We define a translation function $X\mapsto \PhiAppl{X}$ that maps any DHOL-syntax $X$ to HOL-syntax.
Its intuition is to erase type dependencies by translating all types $a\,t_1\,\ldots,\,t_n$ to $a$ and replacing every $\Pi$ with $\to$.
To recover the information of the erased dependencies, we additionally define a partial equivalence relation (PER) $\PredPhiName{A}$ on $\PhiAppl{A}$ for every DHOL-type $A$.

In general, a PER $r$ on type $U$ is a symmetric and transitive relation on $U$.
This is equivalent to $r$ being an equivalence relation on a subtype of $U$.
The intuitive meaning of our translation is that the DHOL-type $A$ corresponds in HOL to the quotient of the appropriate subtype of $\PhiAppl{A}$ by the equivalence $\PredPhiName{A}$.
In particular, the predicate $\PredPhiName{A}\ t\ t$ captures whether $t$ represents a term of type $A$.
More formally, the correspondence is:

\begin{center}
\begin{tabular}{|l|l|}
\hline
DHOL & \hol{} \\
\hline
type $A$ & type $\PhiAppl{A}$ and PER $\PredPhiName{A}:\PhiAppl{A}\to\PhiAppl{A}\to\bool$ \\
term $t:A$ & term $\PhiAppl{t}:\PhiAppl{A}$ satisfying $\PredPhi{A}{\PhiAppl{t}}$\\
\hline
\end{tabular}
\end{center}


\begin{definition}[Translation]\label{def:trans}
We translate DHOL-syntax by induction on the grammar.
Theories and contexts are translated declaration-wise:
\[\PhiAppl{\emptyThy} \;:=\;\emptyThy \tb
  \PhiAppl{\concatThy{T}{D}}\;:=\;\concatThy{\PhiAppl{T}}{\PhiAppl{D}}\tb
  \PhiAppl{\emptyCtx}\;:=\;\emptyCtx \plabel{PTemptyCtx}\tb
  \PhiAppl{\concatCtx{T}{D}}\;:=\;\concatThy{\PhiAppl{T}}{\PhiAppl{D}}
\]
where $\PhiAppl{D}$ is a list of declarations.

The translation $\PhiAppl{\;a: \piType{x_1}{A_1}\ldots\piType{x_n}{A_n}\type}$ of a base type declaration is given by
\[a: \type, \tb \PredPhiName{a}: \PhiAppl{A_1}\to\ldots\to \PhiAppl{A_n} \to a\to a\to \bool\]
\[\namedax{a_{PER}}{} \univQuant{x_1}{\PhiAppl{A_1}}\ldots\univQuant{x_n}{\PhiAppl{A_n}}\univQuant{u,v}{a} \PredPhiName{a} \ x_1\ \ldots\ x_n\ u\ v\impl u\termEquals{a}v\]
Thus, $a$ is translated to a base type of the same name without arguments and a trivial PER for every argument tuple.
Intuitively, $\PredPhiName{a} \ \PhiAppl{t_1}\ \ldots\ \PhiAppl{t_n}\ u\ u$ defines the subtype of the HOL-type $a$ corresponding to the DHOL-type $a\ t_1\ \ldots\ t_n$.


Constant and variable declarations are translated by adding the assumptions that they are in the PER of their type, and axioms and assumptions are translated straightforwardly:
\[\PhiAppl{c:A} \;:=\; c:\PhiAppl{A},\;\namedax{\typingAxName{c}}{\PredPhi{A}{c}} \tb
  \PhiAppl{x:A} \;:=\; x:\PhiAppl{A},\;\namedass{\typingAssName{x}}{\PredPhi{A}{x}}\]
\[\PhiAppl{\namedax{c}{F}} \;:=\; \namedass{c}{\PhiAppl{F}} \tb
  \PhiAppl{\namedass{x}{F}}\;:=\; \namedass{x}{\PhiAppl{F}}
\]

The cases of $\PhiAppl A$ and $\PredPhiName{A}$ for types $A$ are:
	\[\PhiAppl{a\ t_1\ \ldots \ t_n} \;:=\; a \plabel{PTTpAppl} \tb\tb 
	\termEqT{(a\ t_1\ \ldots \ t_n)}{s}{t} \;:=\; \PredPhiName{a}\ \PhiAppl{t_1}\ \ldots\ \PhiAppl{t_n}\ s\ t\]
	\[\PhiAppl{\piType{x}{A}B} \;:=\; \PhiAppl{A} \to \PhiAppl{B}\tb\tb
	\termEqT{(\piType{x}{A}B)}{f}{g} \;:=\; \univQuant{x,y}{\PhiAppl{A}}
	\termEqT{A}{x}{y}\impl \termEqT{B}{\left(f\ x\right)}{\left(g\ y\right)}\]
	\[\PhiAppl{\bool} \;:=\; \bool \tb\tb
	\termEqT{\bool}{s}{t}\;:=\; s\termEqB t\]
	\ifbool{inAppendix}{
	\[\PhiAppl{\subtype{A}{p}} \;:=\; \PhiAppl{A} \tb\tb
	\termEqT{\left(\subtype{A}{p}\right)}{s}{t} \;:=\; \termEqT{A}{s}{t}\land \PhiAppl{p}\ s\land \PhiAppl{p}\ t\]}

Finally, the cases for terms are straightforward except for, crucially, translating equality to the respective PER:
	\[\PhiAppl{c} \;:=\; c \tb 
	\PhiAppl{x} \;:=\; x \tb
	\PhiAppl{\lambdaFun{x}{A} t} \;:=\; \lambdaFun{x}{\PhiAppl{A}} \PhiAppl{t} \tb
	\PhiAppl{f\ t} \;:=\; \PhiAppl{f}\ \PhiAppl{t}\]
	\[\PhiAppl{F\Rightarrow G} \;:=\; \PhiAppl{F} \impl \PhiAppl{G} \tb 
	\PhiAppl{s\termEquals{A}t} \;:=\; \termEqT{A}{\PhiAppl{s}}{\PhiAppl{t}}\]
\end{definition}

\begin{example}[Translating Derived Connectives]\label{ex:transder}
If we define $\T$, $\F$, $\neg$ as usual in HOL and use the definition for dependent conjunction from above, it is straightforward to show that all DHOL-connectives are translated to their HOL-counterparts.
For example, we have (up to logical equivalence in HOL) that $\PhiAppl{F\wedge G}=\PhiAppl{F}\wedge\PhiAppl{G}$.

We also define the quantifiers in the usual way, e.g., using $\forall x:A.F(x)\;:=\;\lambdaFun{x}{A}F(x)\termEquals{A\to\bool}\lambdaFun{x}{A}\T$.
Then applying our translation yields
\[\PhiAppl{\forall x:A.F(x)}=\PredPhiName{(A\to\bool)}\ \PhiAppl{\lambda x:A.F(x)}\ \PhiAppl{\lambda x:A.\T}\]
\[=\forall x,y:\PhiAppl{A}.\PredPhiName{A}\ x\ y\impl \PredPhiName{\bool}\ F(x)\ \T\]
This looks clunky, but (because $\PredPhiName{A}$ is a PER as shown in Theorem~\ref{thm:completePaper}) is equivalent to
$\forall x:\PhiAppl{A}.\PredPhiName{A}\ x\ x\impl F(x)$.
Thus, DHOL-$\forall$ is translated to HOL-$\forall$ relativized using $\PredPhiName{A}\ x\ x$.
The corresponding rule $\PhiAppl{\exists x:A.F(x)}=\exists x:\PhiAppl{A}.\PredPhiName{A}\ x\ x\wedge F(x)$ can be shown accordingly.
\end{example}

\begin{example}[Categories in HOL]
We give a fragment of the translation of Ex.~\ref{ex:cat}:
\[\begin{array}{l@{\tb}l}
  \obj: \type & \PredPhiName\obj: \obj\to\obj\to\bool \\
	\mor: \type &\PredPhiName\mor: \obj\to\obj\to\mor\to\mor\to\bool \\
	\id: \obj\to\mor & \PredPhiName\id: \forall x,y:\obj.\PredPhiName\obj\ x\ y\impl\PredPhiName\mor\ x\ x\ (\id\ x)\ (\id\ y)\\
	\multicolumn{2}{l}{\comp: \obj\to\obj\to\obj\to\mor\to\mor\to\mor}\\
  \multicolumn{2}{l}{\neutl : \forall x:\obj.\PredPhiName\obj\ x\ x\impl \forall y:\obj.\PredPhiName\obj\ y\ y\impl}\\
  \multicolumn{2}{r}{\forall m:\mor.\PredPhiName\mor\ x\ y\ m\ m\impl \PredPhiName\mor\ x\ y\ (\comp\ x\ x\ y \ (\id\ x)\ m)\  m}
\end{array}\]
Here, for brevity, we have omitted $\obj_{PER}$, $\mor_{PER}$, and $\PredPhiName\comp$ and have already used the translation rule for $\forall$ from Ex.~\ref{ex:transder}.
The result is structurally close to what a native formalization of categories in HOL would look like, but somewhat clunkier.
\end{example}

\section{Predicate Subtypes}\label{sec:dphol}
\begin{figure}[htb]\small
Typing rules for predicate subtypes:
\[
	\rnamedRule{$\subtype{}{p}\type$}{\Gamma\dedT \subtype{A}{p}\; \type}{\Gamma\dedT p:\piType{x}{A}\bool}{psubType}\tb
	\rnamedRule{$\subtype{}{p}$I}{\Gamma\dedT t:\subtype{A}{p}}{\Gamma\dedT t:A \tb  \Gamma\dedT p\ t}{psubI}\tb
	\rnamedRule{$\subtype{}{p}$E}{\Gamma\dedT p\ t}{\Gamma\dedT t:\subtype{A}{p}}{psubE}
\]
Congruence and variance rule for predicate subtypes:
	\[
	\rnamedRule{$\subtype{}{p}\!\!\typeEquals$}{\Gamma\dedT \subtype{A}{p} \typeEquals \subtype{A'}{p'}}{\Gamma\dedT A\typeEquals A' \tb \Gamma\dedT p\termEquals{\piType{x}{A}\bool} p'}{psubEq}\tb
	\rnamedRule{$\subtyping\!\! \subtype{}{p}$}{\Gamma\dedT \subtype{A}{p}\subtyping \subtype{A'}{p'}}{\Gamma\dedT A\subtyping A'\tb \concatCtx{\Gamma}{x:A}\dedT p\ x\impl p'\ x}{subtPsubCong}
	\]
Rules that relate $A$ and $\subtype{A}{p}$:
	\[
	\rnamedRule{$\subtyping\!\! $Top}{\Gamma\dedT \subtype{A}{p}\subtyping A'}{\Gamma\dedT A\subtyping A'}{subtPsub}\tb
	\rnamedRule{$\subtype{}{triv}$}{\Gamma\dedT A\typeEquals \subtype{A}{\lambdaFun{x}{A}\T}}{\Gamma\dedT\Type{A}}{psubTriv}\tb
	\rnamedRule{$\subtype{}{triv}'$}{\Gamma\dedT \subtype{A}{\lambdaFun{x}{A}\T}\typeEquals A}{\Gamma\dedT\Type{A}}{psubTriv'}
	\]
Variance rules for other DHOL types:
	\[
	\rnamedRule{$\subtyping\!\!$I}{\Gamma\dedT A\subtyping A'}{\Gamma\dedT A\typeEquals A'}{subtI}\tb
	\rnamedRule{$\subtyping\!\!$Pi}{\Gamma\dedT \piType{x}{A}B\subtyping \piType{x}{A'}B'}{\Gamma\dedT A'\subtyping A\tb \concatCtx{\Gamma}{x:A'}\dedT B\subtyping B'}{subtPi}		
	\]
Rules for normalizing certain subtypes:
	\[
	\rnamedRule{$\pi \subtype{}{p}$Cod}{\Gamma\dedT\piType{x}{A}(\subtype{B}{p})\typeEquals \subtype{(\piType{x}{A}B)}{\lambda f.\univQuant{x}{A}p\ \left(f\ x\right)}}{\Gamma\dedT \Type{A}\tb \concatCtx{\Gamma}{x:A}\dedT \Type{B}\tb \concatCtx{\Gamma}{x:A}\dedT p:\piType{y}{B}\bool}{piPsubCod}\]
	\[
	\rnamedRule{$\subtype{}{p}\subtype{}{q}$}{\Gamma\dedT\subtype{\subtype{A}{p}}{q}\typeEquals \subtype{A}{\lambdaFun{x}{A}p\ x\land q\ x}}{\Gamma\dedT \Type{A}\tb \Gamma\dedT p:\piType{x}{A}\bool\tb \Gamma\dedT q:\piType{x}{(\subtype{A}{p})}\bool}{psubQsub}
	\]
	\caption{Additional Rules for Predicate Subtypes}\label{fig:dpholrules}
\end{figure}

To add predicate subtypes, we extend the \textbf{grammar} with the production $A	\bbc \subtype{A}{F}$.
No new productions for terms are needed because the inhabitants of $\subtype{A}{F}$ use the same syntax as those of $A$.

\newcommand{\invPred}{\symbolFont{inv?}}
\newcommand{\iso}{\symbolFont{isomorphisms}}
\begin{example}[Isomorphisms]\label{ex:cat2}
	We continue Example~\ref{ex:cat} and use predicate subtypes to write the type $\iso\ u$ of automorphisms on $u$ as a subtype of $\mor\ u\ u$.
	We can define $\iso\ u \;:=\;\subtype{(\mor\ u\ u)}{p}$ where the predicate $p$ is given by
	\[\lambda m:\mor\ u\ u.\,\exists i:\mor\ u\ u.\,\left(i\circ m\termEquals{\mor\ u\ u}\ident{u}\right) \land \left(m\circ i\termEquals{\mor\ u\ u}\ident{u}\right)\]
\end{example}

Adding subtyping requires a few extensions to our type system.
First we add a \textbf{judgment} $\Gamma\dedT A \subtyping B$ and replace the lookup rules for variables and constants with their subtyping-aware variants:
\[
\rnamedRule{const'}{\lowlight{\Gamma\dedT c:A}}{\lowlight{c:A'\thyIn{}{T}}\tb \lowlight{\Gamma\dedT A'}\subtyping \lowlight{A}}{const''}\tb
\rnamedRule{var'}{\lowlight{\Gamma\dedT x:A}}{\lowlight{x:A'\ctxIn{}{\Gamma}}\tb \lowlight{\Gamma\dedT A'}\subtyping \lowlight{A}}{var''}
\]
Then we add the \textbf{rules} given in Figure~\ref{fig:dpholrules}.
These induce an algorithm for deciding subtyping relative to an oracle for the undecidable validity judgment.
The latter enters the algorithm when two predicate subtypes are compared.
Note that the type-equality rule for $\subtype{\subtype{A}{p}}{q}$ uses a dependent conjunction.

The resulting system is a conservative extension of the variants of HOL and DHOL without subtyping: we recover these systems as the fragments that do not use $\subtype{A}{p}$.
In particular, in that case $A \subtyping B$ is trivial and holds iff $A\typeEquals B$ holds.

Finally, we extend our \textbf{translation} by adding the cases for predicate subtypes:

\begin{definition}[Translation]
	We extend Def.~\ref{def:trans} with
	\[\PhiAppl{\subtype{A}{p}} \;:=\; \PhiAppl{A} \tb 
	\termEqT{\left(\subtype{A}{p}\right)}{s}{t} \;:=\; \termEqT{A}{s}{t}\land \PhiAppl{p}\ s\land \PhiAppl{p}\ t\]
\end{definition}

\section{Soundness and Completeness}\label{sec:meta}

Now we establish that our translation is faithful, i.e. sound and complete. 
We will use the terms \emph{sound} and \emph{complete} from the perspective of using a HOL-ATP for theorem proving in DHOL, e.g., \emph{sound} means if $\PhiAppl{F}$ is a HOL-theorem, then $F$ is a DHOL-theorem, and \emph{complete} is the dual.\footnote{If, however, we think of our translation as an interpretation function that maps syntax to semantics, we could also justify swapping the names of the theorems.}

The completeness theorem states that our translation preserves all DHOL-judgments.
Moreover, the theorem statement clarifies the intuition behind the translations invariants:

\renewcommand{\dedPT}{\ensuremath{\vdash_{\PhiAppl{T}}}}
\begin{theorem}[Completeness]\label{thm:completePaper}
	We have \\	
	\setlength\extrarowheight{5pt}
	\begin{center}
	\begin{tabular}{|l@{\;\;}|@{\;\;}l@{\tb and \tb}l|}
		\hline 
		if in DHOL & \multicolumn{2}{c|}{then in HOL} \\ 
		\hline 
		$\phantom{\Gamma}\ded \Thy{T}$ & \multicolumn{2}{l|}{$\phantom{\PhiAppl{\Gamma}}\ded\Thy{\PhiAppl T}$}  \\ 
		$\phantom{\Gamma}\dedT \Ctx{\Gamma}$ &  \multicolumn{2}{l|}{$\phantom{\PhiAppl{\Gamma}}\dedPT\Ctx{\PhiAppl{\Gamma}}$}  \\ 
		$\Gamma\dedT \Type{A}$ & $\PhiAppl{\Gamma}\dedPT\Type{\PhiAppl{A}}$ & $\PhiAppl{\Gamma}\dedPT\PredPhiName{A}: \PhiAppl{A}\to\PhiAppl{A}\to \bool$ and $\PredPhiName{A}$ is PER \\ 
		$\Gamma\dedT A \typeEquals B$ & $\PhiAppl{\Gamma}\dedPT\PhiAppl{A} \typeEquals \PhiAppl{B}$ & $\concatCtx{\PhiAppl{\Gamma}}{x,y:\PhiAppl{A}}\dedPT\PredPhiName{A}\ x\ y \termEqB \PredPhiName{B}\ x\ y$ \\ 
		$\Gamma\dedT A \subtyping B$ & $\PhiAppl{\Gamma}\dedPT \PhiAppl{A} \typeEquals \PhiAppl{B}$ & $\concatCtx{\PhiAppl{\Gamma}}{x,y:\PhiAppl{A}}\dedPT \termEqT{A}{x}{y}\impl \termEqT{B}{x}{y}$ \\ 
		$\Gamma\dedT t:A$ & $\PhiAppl{\Gamma}\dedPT\PhiAppl{t}:\PhiAppl{A}$ & $\PhiAppl{\Gamma}\dedPT\PredPhi{A}{\PhiAppl{t}}$ \\ 
		$\Gamma\dedT F$ & \multicolumn{2}{l|}{$\PhiAppl{\Gamma}\dedPT\PhiAppl{F}$} \\
		\hline 
	\end{tabular}
	\end{center}
	
	Additionally the substitution lemma holds, i.e.,
	\begin{align*}
	\concatCtx{\Gamma}{x:A}\dedT t:B\;\Mand\;\Gamma\ded u:A &\Impl \PhiAppl{\Gamma}\dedPT\PhiAppl{\subst{t}{x}{u}}\termEquals{\PhiAppl{B}} \subst{\PhiAppl{t}}{x}{\PhiAppl{u}}
	\end{align*}
\end{theorem}
\begin{proof}
The proof proceeds by induction on derivations. 
It is given in Appendix~\ref{sec:complete}.
\end{proof}

The reverse direction is much trickier.
To understand why, we look at two canaries in the coal mine that we have used to reject multiple intuitive but untrue conjectures:

\begin{example}[Non-Injectivity of the Translation]\label{ex:noninj}
Continuing Ex.~\ref{ex:cat}, assume terms $u,v:\obj$ and consider the identify functions $I_u:=\lambda f:\mor\ u\ u.f$ and $I_v:=\lambda f:\mor\ v \ v.f$.
Both are translated to the same HOL-term $\PhiAppl{I_u}=\PhiAppl{I_v}=\lambda f:\mor.f$ (because $I_u$ and $I_v$ only differ in the types, which are erased by our translation).

Consequently, the ill-typed DHOL-Boolean $b:=I_u\termEquals{\mor\ u\ u\to \mor\ u\ u}I_v$ is translated to the HOL-Boolean $\lambda f:\mor.f\termEquals{\mor\to\mor}\lambda f:\mor.f$, which is not only well-typed but even a theorem.
\end{example}

To better understand the underlying issue we introduce the notion of \emph{spurious} terms.
The well-typed translation $\PhiAppl{t}$ of a DHOL-term $t$ is called \textbf{spurious} if $t$ is ill-typed (otherwise it is called \emph{proper}). 
Intuitively, we should be able to use the PERs $\PredPhiName{A}$ to deal with spurious terms: to type-check $t:A$ in DHOL, we want to use $\PredPhiName{A}\ \PhiAppl{t}\ \PhiAppl{t}$ in HOL.
But even that is tricky:

\begin{example}[Trivial PERs for Built-In Base Types]
Consider the property $\PredPhiName{\bool}\ x\ x$.
Our translation guarantees $\PredPhiName{\bool}\ \T\ \T$ and $\PredPhiName{\bool}\ \F\ \F$.
Thus, we can use Boolean extensionality to prove in HOL that $\forall x:\bool.\,\PredPhiName{\bool}\ x\ x$, making the property trivial.
In particular, we can prove $\PredPhiName{\bool}\ \PhiAppl{b}\ \PhiAppl{b}$ for the spurious Boolean $b$ from Ex.~\ref{ex:noninj}.
Even worse, the property $\PredPhiName{(\piType{x}{A}B)}\ x\ x$ is trivial in this way whenever it is for $B$ and thus for all $n$-ary $\bool$-valued function types.

More generally, this degeneration effect occurs for every base type that is built into both DHOL and HOL and that is translated to itself.
$\bool$ is the simplest example of that kind, and the only one in the setting described here.
But reasonable language extensions like built-in base types $a$ for numbers, strings, etc. would suffer from the same issue.
This is because all of these types would come with built-in induction principles that derive a universal property from its ground instances, at which point $\PredPhiName{a}\ x \ x$ becomes trivial.

Note, however, that the degeneration effect does \emph{not} occur for \emph{user-declared} base types.
For example, consider a theory that declares a base type $N$ for the natural numbers and an induction axiom for it.
$N$ would not be translated to itself but to a fresh HOL-type in whose induction axiom the quantifier $\forall$ is relativized by $\PredPhiName{N}\ x\ x$.
Consequently, $\PredPhiName{N}\ x\ x$ is not trivial and can be used to reject spurious terms.
\end{example}

Therefore, we cannot expect the reverse directions of the statements in Thm.~\ref{thm:completePaper} to hold in general.
However, we can show the following property that is sufficient to make our translation well-behaved:

\begin{theorem}[Soundness]\label{thm:soundPaper}
	Assume a well-formed DHOL-theory $\ded\Thy{T}$.
	\[\text{If}\;\;\Gamma\dedT F:\bool \;\;\text{and}\;\; \PhiAppl{\Gamma}\dedPT\PhiAppl{F}, \;\;\text{then}\;\; \Gamma\dedT F\]
	In particular, if $\Gamma\dedT s:A$ and $\Gamma\dedT t:A$ and $\PhiAppl{\Gamma}\dedPT \termEqT{A}{\PhiAppl{s}}{\PhiAppl{t}}$, then $\Gamma\ded s \termEquals{A}t$.
\end{theorem}
\begin{proof}
The key idea is to transform a HOL-proof of $\PhiAppl{F}$ into one that is in the image of the translation, at which point we can read off a DHOL-proof of $F$.
The full proof is given in Appendix~\ref{sec:soundness}.
\end{proof}

Intuitively, the reverse directions of Thm.~\ref{thm:completePaper} hold if we have already established that all involved expressions are well-typed in DHOL.
Thus, we \emph{can} use a HOL-ATP to prove DHOL-conjectures if we validate independently that the conjecture is well-typed to begin with.
In the remainder of this section, we develop the necessary type-checking algorithm for DHOL.

\paragraph{Type-Checking}
Inspecting the rules of DHOL, we observe that all DHOL-judgments would be decidable if we had an oracle for the validity judgment $\Gamma\dedT F$.
Indeed, our DHOL-rules are already written in a way that essentially allows reading off a bidirectional type-checking algorithm.
It only remains to split the typing judgment $\Gamma\dedT t:A$ into two algorithms for type-inference (which computes $A$ from $t$) and type-checking (which takes $t$ and $A$ and returns yes or no) and to aggregate the rules for subtyping into an appropriate pattern-match.

The construction is routine, and we have implemented the resulting algorithm in our MMT/LF logical framework \cite{rabe:recon:17,lf}.%
\footnote{The formalization of DHOL in \mmt is available at \url{https://gl.mathhub.info/MMT/LATIN2/-/blob/devel/source/logic/hol_like/dhol.mmt}.
The example theories given throughout this paper and a few example conjectures are available at \url{https://gl.mathhub.info/MMT/LATIN2/-/blob/devel/source/casestudies/2023-cade}.}
The oracle for the validity judgment is provided by our translation and a theorem prover for HOL (see Sect.~\ref{sec:impl}).
It remains to show that whenever the algorithm calls the oracle for $\Gamma\dedT F$, we do in fact have that $\Gamma\dedT F:\bool$ so that Thm.~\ref{thm:soundPaper} is applicable.
Formally, we show the following:
\begin{theorem}
Relative to an oracle for $\Gamma\dedT F$, consider a derivation of some DHOL-judgment, in which the children of each node\ednote{CR@FR: I think, you mean node in the proof tree. This is not directly clear to the reader. Can we phrase this better to clarify it?} are ordered according to the left-to-right order of the assumptions in the statement of the applied rule.

If the oracle calls are made in depth-first order, then each such call satisfies $\Gamma\dedT F:\bool$.
\end{theorem}
\begin{proof}
We actually prove, by induction on derivations, the more general statement requires that each rule preserves the following preconditions:\\[0.2cm]
\begin{tabular}{|l|l|}
\hline
Judgment & Precondition \\
\hline
 $\dedT \Ctx{\Gamma}$ & $\ded \Thy{T}$\\
 $\Gamma\dedT A\, \type$ & $\dedT \Ctx{\Gamma}$ \\ 
 $\Gamma\dedT t:A$ & $\Gamma\dedT A\,\type$ (post-condition when used as type-\emph{inference})\\
 $\Gamma\dedT F$ & $\Gamma\dedT F:\bool$ \\ 
 $\Gamma\dedT A\typeEquals B$ or $\Gamma\dedT A\subtyping B$&  $\Gamma\dedT A\,\type$ and  $\Gamma\dedT B\,\type$ \\ 
\hline
\end{tabular}\\[.2cm]
Note that rules whose conclusion is a validity judgment can be ignored because they are replaced by the oracle anyway.

The most interesting case is the rule for $\Gamma\dedT a\ s_1\ \ldots\ s_n\typeEquals a\ t_1\ \ldots t_n$.
Here, the left-to-right order of assumptions is critical because $\Gamma\dedT s_1\termEquals{A_1}t_1$ may be needed to show, e.g., $\Gamma\dedT s_2\termEquals{\subst{A_2}{x_1}{t_1}}t_2:\bool$.
\end{proof}

\section{Theorem Prover Implementation}\label{sec:impl}
We have integrated our translation as a preprocessor to the HOL ATP LEO-III \cite{leo3}.
We chose this ATP because its existing preprocessor infrastructure already includes a powerful logic embedding tool \cite{leo_embedding,S22}.%
\footnote{Our implementation can be found at \url{https://github.com/leoprover/logic-embedding/blob/master/embedding-runtime/src/main/scala/leo/modules/embeddings/DHOLEmbedding.scala}.}
However, with a little more effort, other HOL ATPs work as well.

Furthermore, we developed a bridge between the \mmt logical framework \cite{rabe:recon:17} and LEO-III (both of which are written in the same programming language).%
\footnote{Our implementation can be found at \url{https://gl.mathhub.info/MMT/LATIN2/-/tree/devel/scala/latin2/tptp}.}
This allows us to use our \mmt-based type-checker for DHOL with our Leo-III-based theorem prover to obtain a full-fledge implementation of DHOL.
Moreover, this system can immediately use \mmt's logic-independent frontend features like IDE and module system.

Alternatively, we can use LEO-III as a general purpose DHOL-ATP that accepts input in TPTP.
Even though TPTP does not officially sanction DHOL as a logic, it anticipates dependent function types and already provides syntax for them (although --- to our knowledge --- no ATP system has made use of it so far).
Concretely, TPTP represents the type $\piType{x}{A}B$ as \verb|!>[X:A]:B| and a base type $a\ t_1\ldots\ t_n$ as \verb|a @ t1 ... @ tn|.
TPTP does not yet provide syntax for predicate subtypes, i.e., this approach is currently limited to the no-subtyping fragment of DHOL.
But extending the TPTP syntax with predicate subtypes would be straightforward, e.g., by using \verb!A ?| p! to represent the type $\subtype{A}{p}$.


The encoding of the conjecture given in Ex.~\ref{ex:depimpl} using the theory from Ex.~\ref{ex:cat} is given at \url{https://gl.mathhub.info/MMT/LATIN2/-/blob/devel/source/casestudies/2023-cade/Category Theory/category-theory-lemmas-dhol.p} (which also includes further example conjectures relative to the same theory). 
Running the logic embedding tool translates it into the TPTP TH0 problem given at \url{https://gl.mathhub.info/MMT/LATIN2/-/blob/devel/source/casestudies/2023-cade/Category Theory/category-theory-lemmas-hol.p}.
Unsurprisingly, LEO-III can prove this simple theorem easily.

\subsection{Practical evaluation of the theorem prover implementation}
In order to evaluate the practical usefulness of the translation we studied various example conjectures about function composition on sets and category theory. 
We considered 5 further lemmas based on the theory in Ex.~\ref{ex:cat} which are written directly in TPTP and can all be proven by E, Vampire and cvc5. 
We also studies various harder lemmas about function compositions and category theory. Those examples are written in \mmt and take advantage of abbreviation and user-defined notations which often use implicit arguments (inferred by the prover), significantly improving their readability. 

The examples can be found at  \url{https://gl.mathhub.info/MMT/LATIN2/-/blob/devel/source/casestudies/2023-cade}. The \mmt prover successfully type-checks all problems and translates them into TPTP problems to be solved by HOL ATPs. 

Since LEO-III can solve none of the 6 function composition examples, we also tested other HOL ATPs on the generated TPTP problems.  
Running all HOL ATP provers supported at \url{https://www.tptp.org/cgi-bin/SystemOnTPTP} on the function composition problems yields that many provers can solve 3 of the problems, Vampire can solve 4 of them and 5 out of the 6 conjectures can be solved by at least one HOL ATP.

We also studied 6 lemmas about limits in category theory including the uniqueness, commutativity and associativity of limits.
To better evaluate the usefulness of the translation we also formalized these lemmas in native HOL (in \mmt). 
The DHOL formalization is significantly more readable and benefits from the more expressive type system that can help spot mistakes in the formalization.

Running the HOL ATPs from \url{https://www.tptp.org/cgi-bin/SystemOnTPTP} on the generated TPTP problems (with 60 second timeout) yields the results summarized in below table (we omit results for Isabelle, Isabelle-HOT, Nitpick, TPS and Vampire-FMO which proved none of the lemmas in their DHOL nor their native HOL formalizations): 

\begin{figure}
	\begin{tabular}{|c|cc|cc|cc|}
		\hline 
		\multirow{2}{*}{HOL ATP} & \multicolumn{2}{c|}{lemma 1 proven} & \multicolumn{2}{c|}{lemma 2 proven} & \multicolumn{2}{c|}{lemma 3 proven} \\
		\cline{2-7}
		 & DHOL & native HOL & DHOL & native HOL & DHOL & native HOL \\ 
		\hline 
		agsyHOL & yes & no & no & no & yes & no \\ 
		\hline 
		cocATP & yes & no & no & no & no & no \\ 
		\hline 
		cvc5 & yes & yes & no & no & yes & no \\ 
		\hline 
		cvc5-SAT & yes & no & no & no & no & no \\ 
		\hline 
		E & yes & yes & no & no & no & yes \\ 
		\hline 
		HOLyHammer & yes & yes & no & no & yes & yes \\ 
		\hline 
		Lash & yes & yes & no & no & no & no \\ 
		\hline 
		LEO-II & yes & no & no & no & no & no \\ 
		\hline 
		Leo-III & yes & yes & no & no & no & no \\ 
		\hline 
		Leo-III-SAT & yes & yes & no & no & no & no \\
		\hline 
		Satallax & yes & yes & no & no & yes & no \\ 
		\hline 
		Vampire & yes & yes & no & no & no & yes \\  
		\hline 
		Zipperpin & yes & yes & no & no & yes & yes \\ 
		\hline 
		total & 13 & 9 & 0 & 0 & 5 & 4 \\ 
		\hline 
	\end{tabular}
	\ \\\ \\\ \\\begin{tabular}{|c|cc|cc|cc|}
		\hline 
		\multirow{2}{*}{HOL ATP} & \multicolumn{2}{c|}{lemma 4 proven} & \multicolumn{2}{c|}{lemma 5 proven} & \multicolumn{2}{c|}{lemma 6 proven} \\
		\cline{2-7}
		& DHOL & native HOL & DHOL & native HOL & DHOL & native HOL \\ 
		\hline 
		agsyHOL & no & no & no & no & no & no \\ 
		\hline 
		cocATP & no & no & no & no & no & no \\ 
		\hline 
		cvc5 & no & yes & no & no & no & no \\ 
		\hline 
		cvc5-SAT & no & no & no & no & no & no \\ 
		\hline 
		E & no & yes & no & yes & no & yes \\ 
		\hline 
		HOLyHammer & no & yes & no & no & no & yes \\ 
		\hline 
		Lash & no & no & no & no & no & no \\ 
		\hline 
		LEO-II & no & no & no & no & no & no \\ 
		\hline 
		Leo-III & no & no & no & no & no & no \\ 
		\hline 
		Leo-III-SAT & no & no & no & no & no & no \\
		\hline 
		Satallax & no & no & yes & no & no & no \\ 
		\hline 
		Vampire & no & yes & no & yes & no & yes \\  
		\hline 
		Zipperpin & no & yes & yes & yes & no & yes \\ 
		\hline 
		total & 0 & 5 & 2 & 3 & 0 & 4 \\ 
		\hline 
	\end{tabular}
\end{figure}
Overall more problems generated from the native HOL formalization can be solved by some HOL ATP (5/6 compared to 3/6 for the DHOL formalization). 
But analyzing the test results for individual provers, reveals that in 8 cases a DHOL conjecture but not it's native HOL version can be proven by a HOL ATP and in 13 cases the converse, indicating that both DHOL and native HOL formalizations have different advantages. Overall the HOL ATPs found 25 successful proofs for the native HOL problems and 20 for the DHOL problems.
This suggests that current HOL ATPs can prove native HOL problems somewhat better than their translated DHOL counterparts, but not much better. 

Furthermore, our translation has so far been engineered for generality and soundness/completeness and not for ATP efficiency.
Indeed, future work has multiple options to boost the ATP performance on translated DHOL, e.g., by
\begin{itemize}
	\item developing sufficient criteria for when simpler HOL theories can be produced
	\item inserting lemmas into the translated theories that guide proof search in ATPs, e.g., to speed up equality reasoning
	\item adding definitions to translated DHOL problems and developing better criteria when to expand them
\end{itemize}

Thus, we consider the test results to be very promising.
In particular, the translation could serve as a useful basis for type-checkers and hammer tools for DHOL ITPs.

\section{Conclusion and Future Work}\label{sec:conc}
We have combined two features of standard languages, higher-order logic HOL and dependent type theory DTT, thereby obtaining the new dependently-typed higher-order logic DHOL.
Contrary to HOL, DHOL allows for \emph{dependent} function types.
Contrary to DTT, DHOL retains the simplicity of classical Booleans and standard equality.

On the downside, we have to accept that DHOL, unlike both HOL and DTT, has an undecidable type system. 
Further work will show how big this disadvantage weighs in practical theorem proving applications.
But we anticipate that the drawback is manageable, especially if, as in our case, an implementation of DHOL is coupled tightly with a strong ATP system.
We accomplish this with a sound and complete translation from DHOL into HOL that enables using existing HOL ATPs to discharge the proof obligations that come up during type-checking.
We have implemented our novel translation as a TPTP-to-TPTP preprocessor for HOL ATP systems and outlined the implementation of a type-checker and hammer tool for DHOL based on the resulting prover.

Moreover, once this design is in place, it opens up the possibility to add certain type constructors to DHOL that are often requested by users but difficult to provide for system developers because they automatically make typing undecidable.
We have shown an extension of DHOL with predicate subtypes as an example.
Quotients, partial functions, or fixed-length lists are other examples that can be supported in future work.

We expect our translation remains sound and complete if DHOL is extended with other features underlying common HOL systems such as built-in types for numbers, the axiom of infinity, or the subtype definition principle.
How to extend DHOL with a choice operator remains a question for future work --- if solved, this would allow extending existing HOL ITPs to DHOL.

\bibliographystyle{splncs04}
\bibliography{biblio.bib}

\newpage
\begin{appendix}
	\renewcommand{\rulelabel}[2]{\rulelabelAppendix{#1}{#2}}
\renewcommand{\namedRule}[3]{\rul[\text{#1}]{#3}{#2}}
\renewcommand{\rnamedRule}[4]{\rul[\text{#1}]{#3}{#2}\rulelabel{#4}{#1}}
\renewcommand{\snamedRule}[3]{\rul[\text{#1}]{#3}{#2}\rulelabel{#1}{#1}}
\setbool{inAppendix}{true}
\newpage

\section{Summary of logics and translations}\label{sec:logics}
In this section we collect the inference rules of the logics and the definition of the overall translation.
We name the rules and enumerate the cases in the definition of the translation for reference in the proofs in the subsequent appendices.
We use \dhol{} and \dphol{} to refer to the variants without and with predicate subtypes.

\subsection{\hol{} rules}

\subsection{Modified rules for \dhol{} without predicate subtypes}
\begin{figure}[H]
	{\small
	\[
	\snamedRule{thyType'}{\lowlight{\ded}\Thy{\concatThy{\lowlight{T}}{\lowlight{a:}\piType{x_1}{A_1}\ldots\piType{x_n}{A_n}\lowlight{\type}}}}
	{\dedT \Ctx{x_1:A_1,\,\ldots,x_n:A_n}}
	\]
	\[
	\snamedRule{type'}{\lowlight{\Gamma \dedT a}\ t_1\ \ldots\ t_n\;\lowlight{\type}}
	{\lowlight{a:}\piType{x_1}{A_1}\ldots\piType{x_n}{A_n}\lowlight{\type \text{ in }T}\ \; \lowlight{\dedT \Ctx{\Gamma}}\ \;
	\Gamma \dedT t_1:A_1 \ \;\ldots\ \; \Gamma \dedT t_n:\subst{A_n}{x_1}{t_1}\ldots\substOp{x_{n-1}}{t_{n-1}}}
	\]
	\[	
	\snamedRule{pi}{\lowlight{\Gamma\dedT} \piType{x}{A} B\ \lowlight{\type}}{\lowlight{\Gamma\dedT \Type{A} \quad }\concatCtx{\Gamma}{x:A}\lowlight{\dedT \Type{B}}}\tb \rnamedRule{cong$\Pi$}{\lowlight{\Gamma\dedT}\piType{x}{A} B\lowlight{\typeEquals} \piType{x}{A'}B'}{\lowlight{\Gamma\dedT A\typeEquals A' \quad  \Gamma,x:A\dedT B\typeEquals B'}}{congPi}
	\]
	\[
	\snamedRule{congBase'}{\lowlight{\Gamma\dedT a}\ s_1\ \ldots\ s_n\lowlight{\typeEquals} \lowlight{a}\ t_1\ \ldots t_n}
	{\lowlight{a:}\piType{x_1}{A_1}\ldots\piType{x_n}{A_n}\lowlight{\type\text{ in }T}\ \;\dedT \Ctx{\Gamma}\ \;
	\Gamma \dedT s_1\termEquals{A_1}t_1 \ \ldots\ \Gamma \dedT s_n\termEquals{\subst{A_n}{x_1}{t_1}\ldots\substOp{x_{i-1}}{t_{i-1}}}t_n}
	\]
	\[
	\snamedRule{lambda'}{\lowlight{\Gamma\dedT (\lambdaFun{x}{A} t):} \piType{x}{A} B}{\lowlight{\concatCtx{\Gamma}{x:A}\dedT t:B}}\tb
	\snamedRule{appl'}{\lowlight{\Gamma\dedT f\,t:\subst{B}{x}{t}}}{\lowlight{\Gamma\dedT f:}\piType{x}{A} B \tb \Gamma\dedT t:A}\]
	\[	
	\rnamedRule{$\impl$type'}{\lowlight{\Gamma\dedT F\Rightarrow G:\bool}}{\lowlight{\Gamma\dedT F:\bool}\tb \concatCtx{\Gamma}{\namedass{ass}{F}}\lowlight{\dedT G:\bool}}{implType'}
	\]
	\[
	\snamedRule{congAppl'}{\lowlight{\Gamma\dedT f\ t\termEquals{B} f'\ t'}}{\lowlight{\Gamma\dedT t\termEquals{A} t'}\quad \lowlight{\Gamma\dedT f}\termEquals{\piType{x}{A} B} \lowlight{f'}}
	\quad	
	\rnamedRule{cong$\lambda$'}{\lowlight{\Gamma\dedT \lambdaFun{x}{A} t}\termEquals{\piType{x}{A}B} \lowlight{\lambdaFun{x}{A'} t'}}{\lowlight{\Gamma\dedT A\typeEquals A' \ \;\concatCtx{\Gamma}{x:A} \dedT t\termEquals{B} t'}}{congLam'}
	\]
	\[	
	\snamedRule{etaPi}{\lowlight{\Gamma\dedT t}\termEquals{\piType{x}{A} B} \lowlight{\lambdaFun{x}{A}t\ x}}{\lowlight{\Gamma\dedT t:}\piType{x}{A} B}
	\]
	}
\end{figure}

\subsection{Modified rules for \dphol{} with predicate subtypes}
	\[
	\rnamedRule{const'}{\lowlight{c:A}}{\lowlight{c:}A'\lowlight{\thyIn{}{T}}\tb \lowlight{\Gamma\dedT A'}\subtyping \lowlight{A}}{const''}\tb
	\rnamedRule{var'}{\lowlight{x:A}}{\lowlight{x:}A'\lowlight{\ctxIn{}{\Gamma}}\tb \lowlight{\Gamma\dedT A'}\subtyping \lowlight{A}}{var''}
	\]
	\[
	\rnamedRule{$\subtyping\!\!$Pi}{\Gamma\dedT \piType{x}{A}B\subtyping \piType{x}{A'}B'}{\Gamma\dedT A'\subtyping A\tb \concatCtx{\Gamma}{x:A'}\dedT B\subtyping B'}{subtPi}
	\]
	\[
	\rnamedRule{$\subtyping\!\! \subtype{}{p}$}{\Gamma\dedT \subtype{A}{p}\subtyping \subtype{A'}{p'}}{\Gamma\dedT A\subtyping A'\tb \concatCtx{\Gamma}{x:A}\dedT p\ x\impl p'\ x}{subtPsubCong}\tb
	\rnamedRule{$\subtyping\!\! $Top}{\Gamma\dedT \subtype{A}{p}\subtyping A'}{\Gamma\dedT A\subtyping A'}{subtPsub}\]
	\[
	\rnamedRule{$\subtyping\!\!$I}{\Gamma\dedT A\subtyping A'}{\Gamma\dedT A\typeEquals A'}{subtI}\tb
	\rnamedRule{$\subtype{}{triv}$}{\Gamma\dedT A\typeEquals \subtype{A}{\lambdaFun{x}{A}\T}}{\Gamma\dedT\Type{A}}{psubTriv}\tb
	\rnamedRule{$\subtype{}{triv}'$}{\Gamma\dedT \subtype{A}{\lambdaFun{x}{A}\T}\typeEquals A}{\Gamma\dedT\Type{A}}{psubTriv'}\]
	\[
	\rnamedRule{$\pi \subtype{}{p}$Cod}{\Gamma\dedT\piType{x}{A}\subtype{B}{p}\typeEquals \subtype{\piType{x}{A}B}{\lambdaFun{f}{\piType{x}{A}B}\univQuant{x}{A}p\ \left(f\ x\right)}}{\Gamma\dedT \Type{A}\tb \concatCtx{\Gamma}{x:A}\dedT \Type{B}\tb \concatCtx{\Gamma}{x:A}\dedT p:\piType{y}{B}\bool}{piPsubCod}\]
	\[
	\rnamedRule{$\subtype{}{p}\subtype{}{q}$}{\Gamma\dedT\subtype{\subtype{A}{p}}{q}\typeEquals \subtype{A}{\lambdaFun{x}{A}p\ x\land q\ x}}{\Gamma\dedT \Type{A}\tb \Gamma\dedT p:\piType{x}{A}\bool\tb \Gamma\dedT q:\piType{x}{A}\bool}{psubQsub}
	\]
	\[\begin{array}{l}
	\rnamedRule{$\subtype{}{p}\type$}{\Gamma\dedT \subtype{A}{p}\; \type}{\Gamma\dedT p:\piType{x}{A}\bool}{psubType}\tb
	\rnamedRule{$\subtype{}{p}\!\!\typeEquals$}{\Gamma\dedT \subtype{A}{p} \typeEquals \subtype{A'}{p'}}{\Gamma\dedT A\typeEquals A' \tb  \Gamma\dedT p\termEquals{\piType{x}{A}\bool} p'}{psubEq}
	\end{array}\]
	\[\begin{array}{l}
	\rnamedRule{$\subtype{}{p}$I}{\Gamma\dedT t:\subtype{A}{p}}{\Gamma\dedT t:A \tb  \Gamma\dedT p\ t}{psubI}\tb
	\rnamedRule{$\subtype{}{p}$E}{\Gamma\dedT p\ t}{\Gamma\dedT t:\subtype{A}{p}}{psubE}
	\end{array}\]

\subsection{The translation from \dphol{} into \hol{}}
\begin{definition}[Translation]
	We define a translation from \ifbool{inAppendix}{\dphol{}}{\dhol{}} to \hol{} syntax by induction on the Grammar.
	
	We use the notation $\overrightarrow{x:A}, \overrightarrow{\piType{x}{A}} and \overrightarrow{A}, \overrightarrow{x}$ 
	to denote $x:A_1, \ldots, x_n:A_n$, $\piType{x_1}{A_1}\ldots \piType{x_n}{A_n}$ and $A_1\to\ldots\to A_n$, $x_1\ \ldots\ x_n$ 
	respectively.
	
	The cases for theories and contexts are:
	{\small
		\begin{align*}
		\PhiAppl{\emptyThy}:=&\emptyThy \plabel{PTemptyThy}\\
		\PhiAppl{\concatThy{T}{D}}:=&\concatThy{\PhiAppl{T}}{\PhiAppl{D}}&&\text{where}\\
		\PhiAppl{\;a: \overrightarrow{\piType{x}{A}} \type} :=& a: \type,\\& \PredPhiName{a}: \overrightarrow{\PhiAppl{A}} \to a\to a\to \bool,\\
		& \namedax{a_{trans}}{}\forall \overrightarrow{x\!:\!\PhiAppl{A}}.~ \univQuant{u,v,w}{a}\termEqT{\left(a\ \overrightarrow{x}\right)}{u}{v}\impl \left(\termEqT{\left(a\ \overrightarrow{x}\right)}{v}{w}\impl\termEqT{\left(a\ \overrightarrow{x}\right)}{u}{w}\right),\\
		& \namedax{a_{sym}}{}\forall\overrightarrow{x\!:\!\PhiAppl{A}}.~ \univQuant{u,v}{a}\termEqT{\left(a\ \overrightarrow{x}\right)}{u}{v}\impl\termEqT{\left(a\ \overrightarrow{x}\right)}{v}{u},\\
		&\namedax{a_{PER}}{} \forall\overrightarrow{x\!:\!\PhiAppl{A}}.~ \univQuant{u,v}{a}\PredPhi{\left(a\ \overrightarrow{x}\right)}{v}\impl \termEqT{(a\ \overrightarrow{x})}{u}{v}\termEqB u\termEquals{a}v
		\plabel{PTTpConstr}\\
		\PhiAppl{c:A}:=&c:\PhiAppl{A},\quad \namedax{\typingAxName{c}}{\PredPhi{A}{c}} \plabel{PTtermDecl}\\
		\PhiAppl{\namedax{ax}{F}}:=&\namedass{ax}{\PhiAppl{F}} \plabel{PTax}\\[0.5cm]
		\PhiAppl{\emptyCtx}:=&\emptyCtx \plabel{PTemptyCtx}\\
		\PhiAppl{\Gamma,x:A}:=&\PhiAppl{\Gamma},\;x:\PhiAppl{A},\namedass{\typingAssName{x}}{\PredPhi{A}{x}} \plabel{PTctxVar}\\
		\PhiAppl{\Gamma,\namedass{ass}{F}}:=&\PhiAppl{\Gamma},\;\namedass{ass}{\PhiAppl{F}} \plabel{PTctxAss}
		\end{align*}
	}
	The case of $\PhiAppl A$ and $\termEqT{A}{s}{t}$ for types $A$ are:
	\begin{align*}
	\PhiAppl{(a\ t_1\ \ldots \ t_n)}&:=a \plabel{PTTpAppl}\\
	\termEqT{(a\ t_1\ \ldots \ t_n)}{s}{t}&:=\PredPhiName{a}\ \PhiAppl{t_1}\ \ldots\ \PhiAppl{t_n}\ s\ t \plabel{PTTpPredAppl}\\
	\PhiAppl{\piType{x}{A}B} &:= \PhiAppl{A} \to \PhiAppl{B} \plabel{PTPitype}\\
	\termEqT{(\piType{x}{A}B)}{f}{g} &:= \univQuant{x,y}{\PhiAppl{A}}
	\termEqT{A}{x}{y}\impl \termEqT{B}{\left(f\ x\right)}{\left(g\ y\right)} \plabel{PTPipred}\\
	\PhiAppl{\bool}&:=\bool \plabel{PTTpBool}\\
	\termEqT{\bool}{s}{t}&:=s\termEqB t \plabel{PTTpBoolPred}
	\ifbool{inAppendix}{\\	
	\PhiAppl{\subtype{A}{p}} &:= \PhiAppl{A}\plabel{PTPStype}\\
	\termEqT{\left(\subtype{A}{p}\right)}{s}{t} &:= \termEqT{A}{s}{t}\land \PhiAppl{p}\ s\land \PhiAppl{p}\ t\plabel{PTPSpred}}{}
	\end{align*}
	The cases for terms are:
	\begin{align*}
	\PhiAppl{c} &:= c \plabel{PTmConst}\\
	\PhiAppl{x} &:= x\plabel{PTmVar}\\
	\PhiAppl{\lambdaFun{x}{A} t} &:= \lambdaFun{x}{\PhiAppl{A}} \PhiAppl{t}\plabel{PTLam}\\
	\PhiAppl{f\ t} &:= \PhiAppl{f}\ \PhiAppl{t}\plabel{PTappl}\\
	\PhiAppl{F\Rightarrow G}&:=\PhiAppl{F} \impl \PhiAppl{G}\plabel{PTImpl}\\
	\PhiAppl{s\termEquals{A}t}&:=\termEqT{A}{\PhiAppl{s}}{\PhiAppl{t}}\plabel{PTEq}
	\end{align*}
\end{definition}

\subsection{Admissible rules for HOL}\label{sec:meta-thm:1}

The following lemma collects a few routine meta-theorems that we make use of later on:
\begin{lemma}\label{meta-thm:1}
	Given the inference rules for \hol{} (cfg. Figure~\ref{fig:holrulesAppendix}), the following rules are admissible:
		{\small
		\[
		\snamedRule{ctxThy}{\ded \Thy{T}}{\dedT \Ctx{\Gamma}}\tb
		\snamedRule{tpCtx}{\dedT \Ctx{\Gamma}}{\Gamma\dedT \Type{A}}\tb
		\snamedRule{typingTp}{\Gamma\dedT \Type{A}}{\Gamma\dedT t:A}\tb
		\snamedRule{validTyping}{\Gamma\dedT F:\bool}{\Gamma\dedT F}\tb
		\]
		\[
		\rnamedRule{constS}{\Gamma\dedT c:A}{\thyIn{c:A}{T}}{constS}\tb
		\rnamedRule{varS}{\Gamma\dedT x:A}{\ctxIn{x:A}{\Gamma}}{varS}
		\]
		\[
		\rnamedRule{$\typeEquals$refl}{\Gamma\dedT A\typeEquals A}{\Gamma\dedT \Type{A}}{tpEqRefl}\tb
		\rnamedRule{$\typeEquals$sym}{\Gamma\dedT A'\typeEquals A}{\Gamma\dedT A\typeEquals A'}{tpEqSym}\tb 
		\rnamedRule{$\typeEquals$trans}{\Gamma\dedT A\typeEquals A''}{\Gamma\dedT A\typeEquals A'\tb \Gamma\dedT A'\typeEquals A''}{tpEqTrans}
		\]
		\[
		\snamedRule{eqTyping}{\Gamma\dedT s:A}{\Gamma\dedT s\termEquals{A}t}\tb
		\snamedRule{implTypingL}{\Gamma\dedT F:\bool}{\judg F\impl G}\tb
		\snamedRule{implTypingR}{\Gamma\dedT G:\bool}{\judg F\impl G}\]
		\[
		\snamedRule{typesUnique}{\Gamma\dedT A\typeEquals A'}{\Gamma\dedT s:A\tb \Gamma\dedT s:A'}
		\tb
		\rnamedRule{typingWf}{\Gamma\dedT t:A}{\Gamma\dedT f\ t:B \quad \Gamma\dedT f:A\to B}{typingWellformed}
		\]
		\[
		\snamedRule{applType}{\Gamma\dedT f:A\to B}{\Gamma\dedT t:A\tb \Gamma\dedT f\ t:B}\tb
		\snamedRule{rewriteTyping}{\Gamma\dedT \subst{s}{x}{t}:A}{\concatCtx{\Gamma}{x:B}\dedT s:A\quad \Gamma\dedT t:B}
		\]
		\[
		\rnamedRule{monotonic$\ded$}{\concatCtx{\Gamma}{\namedass{ass}{F}}\dedT G}{\Gamma\dedT F:\bool\quad\Gamma\dedT G}{assDed}
		\quad
		\rnamedRule{var$\ded$}{\concatCtx{\Gamma}{x:A}\dedT J}{\Gamma\dedT \Type{A}\quad\Gamma\dedT J\tb \text{ for any statement }\dedT J}{varDed}
		\]
		\[
		\rnamedRule{$\forall$type}{\Gamma\dedT \univQuant{x}{A}F:\bool}{\concatCtx{\Gamma}{x:A}\dedT F:\bool}{forallType}\tb
		\rnamedRule{$\forall$I}{\Gamma\dedT \univQuant{x}{A}F}{\concatCtx{\Gamma}{x:A}\dedT F}{forallI}\tb
		\rnamedRule{$\forall$E}{\Gamma\dedT \subst{F}{x}{t}}{\Gamma\dedT \univQuant{x}{A}F \quad \Gamma\dedT t:A}{forallE}
		\]
		\[	
		\snamedRule{assTyping}{\Gamma\dedT F:\bool}{\Ctx{\Gamma}\quad \ctxIn{F}{\Gamma}}\tb
		\tb
		\rnamedRule{cong$:$}{\Gamma\dedT t':A'}{\Gamma\dedT t\termEquals{A} t' \quad  \Gamma\dedT A\typeEquals A' \quad  \Gamma\dedT t:A'}{congColon}
		\]
		\[	
		\snamedRule{propExt}{\Gamma\dedT F\termEquals{\bool}G}{\concatCtx{\Gamma}{\namedass{ass}{F}}\dedT G\quad \concatCtx{\Gamma}{\namedass{ass_G}{G}}\dedT F}\tb
		\snamedRule{extensionality}{\Gamma\dedT f\termEquals{A\to B}f'}{\concatCtx{\Gamma}{x:A}\dedT f\ x\termEquals{B} f'\ x}
		\]
		\[\snamedRule{trans}{\Gamma\dedT s\termEquals{A}u}{\Gamma\dedT s\termEquals{A}t\tb \Gamma\dedT t\termEquals{A}u}\tb
		\rnamedRule{$\termEquals{}$cong}{\Gamma\dedT (s\termEquals{A}t) \termEqB (s'\termEquals{A}t')}{\Gamma\dedT s\termEquals{A}s'\tb \Gamma\dedT t\termEquals{A}t'}{termEqCong}\tb
		\]
		\[
		\rnamedRule{$\forall$cong}{\Gamma\dedT\forall x:A.F\termEquals{\bool}\forall x:A'.F'}{\Gamma\dedT A\typeEquals A'\tb \concatCtx{\Gamma}{x:A}\dedT F\termEquals{\bool}F'}{forallCong}\tb
		\rnamedRule{$\impl$cong}{\Gamma\dedT F\Rightarrow G\termEquals{\bool} F'\Rightarrow G'}{\Gamma\dedT F\termEquals{\bool}F'\tb \Gamma\dedT G\termEquals{\bool}G'}{implCong}
		\]
		\[
		\rnamedRule{$\forall\impl$}{\Gamma\dedT\univQuant{x}{A}F\impl \univQuant{x}{A}G}{\concatCtx{\Gamma}{x:A}\dedT F\impl G}{forallImpl}\tb
		\rnamedRule{$\impl$Funct}{\Gamma\dedT (F\impl G)\impl \left(F'\impl G'\right)}{\Gamma\dedT G\impl G'\tb \Gamma\dedT F'\impl F}{implFunctorial}
		\]
		\[
		\rnamedRule{$\ded$cong}{\Gamma\dedT F'}{\Gamma\dedT F\termEquals{\bool} F'\quad \Gamma\dedT F}{dedCong}\tb
		\snamedRule{rewrite}{\Gamma\dedT \subst{F}{x}{t'}}{\Gamma\dedT \subst{F}{x}{t}\tb \Gamma\dedT t\termEquals{A}t'\tb \concatCtx{\Gamma}{x:A}\dedT F:\bool}
		\]
	}
\end{lemma}

\begin{proof}[Proof of Lemma~\ref{meta-thm:1}:]\ 
	\paragraph{Regarding \ruleRef{ctxThy}, \ruleRef{tpCtx}, \ruleRef{typingTp} and \ruleRef{validTyping}:}
	We can show these rules easily by induction on the inference rules of \hol{}. 
	In each step of the induction we show that the assumptions of the rule and these four rules holding on the assumptions of the rule implies that these four rules also hold for the conclusion of the rule. 
	In case of a validity rule this means that we need to show that the conclusion is well-typed, in case of a typing rule we need to show that the type of the term in the typing statement in the conclusion is well-formed, in case of a well-formedness rule for types we need to show that the context of the conclusion is well-formed and in case of a well-formedness rule for contexts we need to show that the theory relative to which the conclusion is stated is well-formed.  
	These properties hold by construction of the inference rules (each inference rules makes whatever assumptions are necessary to guarantee that the inductive step in this proof works).
	\paragraph{Regarding \ruleRef{tpEqRefl}, \ruleRef{tpEqSym} and \ruleRef{tpEqTrans}:}
	We can show them by induction on the two rules that allow showing type equality, namely \ruleRef{congBase} and \ruleRef{congTo}. 
	The rule \ruleRef{congBase} only allows showing type equality for identical types (so symmetry is trivial and transitivity reduces to one of the two assumption of the transitivity rule).
	The rule \ruleRef{congTo} only allows showing type equality for $\to$-types $A\to B$ and $A'\to B'$ based on equal types $A\typeEquals A'$ and $B\typeEquals B'$. 
	Now, if the assumption of \ruleRef{tpEqSym} is shown using rule \ruleRef{congTo}, by induction hypothesis we yield that the type equalities $A\typeEquals A'$ and $B\typeEquals B'$ satisfy symmetry and using \ruleRef{congTo} with the swapped type equalities yields $A'\to B'\typeEquals A\to B$ as desired.
	Similarly, if $A\to B\typeEquals A'\to B'$ and $A'\to B'\typeEquals A''\to B''$ are both shown using the rule \ruleRef{congTo} we have $A\typeEquals A'$ and $A'\typeEquals A''$ and hence by induction hypothesis $A\typeEquals A''$ and $B\typeEquals B''$, so rule \ruleRef{congTo} implies $A\to B\typeEquals A''\to B''$ as desired.
	
	\paragraph{Regarding \ruleRef{eqTyping}, \ruleRef{implTypingL} and \ruleRef{implTypingR}:}
	By assumption we have $\judg s\termEquals{A}t$ resp. $\judg F\impl G$.
	By the rules \ruleRef{ctxThy}, \ruleRef{tpCtx}, \ruleRef{typingTp} and \ruleRef{validTyping} it follows that the theory $T$ and context $\Gamma$ must be well-formed. Since the well-formedness of a context with an additional assumption or of a theory with an additional axiom can only be shown by rule \ruleRef{ctxAssume} and \ruleRef{thyAxiom} respectively it follows that axioms in $T$ and assumptions in $\Gamma$ must be well-typed and that the proof of the axiom or assumption being well-typed must not use that assumption or axiom. 
	We can then prove by induction on derivations that whenever $\judg s\termEquals{A}t$ is concluded in a step then the assumption of the step imply that $\judg s:A$ and $\judg t:A$ are derivable and whenever $\judg F\impl G$ is concluded in a step $\judg F:\bool$ and $\judg G:\bool$ are derivable from the assumptions of the step and furthermore whenever we conclude a validity statement, the formula is well-typed and in all conclusions in the derivation all types and contexts are well-formed.  
	
	The the claim follows directly from using what we just proved for the step showing the assumption $\judg s\termEquals{A}t$ resp. $\judg F\impl G$.
	
	\paragraph{Regarding \ruleRef{congColon}:}
	This follows from the fact that provably equal types in \hol{} are necessarily identical, so whenever we have $A\typeEquals A'$ and $t:A$, we already have that $t:A'$ holds by assumption. 
	We can show that provably equal types are identical by induction on the two type-equality rules in \hol{} namely \ruleRef{congBase} and \ruleRef{congTo}. 
	
	\paragraph{Regarding \ruleRef{varS} and \ruleRef{constS}:}
	These follow from the fact identical types are equal (by rule \ruleRef{congBase}) and the rule \ruleRef{var} respectively \ruleRef{const}.
	
	\paragraph{Regarding \ruleRef{typesUnique}:}
	There is no rule allowing to show that a well-typed Boolean term has any other type than $\bool$ (since we only allow validity but not type equality or typing assumptions). Hence $C, C'$ are not $\bool$ and $s$ not of type $\bool$. Continue by induction on the shape of $s$. 
	
	Depending on the shape of $s$ at most one of the two assumption (wlog. (renaming) assume that it is $\Gamma\dedT s:C$) can be concluded using one of the rules \ruleRef{const} or \ruleRef{var}. The other must be shown using on of the rules \ruleRef{lambda}, \ruleRef{appl} (since names of context variables and constants are mutually distinct). 
	
	In the case that rule \ruleRef{lambda} was used to show $s:C'$ (with $C'=A\to B'$ and $s=\lambdaFun{x}{A}t$) it follows from the shape of $s$ that also $C=A\to B$ for some type $B$ with $\concatCtx{\Gamma}{x:A}\dedT t:B$ (as also $s:C$ must be proven using rule \ruleRef{lambda}). By the induction hypothesis it follows that then $\concatCtx{\Gamma}{x:A}\dedT B\typeEquals B'$. 
	Since equal types in \hol{} must be identical (see proof of rule \ruleRef{congColon}), it follows that $\Gamma\dedT B\typeEquals B'$.
	Rule \ruleRef{congTo} then implies the claim of $\Gamma\dedT C\typeEquals C'$.
	
	In the case that rule \ruleRef{appl} was used to show $s:C'$ (with $C'=B'$ and $s=f\ t$ and $\Gamma\dedT f:A'\to B'$ and $\Gamma\dedT t:A'$) it follows from the shape of $s$ that also $C=B$ for some type $B$ with $\Gamma\dedT f:A\to B$ and $\Gamma\dedT t:A$. 
	By the induction hypothesis (applied to $f$ and $t$) it follows that $\Gamma\dedT A\to B\typeEquals A'\to B'$ and $\Gamma\dedT A\typeEquals A'$. 
	Since equal types in \hol{} are identical it follows that also $\Gamma\dedT B\typeEquals B'$ holds, so by assumption the conclusion holds.
	
	\paragraph{Regarding \ruleRef{typingWellformed}:}
	The first assumption can only be proven using rule \ruleRef{appl}. 
	In the first case the conclusion is an assumption of the rule used to prove $\Gamma\dedT f\ t:B$ and must therefore hold.
	
	\paragraph{Regarding \ruleRef{assDed} and \ruleRef{varDed}:}
	The idea for showing this is observing that adding additional variables or assumptions to the contexts occuring in a derivation will always result in another (equally valid) derivation. 
	This can be shown by induction on derivations. 
	If in a derivation some rule is used to conclude $\Gamma\dedT F$ and adding additional variables and assumptions to $\Gamma$ as well as the contexts of all judgements appearing in those derivations leaves those derivations valid, then we need to show that also the assumptions of the rule are still satisfied. Since we have valid derivations for all assumptions that are judgements themselves, we only need to check assumptions of rules that are not judgements. 
	The only assumptions of this kind that appear in inference rules are assumptions that some variable or context assumption is contained in a context. Adding further assumptions or variables will not make this false if is was true initially. 
	
	This proves that \ruleRef{assDed} and \ruleRef{varDed} hold.
	
	\paragraph{Regarding \ruleRef{rewriteTyping}:}
	This can be seen by applying the substitution $\subst{\cdot}{x}{t}$ to the terms in the derivation (but not the variable declaration $x$ in the context itself) of $\concatCtx{\Gamma}{x:B}\dedT s:A$, using the fact that $\Gamma\dedT t:B$ holds instead of rule \ruleRef{var} (that may be used to conclude $\concatCtx{\Gamma}{x:B}\dedT x:B$ in the original declaration). 
	This leads to a derivation of the conclusion $\concatCtx{\Gamma}{x:A}\dedT \subst{s}{x}{t}:A$. 
	Since $x$ doesn't appear in the derivation, removing the variable declaration $x:A$ from the context leads to a valid derivation of $\Gamma\dedT \subst{s}{x}{t}:A$ as desired. 
	
	\paragraph{Regarding \ruleRef{assTyping}:}
	\begin{align}
	\NDLine{}{\Ctx{\Gamma}}{by assumption}\label{astp1}\\
	\NDLine{}{\ctxIn{F}{\Gamma}}{by assumption}\label{astp2}\\
	\NDLine{\Gamma}{F}{\ruleRef{assume},(\ref{astp2}),(\ref{astp1})}\label{astp3}\\
	\NDLine{\Gamma}{F:\bool}{\ruleRef{validTyping},(\ref{astp3})}
	\end{align}
	
	\paragraph{Regarding \ruleRef{forallType}:}\ 
	\begin{align}
	\NDLineT{\concatCtx{\Gamma}{x:A}}{F:\bool}{\byAss}\label{univTp1}\\
	\NDLineTG{\lambdaFun{x}{A}F:A\to\bool}{\ruleRef{lambda},(\ref{univTp1})}\label{univTp2}\\				
	\NDLineT{\concatCtx{\Gamma}{x:A}}{\T:\bool}{by definition}\label{univTp3}\\
	\NDLineTG{\lambdaFun{x}{A}\T:A\to\bool}{\ruleRef{lambda},(\ref{univTp3})}\label{univTp4}\\
	\NDLineTG{\lambdaFun{x}{A}F\termEqB \lambdaFun{x}{A}\T:\bool}{\ruleRef{eqType},(\ref{univTp2}),(\ref{univTp4})}\label{univTp5}\\
	\NDLineTG{\univQuant{x}{A}F:\bool}{by definition,(\ref{univTp5})}
	\end{align}
	
	\paragraph{Regarding \ruleRef{forallE}:}\ 
	\begin{align}
	\NDLineTG{\univQuant{x}{A}F}{\byAss}\label{univE1}\\
	\NDLineTG{t:A}{\byAss}\label{univE2}\\
	\NDLineTG{\lambdaFun{x}{A}F\termEquals{A\to \bool}\lambdaFun{x}{A}\T}{by definition,(\ref{univE1})}\label{univE3}\\
	\NDLineTG{t\termEquals{A}t}{\ruleRef{refl},(\ref{univE2})}\label{univE4}\\
	\NDLineTG{\big(\lambdaFun{x}{A}F\big)\ t\termEqB \big(\lambdaFun{x}{A}\T\big)\ t}{\ruleRef{congAppl},(\ref{univE3}),(\ref{univE4})}\label{univE5}\\
	\NDLineTG{\big(\lambdaFun{x}{A}\T\big)\ t\termEqB \big(\lambdaFun{x}{A}F\big)\ t}{\ruleRef{sym},(\ref{univE5})}\label{univE5a}\\
	\NDLineTG{\big(\lambdaFun{x}{A}F\big)\ t:\bool}{\ruleRef{eqTyping},(\ref{univE5})}\label{univE6}\\
	\NDLineTG{\big(\lambdaFun{x}{A}\T\big)\ t:\bool}{\ruleRef{eqTyping},(\ref{univE5a})}\label{univE7}\\
	\NDLineTG{\big(\lambdaFun{x}{A}F\big)\ t\termEqB \subst{F}{x}{t}}{\ruleRef{beta},(univE6)}\label{univE8}\\
	\NDLineTG{\big(\lambdaFun{x}{A}\T\big)\ t\termEqB \T}{\ruleRef{beta},(univE6)}\label{univE9}\\
	\NDLineTG{\subst{F}{x}{t}\termEqB \big(\lambdaFun{x}{A}\T\big)\ t}{\ruleRef{rewrite},(\ref{univE5}),(\ref{univE8})}\label{univE10}\\
	\NDLineTG{\subst{F}{x}{t}\termEqB \T}{\ruleRef{rewrite},(\ref{univE10}),(\ref{univE9})}\\
	\NDLineTG{\T}{\ruleRef{refl},definition}\label{trueValidUnivE}\\
	\NDLineTG{\subst{F}{x}{t}}{\ruleRef{congDed},(\ref{univE9}),(\ref{trueValidUnivE})}
	\end{align}
	
	\paragraph{Regarding \ruleRef{propExt}:}
	\newcommand{\p}{\lambdaFun{x}{\bool} \univQuant{y}{\bool} (x\impl y) \impl ((y\impl x)\impl x\termEquals{\bool}y)}
	\newcommand{\pF}{\univQuant{y}{\bool} (F\impl y) \impl ((y\impl F)\impl F\termEquals{\bool}y)}	
	\newcommand{\pFl}{\univQuant{y}{\bool} (\F\impl y) \impl ((y\impl \F)\impl \F\termEquals{\bool}y)}	
	\newcommand{\pTr}{\univQuant{y}{\bool} (\T\impl y) \impl ((y\impl \T)\impl \T\termEquals{\bool}y)}
	\newcommand{\pTT}{(\T\impl \T) \impl ((\T\impl \T)\impl \T\termEquals{\bool}\T)}
	\newcommand{\pTF}{(\T\impl \F) \impl ((\F\impl \T)\impl \T\termEquals{\bool}\F)}
	\newcommand{\pFT}{(\F\impl \T) \impl ((\T\impl \F)\impl \F\termEquals{\bool}\T)}
	\newcommand{\pFF}{(\F\impl \F) \impl ((\F\impl \F)\impl \F\termEquals{\bool}\F)}
	\small{\begin{align}
		\concatCtx{\Gamma}{&\namedass{ass_F}{F}}\dedT G && \text{by assumption}\label{pe1}\\
		\concatCtx{\Gamma}{&\namedass{ass_F}{F}}\dedT G:\bool && \text{\ruleRef{validTyping},(\ref{pe1})}\label{pe1a}\\
		\NDLineT{}{\Ctx{\concatCtx{\Gamma}{\namedass{ass}{F}}}}{\ruleRef{tpCtx},(\ref{pe1a})}\label{pe1b}\\
		\NDLineTG{F:\bool}{\ruleRef{assTyping},(\ref{pe1b})}\label{FTp}\\				
		\concatCtx{\Gamma}{&\namedass{ass_G}{G}}\dedT F&&\text{by assumption}\label{pe2}\\
		\concatCtx{\Gamma}{&\namedass{ass_G}{G}}\dedT F:\bool && \text{\ruleRef{validTyping},(\ref{pe1})}\label{pe2a}\\
		\NDLineT{}{\Ctx{\concatCtx{\Gamma}{\namedass{ass_G}{G}}}}{\ruleRef{tpCtx},(\ref{pe2a})}\label{pe2b}\\
		\NDLineTG{G:\bool}{\ruleRef{assTyping},(\ref{pe2b})}\label{GTp}
		\\
		\NDLineTG{\T:\bool}{by definition}\label{eqnTTp}\\
		\NDLineTG{\F:\bool}{by definition}\label{eqnFTp}\\
		\NDLineTG{F\impl G}{\ruleRef{implI}, (\ref{FTp}),(\ref{pe1})}\label{pe10a}\\
		\NDLineTG{G\impl F}{\ruleRef{implI}, (\ref{GTp}),(\ref{pe2})}\label{pe20a}
		\\
		\NDLineTG{\T\termEqB \T}{\ruleRef{refl},(\ref{eqnTTp})}\label{eqn6T}\\
		\NDLineTG{\T\impl \T:\bool}{\ruleRef{implType},(\ref{eqnTTp}),(\ref{eqnTTp})}\label{pexImxTpT}\\
		\concatCtx{\Gamma&}{\namedass{ass1}{\T\impl \T}}\dedT \T\termEqB \T&&\text{\ruleRef{assDed},(\ref{pexImxTpT}), (\ref{eqn6T})}\label{eqn7T}\\
		\NDLineTG{(\T\impl \T) \impl \T\termEqB \T}{\ruleRef{implI},(\ref{pexImxTpT}), (\ref{eqn7T})}\label{eqn8T}\\
		\concatCtx{\Gamma&}{\namedass{ass1}{\T\impl \T}}\dedT (\T\impl \T) \impl \T\termEqB \T&&\text{\ruleRef{assDed},(\ref{pexImxTpT}),(\ref{eqn8T})}\label{eqn9T}\\
		\NDLineTG{(\T\impl \T) \impl ((\T\impl \T) \impl \T\termEqB \T)}{\ruleRef{implI},(\ref{pexImxTpT}), (\ref{eqn9T})}\label{eqnpTT}
		\intertext{Analogously we can show:}				
		\NDLineTG{(\F\impl \F) \impl ((\F\impl \F) \impl \F\termEqB \F)}{analogous to  (\ref{eqnpTT})}\label{eqnpFF}\\	
		\NDLineTG{\T\impl\F:\bool}{\ruleRef{implType},(\ref{eqnTTp}),(\ref{eqnFTp})}\label{eqnTimFTp}\\
		\NDLineT{}{\Ctx{\concatCtx{\Gamma}{\namedass{ass2}{\T\impl \F}}}}{\ruleRef{ctxAssume},(\ref{eqnTimFTp})}\label{GTimplFCtx}\\
		\concatCtx{\Gamma}{\namedass{ass2}{&\T}}\impl\F\dedT\T\impl\F:\bool\nonumber\\
		&&&\QQQQNegSp\QQQNegSp\text{\ruleRef{implType},\ruleRef{assDed},(\ref{eqnTimFTp}),(\ref{eqnTTp}),\ruleRef{assDed},(\ref{eqnTimFTp}),(\ref{eqnFTp})}\label{eqnTimFTpC}\\
		\NDLineTG{\F\impl\T:\bool}{\ruleRef{implType},(\ref{eqnFTp}),(\ref{eqnTTp})}\label{eqnFimTTp}\\
		\concatCtx{\Gamma}{&\namedass{ass2}{\F}}\impl\T\dedT\F\impl\T:\bool\nonumber\\
		&&&\QQQQNegSp\QQQNegSp\text{\ruleRef{implType},\ruleRef{assDed},(\ref{eqnTimFTp}),(\ref{eqnFTp}),\ruleRef{assDed},(\ref{eqnTimFTp}),(\ref{eqnTTp})}\label{eqnFimTTpC}\\
		\NDLineTG{\T}{\ruleRef{refl},(\ref{eqnTTp})}\label{eqnT}\\
		\concatCtx{\Gamma}{& \namedass{ass2}{\T\impl\F}}\dedT \T\impl\F&&\text{\ruleRef{assume},(\ref{GTimplFCtx})}\label{eqnTFass}\\
		\concatCtx{\Gamma}{& \namedass{ass2}{\T\impl\F}}\dedT \F&&\text{\ruleRef{implE},(\ref{eqnTFass}),(\ref{eqnT})}\label{eqnBasic}\\
		\concatCtx{\Gamma}{& \namedass{ass2}{\T\impl\F}}\dedT \T&&\text{\ruleRef{assDed},(\ref{eqnTimFTp}),(\ref{eqnT})}\label{eqnTriv}
		\\
		\concatCtx{\concatCtx{\Gamma}{& \namedass{ass2}{\T\impl\F}}}{y:\bool}\dedT y:\bool&&\text{\ruleRef{varS},\ruleRef{ctxVar},(\ref{GTimplFCtx})}\label{reyTp}\\
		\concatCtx{\Gamma}{& \namedass{ass2}{\T\impl\F}}\dedT \lambdaFun{y}{\bool}y:\bool\to\bool &&\text{\ruleRef{lambda},(\ref{reyTp})}\label{eqnIdTp}\\
		\concatCtx{\Gamma}{& \namedass{ass2}{\T\impl\F}}\dedT \T:\bool &&\text{by definition}\label{eqnTTpC}\\
		\concatCtx{\Gamma}{& \namedass{ass2}{\T\impl\F}}\dedT \F:\bool &&\text{by definition}\label{eqnFTpC}\\
		\concatCtx{\Gamma}{& \namedass{ass2}{\T\impl\F}}\dedT \F\termEqB\T:\bool &&\text{\ruleRef{eqType},(\ref{eqnFTpC}),(\ref{eqnTTpC})}\label{eqnFeqTTp}\\
		\concatCtx{\Gamma}{& \namedass{ass2}{\T\impl\F}}\dedT (\lambdaFun{y}{\bool}y)\ \F\termEqB\T:\bool &&\text{\ruleRef{appl},(\ref{eqnIdTp}),(\ref{eqnFeqTTp})}\label{eqnIdFeqTTp}\\
		\concatCtx{\Gamma}{& \namedass{ass2}{\T\impl\F}}\dedT \T\termEqB\F:\bool &&\text{\ruleRef{eqType},(\ref{eqnTTpC}),(\ref{eqnFTpC})}\label{eqnTeqFTp}\\
		\concatCtx{\Gamma}{& \namedass{ass2}{\T\impl\F}}\dedT (\lambdaFun{y}{\bool}y)\ \T\termEqB\F:\bool &&\text{\ruleRef{appl},(\ref{eqnIdTp}),(\ref{eqnTeqFTp})}\label{eqnIdTeqFTp}\\
		\concatCtx{\Gamma}{& \namedass{ass2}{\T\impl\F}}\dedT (\lambdaFun{y}{\bool}y)\ \T:\bool &&\text{\ruleRef{appl},(\ref{eqnIdTp}),(\ref{eqnTTpC})}\label{eqnIdTTp}\\
		\concatCtx{\Gamma}{& \namedass{ass2}{\T\impl\F}}\dedT (\lambdaFun{y}{\bool}y)\ \F:\bool &&\text{\ruleRef{appl},(\ref{eqnIdTp}),(\ref{eqnFTpC})}\label{eqnIdFTp}\\
		\concatCtx{\Gamma}{& \namedass{ass2}{\T\impl\F}}\dedT (\lambdaFun{x}{\bool}x)\ \F\termEqB\F&&\text{\ruleRef{beta},(\ref{eqnIdFTp})}\label{eqnFB}\\
		\concatCtx{\Gamma}{& \namedass{ass2}{\T\impl\F}}\dedT (\lambdaFun{y}{\bool}y)\ \F&&\text{\ruleRef{congDed},(\ref{eqnFB}),(\ref{eqnBasic})}\label{eqnBC}\\ 
		\concatCtx{\Gamma}{& \namedass{ass2}{\T\impl\F}}\dedT (\lambdaFun{x}{\bool}x)\ \T\termEqB\T&&\text{\ruleRef{beta},(\ref{eqnIdTTp})}\label{eqnTB}\\
		\concatCtx{\Gamma}{& \namedass{ass2}{\T\impl\F}}\dedT (\lambdaFun{y}{\bool}y)\ \T&&\text{\ruleRef{congDed},(\ref{eqnTB}),(\ref{eqnTriv})}\label{eqnTrivC}\\ 
		\concatCtx{\Gamma}{& \namedass{ass2}{\T\impl\F}}\dedT \univQuant{x}{\bool}(\lambdaFun{y}{\bool}y)\ x &&\text{\ruleRef{boolExt},(\ref{eqnBC}),(\ref{eqnTrivC})}\label{eqnAny}\\
		\concatCtx{\Gamma}{& \namedass{ass2}{\T\impl\F}}\dedT (\lambdaFun{y}{\bool}y)\ (\T\termEqB\F)&&\text{\ruleRef{forallE},(\ref{eqnAny}),(\ref{eqnTeqFTp})}\label{eqn20}\\
		\concatCtx{\Gamma}{& \namedass{ass2}{\T\impl\F}}\dedT (\lambdaFun{y}{\bool}y)\ (\F\termEqB\T)&&\text{\ruleRef{forallE},(\ref{eqnAny}),(\ref{eqnFeqTTp})}\label{eqn21}\\
		\concatCtx{\Gamma}{& \namedass{ass2}{\T\impl\F}}\dedT (\lambdaFun{y}{\bool}y)\ (\T\termEqB\F)\nonumber\\&\quad\termEqB (\T\termEqB\F)&&\text{\ruleRef{beta},(\ref{eqnIdTeqFTp})}\label{eqn20B}\\
		\concatCtx{\Gamma}{& \namedass{ass2}{\T\impl\F}}\dedT (\lambdaFun{y}{\bool}y)\ (\F\termEqB\T)\nonumber\\&\quad\termEqB (\F\termEqB\T)&&\text{\ruleRef{beta},(\ref{eqnIdFeqTTp})}\label{eqn21B}\\
		\concatCtx{\Gamma}{& \namedass{ass2}{\T\impl\F}}\dedT (\T\termEqB\F)\nonumber\\&\quad\termEqB (\lambdaFun{y}{\bool}y)\ (\T\termEqB\F)&&\text{\ruleRef{sym},(\ref{eqn20B})}\label{eqn20BS}\\
		\concatCtx{\Gamma}{& \namedass{ass2}{\T\impl\F}}\dedT (\F\termEqB\T)\nonumber\\&\quad\termEqB (\lambdaFun{y}{\bool}y)\ (\F\termEqB\T)&&\text{\ruleRef{sym},(\ref{eqn21B})}\label{eqn21BS}\\
		\concatCtx{\Gamma}{& \namedass{ass2}{\T\impl\F}}\dedT \T\termEqB\F&&\text{\ruleRef{congDed},(\ref{eqn20BS}),(\ref{eqn20})}\label{eqn20red}\\
		\concatCtx{\Gamma}{& \namedass{ass2}{\T\impl\F}}\dedT \F\termEqB\T&&\text{\ruleRef{congDed},(\ref{eqn21BS}),(\ref{eqn21})}\label{eqn21red}\\		
		\concatCtx{\concatCtx{\Gamma}{& \namedass{ass2}{\T\impl\F}}}{\F\impl\T}\dedT \T\termEqB\F&&\text{\ruleRef{assDed},(\ref{eqnFimTTpC}),(\ref{eqn20red})}\label{eqn20ext}\\	
		\concatCtx{\Gamma}{& \namedass{ass2}{\T\impl\F}}\dedT (\F\impl\T)\impl(\T\termEqB\F)&&\text{\ruleRef{implI},(\ref{eqnTimFTpC}),(\ref{eqn20ext})}\label{eqn20imp}\\
		\NDLineTG{(\T\impl\F)\impl(\F\termEqB\T)}{\ruleRef{implI},(\ref{eqnFimTTp}),(\ref{eqn21red})}\label{eqn21imp}\\
		\NDLineTG{\pTF}{\ruleRef{implI},(\ref{eqnTimFTp}),(\ref{eqn20imp})}\label{eqnpTF}\\
		\concatCtx{\Gamma}{&\namedass{ass3}{\F}}\impl\T\dedT (\T\impl\F)\impl(\T\termEqB\F)&&\text{\ruleRef{assDed},(\ref{eqnFimTTpC}),(\ref{eqn21imp})}\label{eqn20impA}\\
		\NDLineTG{\pFT}{\ruleRef{implI},(\ref{eqnFimTTpC}),(\ref{eqn20impA})}\label{eqnpFT}\\
		\Gamma\dedT&\pTr\nonumber\\&&&\QQQNegSp\QQNegSp
		\text{\ruleRef{boolExt},(\ref{eqnpTT}),\ruleRef{congDed},\ruleRef{beta},(\ref{eqnpTF})}\label{eqnpTred}\\
		\Gamma\dedT&\pFl\nonumber\\&&&\QQQNegSp\QQNegSp\text{\ruleRef{boolExt},(\ref{eqnpFT}),\ruleRef{congDed},\ruleRef{beta},(\ref{eqnpFF})}\label{eqnpFred}
		\end{align}
		Showing that these terms are all well-typed:
		\begin{align}
		\concatCtx{\Gamma}{&x:\bool}, y:\bool\dedT x:\bool &&\text{\ruleRef{varS},\ruleRef{ctxAssume},\ruleRef{ctxAssume}}\label{peXTp}\\
		\concatCtx{\Gamma}{&x:\bool}, y:\bool\dedT y:\bool &&\text{\ruleRef{varS},\ruleRef{ctxAssume},\ruleRef{ctxAssume}}\label{peYTp}\\
		\concatCtx{\Gamma}{&x:\bool}, y:\bool\dedT (x\impl y):\bool &&\text{\ruleRef{implType},(\ref{peXTp}),(\ref{peYTp})}\label{peXimYTp}\\
		\concatCtx{\Gamma}{&x:\bool}, y:\bool\dedT (y\impl x):\bool &&\text{\ruleRef{implType},(\ref{peYTp}),(\ref{peXTp})}\label{peYimXTp}\\
		\concatCtx{\Gamma}{&x:\bool}, y:\bool\dedT (x\termEqB y):\bool &&\text{\ruleRef{eqType},(\ref{peXTp}),(\ref{peYTp})}\label{peXeqYTp}\\
		\concatCtx{\Gamma}{&x:\bool}, y:\bool\dedT (y\impl x)\impl(x\termEqB y):\bool &&\text{\ruleRef{implType},(\ref{peYimXTp}),(\ref{peXeqYTp})}\label{peYimXimXeqYTp}\\
		\concatCtx{\Gamma}{&x:\bool}, y:\bool\dedT (x\impl y) \impl ((y\impl x)\impl x\termEquals{\bool}y):\bool
		\negSp\negSp\negSp&&\qquad\qquad\text{\ruleRef{implType},(\ref{peXimYTp}),(\ref{peYimXimXeqYTp})}\label{peqbTp}\\
		\concatCtx{\Gamma}{&x:\bool}\dedT \univQuant{y}{\bool} (x\impl y) \impl ((y\impl x)\impl x\termEquals{\bool}y):\bool \QQQNegSp&&\qqquad\quad\text{\ruleRef{forallType},(\ref{peqbTp})}\label{peqTp}\\
		\NDLineTG{\p:\bool\to\bool\QQQNegSp}{\quad\qqquad\ruleRef{lambda},(\ref{peqTp})}\label{pepTp}\\
		\NDLineTG{(\p)\ \T:\bool\QQQNegSp}{\qqquad\quad\ruleRef{appl},(\ref{pepTp}),(\ref{eqnTTp})}\label{eqnPTTp}\\
		\NDLineTG{(\p)\ \F:\bool\QQQNegSp}{\qqquad\quad\ruleRef{appl},(\ref{pepTp}),(\ref{eqnFTp})}\label{eqnPFTp}
		\end{align}
		Beta reduce and conclude claim:
		\begin{align}
		\NDLineTG{(\p)\ \T&&\nonumber\\&\termEqB\pTr}{\ruleRef{beta},(\ref{eqnPTTp})}\label{eqnpTB}\\
		\NDLineTG{(\p)\ \F&&\nonumber\\&\termEqB\pFl}{\ruleRef{beta},(\ref{eqnPFTp})}\label{eqnpFB}\\
		\NDLineTG{(\p)\ \T}{\ruleRef{congDed},(\ref{eqnpTB}),(\ref{eqnpTred})}\label{eqnpT}\\
		\NDLineTG{(\p)\ \F}{\ruleRef{congDed},(\ref{eqnpFB}),(\ref{eqnpFred})}\label{eqnpF}\\
		\NDLineTG{\univQuant{x}{\bool}(\p)\ x\QNegSp}{\qquad\ruleRef{boolExt},(\ref{eqnpT}),(\ref{eqnpF})}\label{eqnGoal}\\
		\NDLineTG{\p\ F}{\ruleRef{forallE},(\ref{eqnGoal}),(\ref{FTp})}\label{eqnGoal2pre}
		\\
		\NDLineTG{\p\ F:\bool}{\ruleRef{validTyping},(\ref{eqnGoal2pre})}\label{eqnGoal2preTp}
		\\
		\NDLineTG{\pF\termEqB\nonumber\\&
			\p\ F}{\ruleRef{beta},(\ref{eqnGoal2preTp})}\label{eqnGoal2pre2}
		\\\NDLineTG{\pF}{\ruleRef{congDed},(\ref{eqnGoal2pre2}),(\ref{eqnGoal2pre})}\label{eqnGoal2}\\
		\NDLineTG{(F\impl G)\impl ((G\impl F)\impl F\termEqB G}{\ruleRef{forallE},(\ref{eqnGoal2}),(\ref{GTp})}\label{eqn2last}\\
		\NDLineTG{(G\impl F)\impl F\termEqB G}{\ruleRef{implE}, (\ref{eqn2last}), (\ref{pe10a})}\label{eqnlast}\\
		\NDLineTG{F\termEqB G}{\ruleRef{implE}, (\ref{eqnlast}), (\ref{pe20a})}
		\end{align}}
	
	\paragraph{Regarding \ruleRef{forallI}:}\ 
	Note that while this proof uses Lemma~\ref{lem:eqTrue}, this is not a problem since we only use those of the rules in this lemma that are shown in the above cases which don't assume \ruleRef{forallI} or Lemma~\ref{lem:eqTrue}.
	\begin{align}
	\NDLineT{\concatCtx{\Gamma}{x:A}}{F}{\byAss}\label{univI1}\\
	\NDLineT{\concatCtx{\Gamma}{x:A}}{F\termEqB\T}{lemma\ \ref{lem:eqTrue},(\ref{univI1})}\label{univI2}\\
	\NDLineTG{A\typeEquals A}{\ruleRef{congBase}}\label{univI3}\\
	\NDLineTG{\lambdaFun{x}{A}F\termEquals{A\to \bool}\lambdaFun{x}{A}\T}{\ruleRef{congLam},(\ref{univI2}),(\ref{univI3})}\label{univI4}\\
	\NDLineTG{\univQuant{x}{A}F}{by definition,(\ref{univI4})}
	\end{align}
	
	\paragraph{Regarding \ruleRef{applType}:}
	\begin{align}
	\NDLineTG{t:A}{by assumption}\label{atass1}\\
	\NDLineTG{f\ t:B}{by assumption}\label{atass2}\\
	\intertext{Since typing of a function application can only be proven by rule \ruleRef{appl}, the assumptions of the rule must be satisfied, so we yield:}
	\NDLineTG{f:A'\to B}{see above}\label{atfTp}\\
	\NDLineTG{t:A'}{see above}\label{attTp}\\
	\NDLineTG{A\typeEquals A'}{\ruleRef{typesUnique},(\ref{attTp})}\label{atAeqAp}\\
	\NDLineTG{A'\typeEquals A}{\ruleRef{tpEqSym},(\ref{atAeqAp})}\label{atApeqA}\\
	\NDLineTG{A'\to B\typeEquals A\to B}{\ruleRef{congTo},(\ref{atApeqA})}\label{atApToBeqAToB}\\
	\NDLineTG{f\termEquals{A'\to B}f}{\ruleRef{refl},(\ref{atfTp})}\label{atfeqf}\\
	\NDLineTG{f:A\to B}{\ruleRef{congColon},(\ref{atfeqf}),(\ref{atApToBeqAToB}),(\ref{atfTp})}
	\end{align}
	
	\paragraph{Regarding \ruleRef{trans}:}
	\renewcommand{\p}{\lambdaFun{x}{A} \lambdaFun{y}{A} \lambdaFun{z}{A} x\termEquals{A}y \impl y\termEquals{A}z \impl x\termEquals{A}z}	
	\newcommand{\ps}{\lambdaFun{y}{A} \lambdaFun{z}{A} s\termEquals{A}y \impl y\termEquals{A}z \impl s\termEquals{A}z}	
	\newcommand{\pst}{\lambdaFun{z}{A} s\termEquals{A}t \impl t\termEquals{A}z \impl s\termEquals{A}z}
	\newcommand{\pstt}{s\termEquals{A}t \impl t\termEquals{A}t \impl s\termEquals{A}t}
	\newcommand{\pstu}{s\termEquals{A}t \impl t\termEquals{A}u \impl s\termEquals{A}u}
	
	\begin{align}
	\NDLineTG{s\termEquals{A}t}{by assumption}\label{trans1}\\
	\NDLineTG{s\termEquals{A}t:\bool}{\ruleRef{validTyping},(\ref{trans1})}\label{transsEqtTp}\\
	\NDLineTG{t\termEquals{A}u}{by assumption}\label{trans2}\\
	\concatCtx{\Gamma}{&\namedass{eq}{s\termEquals{A}t}}\dedT s\termEquals{A}t&&\text{\ruleRef{assume},\ruleRef{ctxAssume},(\ref{transsEqtTp})}\label{trans3}\\
	\concatCtx{\concatCtx{\Gamma}{& \namedass{eq}{s\termEquals{A}t}}}{t\termEquals{A}t}\dedT s\termEquals{A}t&&\text{\ruleRef{assDed},(\ref{transsEqtTp}), (\ref{trans3})}\label{trans4}\\
	\NDLineTG{t\termEquals{A}s}{\ruleRef{sym},(\ref{trans1})}\label{trans1S}\\
	\NDLineTG{t:A}{\ruleRef{eqTyping},(\ref{trans1S})}\label{transtTp}\\
	\NDLineTG{\Type{A}}{\ruleRef{typingTp},(\ref{transtTp})}\label{transATp}\\
	\NDLineTG{u\termEquals{A}t}{\ruleRef{sym},(\ref{trans2})}\label{trans2S}\\
	\NDLineTG{u:A}{\ruleRef{eqTyping},(\ref{trans2S})}\label{transuTp}\\
	\NDLineTG{t\termEquals{A}t:\bool}{\ruleRef{eqType},(\ref{transtTp}),(\ref{transtTp})}\label{transtEqtTp}\\
	\concatCtx{\Gamma}{& \namedass{eq}{s\termEquals{A}t}}\dedT  t\termEquals{A}t:\bool&&\text{\ruleRef{assDed},(\ref{transsEqtTp}), (\ref{transtEqtTp})}\label{transtEqtTpC}\\		
	\NDLineTG{s:A}{\ruleRef{eqTyping},(\ref{trans1})}\label{transsTp}\\
	\concatCtx{\Gamma}{& \namedass{eq}{s\termEquals{A}t}}\dedT t\termEquals{A}t\impl s\termEquals{A}t&&\text{\ruleRef{implI},(\ref{transtEqtTpC}) (\ref{trans4})}\label{trans5}\\
	\NDLineTG{\pstt}{\ruleRef{implI},(\ref{transsEqtTp}), (\ref{trans5})}\label{trans6}
	\\
	\Gamma,& z:A\dedT s:A &&\text{\ruleRef{varDed},(\ref{transATp}),(\ref{transsTp})} \label{transsTpC}\\
	\NDLineT{}{\Ctx{\concatCtx{\Gamma}{z:A}}}{\ruleRef{ctxVar},\ruleRef{tpCtx},(\ref{transsTp}),(\ref{transATp})}\label{transGCCtx}\\
	\Gamma,& z:A\dedT t:A &&\text{\ruleRef{varDed},(\ref{transATp}),(\ref{transtTp})} \label{transtTpC}\\		
	\Gamma,& z:A\dedT z:A &&\text{\ruleRef{varS},(\ref{transGCCtx})} \label{transzTpC}\\
	\Gamma,& z:A\dedT t\termEquals{A}z:\bool &&\text{\ruleRef{implI},(\ref{transtTpC}),(\ref{transzTpC})}\label{transtEqzTpC}\\
	\Gamma,& z:A\dedT s\termEquals{A}z:\bool &&\text{\ruleRef{implI},(\ref{transsTpC}),(\ref{transzTpC})}\label{transsEqzTpC}\\
	\Gamma,& z:A\dedT (t\termEquals{A}z) \impl (s\termEquals{A}z):\bool&&\text{\ruleRef{implI},(\ref{transtEqzTpC}),(\ref{transsEqzTpC})}\label{transtEqzImsEqzTpC}\\
	\Gamma,& z:A\dedT s\termEquals{A}t:\bool &&\text{\ruleRef{implI},(\ref{transsTpC}),(\ref{transtTpC})}\label{transsEqtTpC}\\
	\Gamma,& z:A\dedT (s\termEquals{A}t) \impl ((t\termEquals{A}z) \impl (s\termEquals{A}z)):\bool&&\text{\ruleRef{implType},(\ref{transsEqtTpC}),(\ref{transtEqzImsEqzTpC})} \label{transpbTp}\\
	\NDLineTG{\pst:A\to \bool}{\ruleRef{lambda},(\ref{transpbTp})}\label{transpstTp}\\
	\NDLineTG{(\pst)\ t:\bool}{\ruleRef{appl},(\ref{transpstTp}),(\ref{transtTp})}\label{transpstTTp}\\
	\NDLineTG{(\pst)\ t\nonumber&&\\
		&\termEqB \pstt}{\ruleRef{beta},(\ref{transpstTTp})}\label{trans7}\\
	\NDLineTG{(\pst)\ t}{\ruleRef{congDed}, (\ref{trans6}), (\ref{trans7})}\label{trans8}\\
	\NDLineTG{(\pst)\ t\nonumber&&\\
		&\termEqB(\pst)\ u}{\ruleRef{congAppl}, (\ref{trans1})}\label{trans9}\\
	\NDLineTG{(\pst)\ u\nonumber&&\\
		&\termEqB(\pst)\ t}{\ruleRef{sym}, (\ref{trans9})}\label{trans10}\\
	\NDLineTG{(\pst)\ u}{\ruleRef{congDed}, (\ref{trans10}), (\ref{trans8})}\label{trans11}\\
	\NDLineTG{(\pst)\ t:\bool}{\ruleRef{appl},(\ref{transpstTp}),(\ref{transuTp})}\label{transpstUTp}\\
	\NDLineTG{(\pst)\ u\nonumber&&\\
		&\termEqB\pstu}{\ruleRef{beta},(\ref{transpstUTp})}\label{trans12}\\
	\NDLineTG{\pstu}{\ruleRef{congDed},(\ref{trans12}), (\ref{trans11})}\label{trans13}\\
	\NDLineTG{t\termEquals{A}u \impl s\termEquals{A}u}{\ruleRef{implE},(\ref{trans13}),(\ref{trans1})}\label{trans14}\\
	\NDLineTG{s\termEquals{A}u}{\ruleRef{implE},(\ref{trans14}),(\ref{trans2})}
	\end{align}
	
	\paragraph{Regarding \ruleRef{extensionality}:}
	\begin{align}
	\NDLineT{\Gamma, x:A}{f\ x\termEquals{B}f'\ x}{by assumption}\label{ext1}\\
	\NDLineT{\Gamma, x:A}{f\ x:B}{\ruleRef{eqTyping},(\ref{ext1})}\label{extfxTp}\\
	\NDLineT{\Gamma, x:A}{f'\ x\termEquals{B}f\ x}{\ruleRef{sym},(\ref{ext1})}\label{ext1S}\\
	\NDLineT{\Gamma, x:A}{f'\ x:B}{\ruleRef{eqTyping},(\ref{ext1S})}\label{extfpxTp}\\
	\NDLineT{\Gamma, x:A}{x:A}{\ruleRef{varS},\ruleRef{tpCtx},(\ref{extfxTp})}\label{extXTp}\\
	\NDLineT{\Gamma, x:A}{f:A\to B}{\ruleRef{applType},(\ref{extXTp}),(\ref{extfxTp})}\label{extfTp}\\
	\NDLineT{\Gamma, x:A}{f':A\to B}{\ruleRef{applType},(\ref{extXTp}),(\ref{extfpxTp})}\label{extfpTp}\\
	\NDLineTG{f\termEquals{A\to B}\lambdaFun{x}{A}f\ x}{\ruleRef{eta},(\ref{extfTp})}\label{extEq1a}\\
	\NDLineTG{\lambdaFun{x}{A}f\ x\termEquals{A\to B}\lambdaFun{x}{A}f'\ x}{\ruleRef{congLam}, (\ref{ext1})}\label{extEq2a}\\
	\NDLineTG{f\termEquals{A\to B}\lambdaFun{x}{A}f'\ x}{\ruleRef{trans}, (\ref{extEq1a}), (\ref{extEq2a})}\label{extEq1}\\
	\NDLineTG{f'\termEquals{A\to B}\lambdaFun{x}{A}f'\ x}{\ruleRef{eta},(\ref{extfpTp})}\label{extfpE}\\	
	\NDLineTG{\lambdaFun{x}{A}f'\ x\termEquals{A\to B}f'}{\ruleRef{sym},(\ref{extfpE})}\label{extEq2}\\
	\NDLineTG{f\termEquals{A\to B}f'}{\ruleRef{trans}, (\ref{extEq1}), (\ref{extEq2})}
	\end{align}
	
	\paragraph{Regarding \ruleRef{forallCong}:}
	{\small\begin{align}
		\NDLineTG{A\typeEquals A'}{by assumption}\label{fcass1}
		\intertext{By induction on derivations, it follows that $A, A'$ well-typed}
		\NDLineT{\Gamma, x:A}{F\termEqB F'}{by assumption}\label{fcass2}\\
		\QQNegSp\Gamma\dedT\univQuant{x}{A}F:\bool&&&\text{\QQQQNegSp\ruleRef{eqType},\ruleRef{lambda},definition,\ruleRef{lambda},\ruleRef{eqTyping},(\ref{fcass2})}\label{univTp}\\
		\NDLineT{\concatCtx{\Gamma}{\namedass{uqf}{\univQuant{x}{A}F}}}{\univQuant{x}{A}F}{\ruleRef{assume},\ruleRef{ctxAssume},(\ref{univTp})}\label{fc1}\\
		\NDLineT{\concatCtx{\concatCtx{\Gamma}{\namedass{uqf}{\univQuant{x}{A}F}}}{y:A'}}{\univQuant{x}{A}F}{\ruleRef{varDed},\ruleRef{ctxVar},(\ref{fc1})}\label{fc2}\\		
		\NDLineT{\concatCtx{\concatCtx{\concatCtx{\Gamma}{\namedass{uqf}{\univQuant{x}{A}F}}}{y:A'}}{x':A}}{\univQuant{x}{A}F}{\ruleRef{varDed},\ruleRef{ctxVar},(\ref{fc2})}\label{fcftfp1}\\		
		\NDLineT{\concatCtx{\concatCtx{\concatCtx{\Gamma}{\namedass{uqf}{\univQuant{x}{A}F}}}{y:A'}}{x':A}}{x':A}{\ruleRef{varS},\ruleRef{ctxVar}}\label{fcftfp2}\\
		\NDLineT{\concatCtx{\concatCtx{\concatCtx{\Gamma}{\namedass{uqf}{\univQuant{x}{A}F}}}{y:A'}}{x:A}}{F}{\ruleRef{forallE},(\ref{fcftfp1}),(\ref{fcftfp2})}\label{fcftfp3}\\
		\NDLineT{\concatCtx{\Gamma}{\namedass{uqf}{\univQuant{x}{A}F}}, x:A}{F\termEqB F'}{\ruleRef{assDed},(\ref{fcass2})}\label{fcftfp4}\\
		\NDLineT{\concatCtx{\concatCtx{\concatCtx{\Gamma}{\namedass{uqf}{\univQuant{x}{A}F}}}{y:A'}}{x:A}}{F\termEqB F'}{\ruleRef{varDed},\ruleRef{tpCtx},(\ref{fcftfp2}),(\ref{fcftfp4})}\label{fcftfp5}\\
		\NDLineT{\concatCtx{\concatCtx{\concatCtx{\Gamma}{\namedass{uqf}{\univQuant{x}{A}F}}}{y:A'}}{x:A}}{F'}{\ruleRef{congDed},(\ref{fcftfp5}),(\ref{fcftfp3})}\label{fcftfp6}\\
		\NDLineT{\concatCtx{\concatCtx{\Gamma}{\namedass{uqf}{\univQuant{x}{A}F}}}{y:A'}}{\univQuant{x}{A}F'}{\ruleRef{forallI},(\ref{fcftfp6})}\label{fcftfp7}\\
		\intertext{Now we will prove that $\concatCtx{\concatCtx{\Gamma}{\namedass{uqf}{\univQuant{x}{A}F}}}{y:A'}\dedT y:A$}
		\NDLineT{\concatCtx{\concatCtx{\Gamma}{\namedass{uqf}{\univQuant{x}{A}F}}}{y:A'}}{y:A'}{\ruleRef{varS},\ruleRef{ctxVar}}\label{fc3}\\
		\NDLineT{\concatCtx{\concatCtx{\Gamma}{\namedass{uqf}{\univQuant{x}{A}F}}}{y:A'}}{y\termEquals{A'}y}{\ruleRef{refl},(\ref{fc3})}\label{fc4}\\
		\NDLineT{\concatCtx{\concatCtx{\Gamma}{\namedass{uqf}{\univQuant{x}{A}F}}}{y:A'}}{A\typeEquals A'}{\QQQNegSp\QQNegSp\quad\ruleRef{varDed},\ruleRef{assDed},\ruleRef{ctxAssume},(\ref{univTp}),(\ref{fcass1})}\label{fc5}\\
		\NDLineT{\concatCtx{\concatCtx{\Gamma}{\namedass{uqf}{\univQuant{x}{A}F}}}{y:A'}}{A'\typeEquals A}{\ruleRef{tpEqSym}}\label{fc6}\\
		\NDLineT{\concatCtx{\concatCtx{\Gamma}{\namedass{uqf}{\univQuant{x}{A}F}}}{y:A'}}{y:A}{\ruleRef{congColon},(\ref{fc4}),(\ref{fc6}),(\ref{fc3})}\label{fc7}\\
		\intertext{Substitute $y$ for $x$ in (\ref{fcftfp7}) and move $y$ from context into a $\forall$ binder.}
		\NDLineT{\concatCtx{\concatCtx{\Gamma}{\namedass{uqf}{\univQuant{x}{A}F}}}{y:A'}}{\subst{F'}{x}{y}}{\ruleRef{forallE},(\ref{fcftfp7}),(\ref{fc7})}\label{fc8}\\
		\NDLineT{\concatCtx{\Gamma}{\namedass{uqf}{\univQuant{x}{A}F}}, x:A'}{F'}{renaming $y$ to $x$}\label{fc9}\\
		\NDLineT{\concatCtx{\Gamma}{\namedass{uqf}{\univQuant{x}{A}F}}, x:A'}{x:A'}{\ruleRef{varS},\ruleRef{tpCtx},(\ref{fc7})}\label{fc10}\\
		\NDLineT{\concatCtx{\Gamma}{\namedass{uqf}{\univQuant{x}{A}F}}}{\univQuant{x}{A'}F'}{\ruleRef{forallI},(\ref{fc9}),(\ref{fc10})}\label{fcimpl1}\\
		\intertext{Since term and type equality are both symmetric (and we can use the same trick as above to show that variables of type $A'$ are also of type $A$ and vice versa), we can prove the following formula analogously:}
		\NDLineT{\concatCtx{\Gamma}{ \univQuant{x}{A'}F'}}{\univQuant{x}{A}F}{analogously}\label{fcimpl2}\\
		\NDLineTG{\univQuant{x}{A}F\termEqB\univQuant{x}{A}F'}{\ruleRef{propExt},(\ref{fcimpl1}),(\ref{fcimpl2})}
		\end{align}}

	\paragraph{Regarding \ruleRef{implCong}:}
	\renewcommand{\p}{\lambdaFun{x}{\bool}\lambdaFun{y}{\bool}x\impl y}
	\renewcommand{\pF}{\lambdaFun{y}{\bool}(F\impl y)}
	\newcommand{\pFG}{(F\impl G)}
	\newcommand{\pFp}{\lambdaFun{y}{\bool}(F'\impl y)}
	\newcommand{\pFpGp}{(F'\impl G')}
	
	\begin{align}
	\NDLineTG{F\termEqB F'}{by assumption}\label{rcass1}\\
	\NDLineTG{G\termEqB G'}{by assumption}\label{rcass2}\\
	\intertext{Introduce lambda function in two variables $\lambdaFun{x}{\bool}\lambdaFun{y}{\bool}x\impl y$ show rule for it using \ruleRef{congAppl} and use \ruleRef{beta} and \ruleRef{trans} to conclude the same about $\impl$:}
	\NDLineTG{(\p)&&\nonumber\\&\termEqB (\p)}{\ruleRef{refl},\ruleRef{lambda},\ruleRef{lambda},\ruleRef{implType},\ruleRef{varS},\ruleRef{varS}}\label{rcPeqP}\\
	\NDLineTG{(\p)\ F&&\nonumber\\&\termEqB (\p)\ F'}{\ruleRef{congAppl},(\ref{rcPeqP}),(\ref{rcass1})}\label{rcPFeqPFp}\\
	\NDLineTG{(\p)\ F\ G&&\nonumber\\&\termEqB (\p)\ F'\ G'}{\ruleRef{congAppl},(\ref{rcPFeqPFp}),(\ref{rcass2})}\label{rcPFGeqPFpGp}
	\end{align}
	Now we prove that $\p$ and its applications are well-typed to allow using the rule \ruleRef{beta} for it:
	\begin{align}
	\concatCtx{\Gamma}{&x:\bool}, y:\bool\dedT x:\bool&&\text{\ruleRef{varS},\ruleRef{ctxVar},\ruleRef{ctxVar}}\label{rcxTp}\\
	\concatCtx{\Gamma}{&x:\bool}, y:\bool\dedT y:\bool&&\text{\ruleRef{varS},\ruleRef{ctxVar},\ruleRef{ctxVar}}\label{rcyTp}\\
	\concatCtx{\Gamma}{&x:\bool}, y:\bool\dedT x\impl y:\bool&&\text{\ruleRef{implType},(\ref{rcxTp}),(\ref{rcyTp})}\label{rcXimYTp}\\
	\concatCtx{\Gamma}{&x:\bool}\dedT \lambdaFun{y}{\bool}(x\impl y):\bool\to\bool&&\text{\ruleRef{lambda},(\ref{rcXimYTp})}\label{rclam1Tp}\\
	\NDLineTG{\p:\bool\to\bool\to\bool}{\ruleRef{lambda},(\ref{rclam1Tp})}\label{rcpTp}\\
	\NDLineTG{F:\bool}{\ruleRef{eqTyping},(\ref{rcass1})}\label{rcFTp}\\
	\NDLineTG{G:\bool}{\ruleRef{eqTyping},(\ref{rcass2})}\label{rcGTp}\\
	\NDLineTG{F'\termEqB F}{\ruleRef{sym},(\ref{rcass1})}\label{rcass1S}\\
	\NDLineTG{G'\termEqB G}{\ruleRef{sym},(\ref{rcass2})}\label{rcass2S}\\
	\NDLineTG{F':\bool}{\ruleRef{eqTyping},(\ref{rcass1S})}\label{rcFpTp}\\
	\NDLineTG{G':\bool}{\ruleRef{eqTyping},(\ref{rcass2S})}\label{rcGpTp}\\
	\NDLineTG{\left(\p\right)\ F:\bool\to\bool}{\ruleRef{appl},(\ref{rcpTp}),(\ref{rcFTp})}\label{rcpFTp}\\
	\NDLineTG{\left(\p\right)\ F':\bool\to\bool}{\ruleRef{appl},(\ref{rcpTp}),(\ref{rcFpTp})}\label{rcpFpTp}\\
	\NDLineTG{\left(\p\right)\ F\ G:\bool}{\ruleRef{appl},(\ref{rcpFTp}),(\ref{rcGTp})}\label{rcpFGTp}\\
	\NDLineTG{\left(\p\right)\ F'\ G':\bool}{\ruleRef{appl},(\ref{rcpFpTp}),(\ref{rcGpTp})}\label{rcpFpGpTp}
	\intertext{Now we can use rule \ruleRef{beta} to show that the applications of $\p$ to $F,G$ and $f',G'$ respectively are equal to their \ruleRef{beta} reduced versions:}
	\NDLineTG{(\p)\ F\termEquals{\bool\to\bool} (\pF)}{\ruleRef{beta},(\ref{rcpFTp})}\label{rcpFB}\\
	\NDLineTG{(\p)\ F'\termEquals{\bool\to\bool} (\pFp)}{\ruleRef{beta},(\ref{rcpFpTp})}\label{rcpFpB}\\
	\intertext{And again:}
	\NDLineTG{\bool\to \bool \typeEquals \bool\to \bool}{\ruleRef{congBase}}\label{rcBteB}\\
	\NDLineTG{\pF:\bool\to \bool}{\ruleRef{congColon},(\ref{rcpFB}),(\ref{rcBteB}),(\ref{rcpFTp})}\label{rcpFredTp}\\
	\NDLineTG{\pFp:\bool\to \bool}{\ruleRef{congColon},(\ref{rcpFpB}),(\ref{rcBteB}),(\ref{rcpFpTp})}\label{rcpFpredTp}\\
	\NDLineTG{\pF\ G:\bool}{\ruleRef{appl},(\ref{rcpFredTp}),(\ref{rcGTp})}\label{rcpFBGTp}\\
	\NDLineTG{\pFp\ G':\bool}{\ruleRef{appl},(\ref{rcpFpredTp}),(\ref{rcGpTp})}\label{rcpFpBGpTp}\\
	\NDLineTG{(\pF)\ G\termEqB (\pFG)}{\ruleRef{beta},(\ref{rcpFGTp})}\label{rcpFGB}\\
	\NDLineTG{(\pFp)\ G'\termEqB (\pFpGp)}{\ruleRef{beta},(\ref{rcpFpGpTp})}\label{rcpFpGpB}\\
	\NDLineTG{G\termEqB G}{\ruleRef{refl},(\ref{rcGTp})}\label{rcGeqG}\\
	\NDLineTG{G'\termEqB G'}{\ruleRef{refl},(\ref{rcGpTp})}\label{rcGpeqGp}\\
	\NDLineTG{(\p)\ F\ G\termEqB \pF\ G}{\ruleRef{congAppl},(\ref{rcGeqG}),(\ref{rcpFB})}\label{rcpFBAp}\\
	\NDLineTG{(\p)\ F'\ G'\termEqB \pFp\ G'}{\ruleRef{congAppl},(\ref{rcGpeqGp}),(\ref{rcpFpB})}\label{rcpFpBAp}\\		
	\NDLineTG{(\p)\ F\ G\termEqB \pFG}{\ruleRef{trans},(\ref{rcpFBAp}),(\ref{rcpFGB})}\label{rcpFGBT}\\
	\NDLineTG{(\p)\ F'\ G'\termEqB \pFpGp}{\ruleRef{trans},(\ref{rcpFpBAp}),(\ref{rcpFpGpB})}\label{rcpFpGpBT}\\
	\NDLineTG{\pFG\termEqB (\p)\ F\ G}{\ruleRef{sym},(\ref{rcpFGBT})}\label{rcpFGS}\\
	\NDLineTG{\pFG\termEqB (\p)\ F'\ G'}{\ruleRef{trans},(\ref{rcpFGS}),(\ref{rcPFGeqPFpGp})}\label{rcpFGT}\\
	\NDLineTG{(F\impl F')\termEqB (G\impl G')}{\ruleRef{trans},(\ref{rcpFGT}),(\ref{rcpFpGpBT})}
	\end{align}	
	
	\paragraph{Regarding \ruleRef{forallImpl}:}
	\begin{align}
		\NDLineT{\concatCtx{\Gamma}{x:A}}{F\impl G}{\byAss}\label{forallImpl1}\\
		\NDLineTG{\univQuant{z}{A}\subst{F}{x}{z}\impl \subst{G}{x}{z}}{$\alpha$-renaming,\ruleRef{forallI},(\ref{forallImpl1})}\label{forallImpl2}\\
		\NDLineT{\concatCtx{\concatCtx{\Gamma}{\namedass{suq}{\univQuant{z}{A}\subst{F}{x}{z}}}}{y:A}}{\subst{F}{x}{y}}{\ruleRef{forallE},\ruleRef{assume},\ruleRef{varS}}\label{forallImpl3}\\	
		\NDLineT{\concatCtx{\concatCtx{\Gamma}{\namedass{suq}{\univQuant{z}{A}\subst{F}{x}{z}}}}{y:A}}{\subst{F}{x}{y}\impl \subst{G}{x}{y}}{\ruleRef{forallE},(\ref{forallImpl2}),\ruleRef{varS}}\label{forallImpl4}\\
		\NDLineT{\concatCtx{\concatCtx{\Gamma}{\namedass{suq}{\univQuant{z}{A}\subst{F}{x}{z}}}}{y:A}}{\subst{G}{x}{y}}{\ruleRef{implE},(\ref{forallImpl4}),(\ref{forallImpl3})}\label{forallImpl5}\\
		\NDLineT{\concatCtx{\Gamma}{\namedass{suq}{\univQuant{z}{A}\subst{F}{x}{z}}}}{\univQuant{y}{A}\subst{G}{x}{y}}{\ruleRef{forallI},(\ref{forallImpl5})}\label{forallImpl6}\\
		\NDLineTG{\univQuant{z}{A}\subst{F}{x}{z}\impl \univQuant{y}{A}\subst{G}{x}{y}}{\ruleRef{implI},(\ref{forallImpl6})}\label{forallImpl7}\\
		\NDLineTG{\univQuant{x}{A}F\impl \univQuant{x}{A}G}{$\alpha$-renaming,(\ref{forallImpl7})}\nonumber
	\end{align}
	\paragraph{Regarding \ruleRef{implFunctorial}:}
	\begin{align}
		\NDLineTG{G\impl G'}{\byAss}\label{implFunct1}\\
		\NDLineTG{F'\impl F}{\byAss}\label{implFunct2}\\
		\NDLineT{\concatCtx{\concatCtx{\namedass{as}{\Gamma}{F\impl G}}}{F'}}{F}{\ruleRef{implE},(\ref{implFunct2}),\ruleRef{assume}}\label{implFunct3}\\
		\NDLineT{\concatCtx{\concatCtx{\namedass{as}{\Gamma}{F\impl G}}}{F'}}{G}{\ruleRef{implE},\ruleRef{assume},(\ref{implFunct3})}\label{implFunct4}\\
		\NDLineT{\concatCtx{\concatCtx{\namedass{as}{\Gamma}{F\impl G}}}{F'}}{G'}{\ruleRef{implE},\ruleRef{implFunct1},(\ref{implFunct4})}\label{implFunct5}\\
		\NDLineT{\concatCtx{\namedass{as}{\Gamma}{F\impl G}}}{F'\impl G'}{\ruleRef{implI},(\ref{implFunct5})}\label{implFunct6}\\
		\NDLineTG{\left(F\impl G\right)\impl \left(F'\impl G'\right)}{\ruleRef{implI},(\ref{implFunct6})}\nonumber
	\end{align}
	\paragraph{Regarding \ruleRef{dedCong}:}
	\begin{align}
	\NDLineTG{F\termEqB F'}{by assumption}\label{dcass1}\\
	\NDLineTG{F}{by assumption}\label{dcass2}\\
	\NDLineTG{F'\termEqB F}{\ruleRef{sym},(\ref{dcass1})}\label{dc1S}\\
	\NDLineTG{F'}{\ruleRef{congDed},(\ref{dc1S}),(\ref{dcass2})}
	\end{align}
	
	\paragraph{Regarding \ruleRef{rewrite}:}
	\renewcommand{\p}{\lambdaFun{y}{A} \subst{F}{x}{y}}
	\begin{align}
	\NDLineTG{\subst{F}{x}{t}}{by assumption}\label{rwass3}\\
	\NDLineTG{t\termEquals{A}t'}{by assumtion}\label{rwass1}\\
	\NDLineT{\Gamma, y:A}{\subst{F}{x}{y}:\bool}{by assumption}\label{rwass2}\\
	\NDLineTG{\p:A\to \bool}{\ruleRef{lambda},(\ref{rwass2})}\label{rwPTp}\\
	\NDLineTG{t:A}{\ruleRef{eqTyping},(\ref{rwass1})}\label{rwTTp}\\
	\NDLineTG{(\p)\ t:\bool}{\ruleRef{appl},(\ref{rwPTp}),(\ref{rwTTp})}\label{rwPtTp}\\
	\NDLineTG{(\p)\ t\termEqB \subst{F}{x}{t}}{\ruleRef{beta},(\ref{rwPtTp})}\label{rwPtEqB}\\
	\NDLineTG{(\p)\ t}{\ruleRef{congDed},(\ref{rwPtEqB}),(\ref{rwass3})}\label{rwPt}\\
	\NDLineTG{\p\termEquals{A\to\bool}\p}{\ruleRef{refl},(\ref{rwPTp})}\label{rwpEqp}\\
	\NDLineTG{(\p)\ t\termEqB (\p)\ t'}{\ruleRef{congAppl},(\ref{rwass1}),(\ref{rwpEqp})}\label{rwPtEqPtp}\\
	\NDLineTG{(\p)\ t'}{\ruleRef{dedCong},(\ref{rwPtEqPtp}),(\ref{rwPt})}\label{rwPtp}\\
	\NDLineTG{t':A}{\ruleRef{eqTyping},\ruleRef{sym},(\ref{rwass1})}\label{rwTpTp}\\
	\NDLineTG{(\p)\ t':\bool}{\ruleRef{appl},(\ref{rwPTp}),(\ref{rwTpTp})}\label{rwPtpTp}\\
	\NDLineTG{(\p)\ t'\termEqB \subst{F}{x}{t'}}{\ruleRef{beta},(\ref{rwPtpTp})}\label{rwPTpEqPTp}\\
	\NDLineTG{\subst{F}{x}{t'}}{\ruleRef{dedCong},(\ref{rwPTpEqPTp}),(\ref{rwPtp})}
	\end{align}
\end{proof}

\begin{lemma}\label{lem:eqTrue}
	If $F$ is a well-typed \hol{} formula we can show $\Gamma\dedT F\termEquals{\bool}\T$ iff $\Gamma\dedT F$.
\end{lemma}

\begin{proof}[Proof of Lemma~\ref{lem:eqTrue}]
	We start with the $\impl$ direction:
	\begin{align}
	\NDLineTG{\lambdaFun{x}{\bool}x\termEquals{\bool\to\bool}\lambdaFun{x}{\bool}x}{\ruleRef{refl},\ruleRef{lambda},\ruleRef{varS}}\label{let2}\\
	\NDLineTG{\T}{definition of $\T$,(\ref{let2})}\label{let3}\\
	\NDLineTG{F\termEqB \T}{\byAss}\label{let1}\\
	\NDLineTG{F}{\ruleRef{congDed},(\ref{let1}),(\ref{let3})}
	\end{align}
	
	Now the $\Leftarrow$ direction:
	\begin{align}
	\NDLineTG{F}{\byAss}\label{let5}\\
	\NDLineTG{F:\bool}{\byAss}\label{let8}\\
	\NDLineTG{\T:\bool}{by definition}\label{let9}\\
	\NDLineT{\concatCtx{\Gamma}{\T}}{\namedass{ass}{F}}{\ruleRef{assDed},(\ref{let5})}\label{let6}\\		
	\NDLineT{\concatCtx{\Gamma}{\namedass{ass}{F}}}{\T}{\ruleRef{assDed},(\ref{let3})}\label{let7}\\
	\NDLineTG{F\termEqB \T}{\ruleRef{propExt},(\ref{let8}),(\ref{let9}),(\ref{let7}),(\ref{let6})}
	\end{align}
\end{proof}

Finally, using the definitions of the connectives and quantifiers we can prove the rules:
\[
\rnamedRule{$\land$I}{\Gamma\dedT F\land G}{\Gamma\dedT F\tb \Gamma\dedT G}{andI}\tb
\rnamedRule{$\lor$Il}{\Gamma\dedT F\lor G}{\Gamma\dedT F}{orIl}\tb 
\rnamedRule{$\lor$Ir}{\Gamma\dedT F\lor G}{\Gamma\dedT G}{orIr}\]
\[
\rnamedRule{$\land$El}{\Gamma\dedT F}{\Gamma\dedT F\land G}{andEl}\tb
\rnamedRule{$\land$Er}{\Gamma\dedT G}{\Gamma\dedT F\land G}{andEr}\tb
\rnamedRule{$\exists$I}{\Gamma\dedT\existQuant{x}{A} F}{\Gamma\dedT t:A\tb \Gamma\dedT \subst{F}{x}{t}}{existI}
\]

\begin{rem}\label{rem:analogueDerivedRulesDHOLP}
	Observe that many of the rules derived for \hol{} in Lemma~\ref{meta-thm:1} still hold in \dphol{}. 
	In particular, the rules \ruleRef{ctxThy}, \ruleRef{tpCtx}, \ruleRef{typingTp} and \ruleRef{validTyping} can be proven by the same method. 
	Also the rules \ruleRef{tpEqRefl} and \ruleRef{tpEqSym} can be proven easily in \dphol{} by induction on the type equality rules (for the rule \ruleRef{tpEqSym} the proof easily generalizes to \dhol{} but proving it in \dphol{} is harder). 
	Finally the rules \ruleRef{assDed}, \ruleRef{varDed}, \ruleRef{rewrite} and the introduction and elimination rules for the quantifier and defined logical connectives can be derived in \dphol{} with the same proofs.
\end{rem}

\section{Completeness proof}\label{sec:complete}
We will prove the following slightly stronger version of the theorem:
\begin{theorem}[Completeness]\label{thm:complete}
	We have 
	\begin{align}
	\phantom{\Gamma}&\ded \Thy{T}  &&\text{ implies } \ded\Thy{\PhiAppl T} \label{correct:theoremhood}\\
	\phantom{\Gamma}&\dedT \Ctx{\Gamma} &&\text{ implies } \dedPT\Ctx{\PhiAppl{\Gamma}} \label{correct:contexthood}\\
	\Gamma&\dedT \Type{A} &&\text{ implies } \PhiAppl{\Gamma}\dedPT\Type{\PhiAppl{A}}\text{ and }\ \PhiAppl{\Gamma}\dedPT\PredPhiName{A}: \PhiAppl{A}\to\PhiAppl{A}\to \bool \label{correct:typehood}\\
	\Gamma&\dedT A \typeEquals B &&\text{ implies } \PhiAppl{\Gamma}\dedPT\PhiAppl{A} \typeEquals \PhiAppl{B}\text{ and }\concatCtx{\PhiAppl{\Gamma}}{x:\PhiAppl{A}}\dedPT\PredPhi{A}{x} \termEqB \PredPhi{B}{x}\label{correct:typeEq}\\
	\Gamma&\dedT A \subtyping B&&\text{ implies } \PhiAppl{\Gamma}\dedPT \PhiAppl{A} \typeEquals \PhiAppl{B}\text{ and }\concatCtx{\PhiAppl{\Gamma}}{x,y:\PhiAppl{B}}\dedPT \termEqT{A}{x}{y}\impl \termEqT{B}{x}{y}\label{correct:subtyping}\\
	\Gamma&\dedT t:A &&\text{ implies } \PhiAppl{\Gamma}\dedPT\PhiAppl{t}:\PhiAppl{A} \ \text{and}\  \PhiAppl{\Gamma}\dedPT\PredPhi{A}{\PhiAppl{t}}\label{correct:typing}\\
	\Gamma&\dedT F &&\text{ implies } \PhiAppl{\Gamma}\dedPT\PhiAppl{F}\label{correct:validity}
	\intertext{In case of term equality, we strengthen the claim to:}
	\Gamma&\dedT t \termEquals{A} t' &&\text{ implies } \PhiAppl{\Gamma}\dedPT\termEqT{A}{\PhiAppl{t}}{\PhiAppl{t'}}\ \text{and}\ \PhiAppl{\Gamma}\dedPT \PhiAppl{t}:\PhiAppl{A}\ \text{and}\ \PhiAppl{\Gamma}\dedPT \PhiAppl{t'}:\PhiAppl{A}\label{correct:termEq}
	\end{align}
	
	Furthermore, the typing relations $\PredPhiName{A}$ are symmetric and transitive on all well-formed types $A$:
	\begin{align}
	\Gamma\dedT \Type{A} &\Impl \PhiAppl{\Gamma}\dedPT 
	\univQuant{x,y}{\PhiAppl{A}}
	\termEqT{A}{x}{y}\impl \termEqT{A}{y}{x}
	\label{correct:relatSym}\\
	\Gamma\dedT \Type{A} &\Impl \PhiAppl{\Gamma}\dedPT 
	\univQuant{x,y,z}{\PhiAppl{A}}\termEqT{A}{x}{y}\impl \left(\termEqT{A}{y}{z}\impl \termEqT{A}{x}{z}\right)\label{correct:relatTrans}
	\end{align}
	
	Additionally the substitution lemma holds, i.e.,
	\begin{align}
	\concatCtx{\Gamma}{x:A}\dedT t:B\Mand\Gamma\ded u:A &\Impl \PhiAppl{\Gamma}\dedPT\PhiAppl{\subst{t}{x}{u}}\termEquals{\PhiAppl{B}} \subst{\PhiAppl{t}}{x}{\PhiAppl{u}}\label{correct:substTerm}\\
	\concatCtx{\Gamma}{x:A}\dedT B\ \type\Mand\Gamma\dedT u:B &\Impl \PhiAppl{\Gamma}\dedPT\PhiAppl{\subst{B}{x}{u}}\typeEquals \subst{\PhiAppl{B}}{x}{\PhiAppl{u}}\label{correct:substType}
	\end{align}
	In the following lines, we assume that if $t=\lambdaFun{y}{C}s$ for $s$ of type $D$, then $B=\piType{y}{C}D$ (this is enough in practise and we cannot easily show more).
	\begin{align}
	\concatCtx{\Gamma}{x:A}\dedT t:B &\Impl \concatCtx{\concatCtx{\PhiAppl{\Gamma}}{x,x':\PhiAppl{A}}}{\namedass{xRx'}{\termEqT{A}{x}{x'}}} \dedPT \termEqT{B}{\PhiAppl{t}}{\subst{\PhiAppl{t}}{x}{x'}}\label{correct:substRelatTerms}
	\end{align}
\end{theorem}

Here Case~\ref{correct:typeEq} looks weaker than in the original statement, but is easily seen to be equivalent.
The proof uses rule \ruleRef{subtI}, the completeness claim for the subtyping judgement and rule \ruleRef{propExt} in \hol{}. 

\renewcommand{\subparagraph}[1]{\paragraph{#1}}
\begin{proof}[Proof of Theorem~\ref{thm:complete}]
	Firstly, we will prove the substitution lemma by induction on the grammar, i.e. by induction on the shape of the terms and types.
	
	Afterwards, we will prove completeness of the translation w.r.t. all \dphol{} judgements by induction on the derivations. This means that we consider the inference rules of \dphol{} and prove that if completeness holds for the assumptions of a \dphol{} inference rule, then it also holds for the conclusion of the rule. 
	For the inductive steps for some typing rules, namely \ruleRef{eqType} and \ruleRef{congColon}, we also require the fact that for any (well-formed) type $A$ in \dphol{} we have $\PredPhiName{A}:\PhiAppl{A}\to\PhiAppl{A}\to\bool$. 
	This follows directly from how the $\PredPhiName{A}$ are generated/defined in the translation.
	
	\paragraph{Substitution lemma and symmetry and transitivity of the typing relations}
	Since the translation of types commutes with the type productions of the grammar (\ref{correct:substType}) is obvious.
	
	We show (\ref{correct:substTerm}) by induction on the grammar of \dphol{}. 
	If $x$ is not a free variable in $t$, then $\PhiAppl{\subst{t}{x}{u}}=\PhiAppl{t}=\subst{\PhiAppl{t}}{x}{\PhiAppl{u}}$ and the claim (\ref{correct:substTerm}) follows by rule \ruleRef{refl}. 
	So assume that $x$ is a free variable of $t$.
	
	If $t$ is a variable, then by assumption (that $x$ is a free variable in $t$) it follows that $t=x$ and thus $\PhiAppl{\subst{t}{x}{u}}=\PhiAppl{u}=\subst{\PhiAppl{t}}{x}{\PhiAppl{u}}$ and the claim follows by rule \ruleRef{refl}. 
	
	If $t$ is a $\lambda$-term $\lambdaFun{y}{A} s$, then by induction hypothesis we have $\concatCtx{\PhiAppl{\Gamma}}{y:\PhiAppl{A}}\dedPT \PhiAppl{\subst{s}{x}{u}}\termEquals{\PhiAppl{A}}\subst{\PhiAppl{s}}{x}{\PhiAppl{u}}$, where $A$ is the type of $s$. By rule \ruleRef{congLam}, the claim of $\PhiAppl{\Gamma}\dedPT \PhiAppl{\subst{\lambdaFun{y}{A} s}{x}{u}}\termEquals{\PhiAppl{B}}\subst{\PhiAppl{\lambdaFun{y}{A} s}}{x}{\PhiAppl{u}}$ follows.
	
	If $t$ is a function application $f\ s$, then by induction hypothesis we have $\PhiAppl{\Gamma}\dedPT \PhiAppl{\subst{s}{x}{u}}\termEquals{\PhiAppl{A}}\subst{\PhiAppl{s}}{x}{\PhiAppl{u}}$ and $\PhiAppl{\Gamma}\dedPT \PhiAppl{\subst{f}{x}{u}}\termEquals{\PhiAppl{A}\to\PhiAppl{B}}\subst{\PhiAppl{f}}{x}{\PhiAppl{u}}$, where $T$ is the type of $s$. By rule \ruleRef{congAppl}, the claim of $\PhiAppl{\Gamma}\dedPT \PhiAppl{\subst{\left(f\ s\right)}{x}{u}}\termEquals{\PhiAppl{B}}\subst{\PhiAppl{f\ s}}{x}{\PhiAppl{u}}$ follows.
	
	If $t$ is an equality $s\termEquals{A}s'$, then by induction hypothesis we have $\PhiAppl{\Gamma}\dedPT \PhiAppl{\subst{s}{x}{u}}\termEquals{\PhiAppl{A}}\subst{\PhiAppl{s}}{x}{\PhiAppl{u}}$ and $\PhiAppl{\Gamma}\dedPT \PhiAppl{\subst{s'}{x}{u}}\termEquals{\PhiAppl{A}}\subst{\PhiAppl{s'}}{x}{\PhiAppl{u}}$, where $A$ is the type of $s$ and $s'$. 
	By rule \ruleRef{termEqCong}, the claim of $\PhiAppl{\Gamma}\dedPT \PhiAppl{\subst{\left(s\termEquals{A}s'\right)}{x}{u}}\termEqB\subst{\left(\PhiAppl{s\termEquals{A}s'}\right)}{x}{\PhiAppl{u}}$ follows.
	
	Before we can show (\ref{correct:substRelatTerms}), we first need to prove the symmetry and transitivity of the typing relations:
	We can prove both by induction on the type $A$. 
	Denote $\Delta:=\concatCtx{\concatCtx{\PhiAppl{\Gamma}}{x,y:\PhiAppl{A}}}{\namedass{xRy}{\termEqT{A}{x}{y}}}$ and $\Theta:=\concatCtx{\concatCtx{\concatCtx{\PhiAppl{\Gamma}}{x,y,z:\PhiAppl{A}}}{\namedass{xRy}{\termEqT{A}{x}{y}}}}{\namedass{yRz}{\termEqT{A}{y}{z}}}$ respectively. 
	If we can show $\Delta\dedPT \termEqT{A}{y}{x}$ and $\Theta\dedPT \termEqT{A}{x}{z}$ respectively, then the claims (\ref{correct:relatSym}) and (\ref{correct:relatTrans}) follows by the rules \ruleRef{implI}, \ruleRef{forallI}, \ruleRef{varS} and \ruleRef{assume}.
	Those are therefore the claims we are going to show.
	
	Observe that for types declared in the theory $T$, the symmetry and transitivity of $\PredPhiName{A}$ follows from the axiom generated by the translation (in case (\ref{PTTpConstr})) of the type declaration declaring $A$.
	This follows from the symmetry and transitivity of equality and (\ref{correct:substTerm}). 
		
	If $A$ is $\bool$, the typing relation is $\termEqB$ which is symmetric and transitive by the rules \ruleRef{sym} and \ruleRef{trans} respectively. In these cases the claims follows by the rule \ruleRef{assume} and rule \ruleRef{sym} resp. by rule \ruleRef{assume} and rule \ruleRef{trans}.
	
	If $A$ is a $\Pi$-type $\piType{x}{C}D$ we have $\termEqT{C}{f}{g}=\univQuant{x,y}{\PhiAppl{C}}\termEqT{C}{x}{y}\impl \termEqT{D}{f\ x}{g\ y}$.
	Then we have $$\Delta=\concatCtx{\concatCtx{\PhiAppl{\Gamma}}{x,y:\PhiAppl{A}}}{\namedass{xRy}{\univQuant{w}{\PhiAppl{C}}\univQuant{w'}{\PhiAppl{C}}\termEqT{C}{w}{w'}\impl \termEqT{D}{\left(x\ w\right)}{\left(y\ w'\right)}}}$$ and \begin{align*}
	\Theta=&\concatCtx{\concatCtx{\concatCtx{\PhiAppl{\Gamma}}{x,y,z:\PhiAppl{A}}}{\namedass{xRy}{\univQuant{w}{\PhiAppl{C}}\univQuant{w'}{\PhiAppl{C}}\termEqT{C}{w}{w'}\impl \termEqT{D}{\left(x\ z\right)}{\left(y\ z'\right)}}}}{\\&\qquad\namedass{yRz}{\univQuant{w}{\PhiAppl{C}}\univQuant{w'}{\PhiAppl{C}}\termEqT{C}{w}{w'}\impl \termEqT{D}{\left(y\ w\right)}{\left(z\ w'\right)}}}.
	\end{align*}
	The claim is \[
	\Delta\dedPT \univQuant{w,w'}{\PhiAppl{C}}\termEqT{C}{w}{w'}\impl \termEqT{D}{\left(y\ w\right)}{\left(x\ w'\right)}
	\] and \[
	\Delta\dedPT \univQuant{w,w'}{\PhiAppl{C}}\termEqT{C}{w}{w'}\impl \termEqT{D}{\left(x\ w\right)}{\left(z\ w'\right)}
	\] respectively. 
	
	We can prove the claim for (\ref{correct:relatSym}) by 
	{\footnotesize
	\begin{align}
		\NDLinePT{\concatCtx{\concatCtx{\Delta}{w,w':\PhiAppl{C}}}{\namedass{wRw'}{\termEqT{C}{w}{w'}}}}{\nonumber\\&\QQQQNegSp\termEqT{C}{w}{w'}\impl\termEqT{D}{(x\ w)}{(y\ w')}}{\ruleRef{forallE},\ruleRef{forallE},\ruleRef{assume}}\label{relatSym1}\\
		\NDLinePT{\concatCtx{\concatCtx{\Delta}{w,w':\PhiAppl{C}}}{\namedass{wRw'}{\termEqT{C}{w}{w'}}}}{\nonumber\\&\QQQQNegSp\termEqT{D}{(x\ w)}{(y\ w')}}{\ruleRef{implE},(\ref{relatSym1}),\ruleRef{assume}}\label{relatSym2}\\
		\NDLinePT{\concatCtx{\concatCtx{\Delta}{w,w':\PhiAppl{C}}}{\namedass{wRw'}{\termEqT{C}{w}{w'}}}}{\nonumber\\&\QQQQNegSp\termEqT{D}{(x\ w)}{(y\ w')}\impl\termEqT{D}{(y\ w)}{(x\ w')}}{\IH}\label{relatSym3pre}\\
		\NDLinePT{\concatCtx{\concatCtx{\Delta}{w,w':\PhiAppl{C}}}{\namedass{wRw'}{\termEqT{C}{w}{w'}}}}{\nonumber\\&\QQQQNegSp\termEqT{D}{(x\ w)}{(y\ w')}}{\ruleRef{implE},(\ref{relatSym3pre}),(\ref{relatSym2})}\label{relatSym3}
		\end{align}
		\begin{align}
		\NDLinePT{\concatCtx{\Delta}{w,w':\PhiAppl{C}}}{\termEqT{C}{w}{w'}\impl \termEqT{D}{(y\ w)}{(x\ w')}}{\ruleRef{implI},(\ref{relatSym3})}\label{relatSym4}\\
		\NDLinePTD{\univQuant{w,w'}{\PhiAppl{C}}\termEqT{C}{w}{w'}\impl \termEqT{D}{(y\ w)}{(x\ w')}}{\ruleRef{forallI},\ruleRef{forallI},(\ref{relatSym4})}\nonumber
	\end{align}
	}
	
	We can prove the claim for (\ref{correct:relatTrans}) similarly. For this denote $\Lambda:=\concatCtx{\concatCtx{\Theta}{w,w':\PhiAppl{C}}}{\namedass{wRw'}{\termEqT{C}{w}{w'}}}$. 
	\begin{align}
	\NDLinePT{\Lambda}{\termEqT{C}{w}{w'}\impl\termEqT{D}{(x\ w)}{(y\ w')}}{\ruleRef{forallE},\ruleRef{forallE},\ruleRef{assume}}\label{relatTrans1}\\
	\NDLinePT{\Lambda}{\termEqT{C}{w}{w'}\impl\termEqT{D}{(y\ w)}{(z\ w')}}{\ruleRef{forallE},\ruleRef{forallE},\ruleRef{assume}}\label{relatTrans2}\\
	\NDLinePT{\Lambda}{\termEqT{D}{(x\ w)}{(y\ w')}}{\ruleRef{implE},(\ref{relatTrans1}),\ruleRef{assume}}\label{relatTrans3}\\
	\NDLinePT{\Lambda}{\termEqT{D}{(y\ w)}{(z\ w')}}{\ruleRef{implE},(\ref{relatTrans2}),\ruleRef{assume}}\label{relatTrans4}\\
	\NDLinePT{\Lambda}{\termEqT{D}{(x\ w)}{(y\ w')}\impl\big(\termEqT{D}{(y\ w)}{(z\ w')}\nonumber\\&\impl \termEqT{D}{(x\ w)}{(z\ w')}\big)}{\IH}\label{relatTransIH}\\
	\NDLinePT{\Lambda}{\termEqT{D}{(y\ w)}{(z\ w')}\impl \termEqT{D}{(x\ w)}{(z\ w')}}{\ruleRef{implE},(\ref{relatTransIH}),(\ref{relatTrans3})}\label{relatTrans6}\\
	\NDLinePT{\Lambda}{\termEqT{D}{(x\ w)}{(z\ w')}}{\ruleRef{implE},(\ref{relatTrans6}),(\ref{relatTrans4})}\label{relatTrans7}
	\end{align}
	\begin{align}
		\concatCtx{\Theta}{w,w':\PhiAppl{C}}\dedPT\QQQNegSp&\QQQuad \termEqT{C}{w}{w'}\impl \termEqT{D}{(x\ w)}{(z\ w')}&&\text{\ruleRef{implE},(\ref{relatTrans7}),\ruleRef{assume}}\label{relatTrans8}\\
		\NDLinePTT{\univQuant{w,w'}{\PhiAppl{C}} \termEqT{C}{w}{w'}\impl \termEqT{D}{(x\ w)}{(z\ w')}}{\ruleRef{forallI},\ruleRef{forallI},(\ref{relatTrans8})}\label{relatTrans9}\nonumber
	\end{align}
	
	It remains to consider the case of $A=\subtype{B}{p}$. 
	In this case, the claim is
	$\Delta\dedPT \termEqT{B}{y}{x}\land \PhiAppl{p}\ y\land \PhiAppl{p}\ x$ respectively $\Theta\dedPT \termEqT{B}{x}{z}\land \PhiAppl{p}\ x\land \PhiAppl{p}\ z$.
	Applying the induction hypothesis for type $B$ yields 
	$\Delta\dedPT \termEqT{B}{y}{x}$ respectively $\Theta\dedPT \termEqT{B}{x}{z}$. 
	So it remains to show that $\Delta\dedPT \PhiAppl{p}\ y\land \PhiAppl{p}\ x$ and $\Theta\dedPT \PhiAppl{p}\ x\land \PhiAppl{p}\ z$ respectively hold.
	We can show them using rule \ruleRef{andI} given $\Delta\dedPT \PhiAppl{p}\ y$ and $\Delta\dedPT \PhiAppl{p}\ x$ respectively $\Theta\dedPT \PhiAppl{p}\ x$ and $\Theta\dedPT \PhiAppl{p}\ x$. 
	Those statements follow from rule \ruleRef{assume} and the elimination rules of $\land$. 
	
	This concludes the proof of (\ref{correct:relatSym}) and (\ref{correct:relatTrans}).
	
	We show (\ref{correct:substRelatTerms}) by induction on the grammar:
	Without loss of generality we may assume that $B=:\subtype{B'}{p}$ for $B'$ either a base or a $\Pi$-type. This is due to the fact that base types and $\Pi$-types $B'$ can be written as $\subtype{B'}{\lambdaFun{x}{B'}\T}$ and types of the form $\subtype{\subtype{B''}{p}}{q}$ can be rewritten as $\subtype{B''}{\lambdaFun{x}{B''}p\ x\land q\ x}$.
	
	If $t$ is a constant or variable then $\subst{\PhiAppl{t}}{x}{x'}=\PhiAppl{t}$ and by case (\ref{PTctxVar}) resp. by case (\ref{PTax}) in the definition of the translation, we have $\termEqT{A}{\PhiAppl{t}}{\PhiAppl{t}}$. So the claim holds.
	
	If $t$ is a $\lambda$-term $\lambdaFun{y}{C}s$ and $B'=\piType{z}{C}D$, then by induction hypothesis we have \[\concatCtx{\concatCtx{\PhiAppl{\Gamma}}{x,x':\PhiAppl{A}}}{\namedass{xRx'}{\termEqT{A}{x}{x'}}} \dedPT \termEqT{D}{\PhiAppl{s}}{\subst{\PhiAppl{s}}{x}{x'}}.\]
	By the rules \ruleRef{forallI}, \ruleRef{implI}, we yield \[\PhiAppl{\Gamma}\dedPT \univQuant{x,y}{\PhiAppl{A}}\termEqT{A}{x}{y}\impl\termEqT{D}{\PhiAppl{s}}{\subst{\PhiAppl{s}}{x}{x'}}. \]
	By definition (\ref{PTPipred}) this is exactly \[
	\concatCtx{\concatCtx{\PhiAppl{\Gamma}}{x,x':\PhiAppl{A}}}{\namedass{xRx'}{\termEqT{A}{x}{x'}}} \dedPT \termEqT{B'}{\PhiAppl{t}}{\subst{\PhiAppl{t}}{x}{x'}}.\]
	Since $t$ is a $\lambda$-term, by assumption we have that $B\typeEquals B'=\subtype{B}{\lambdaFun{z}{B}\T}$, so the claim follows trivially.
	
	If $t$ is a function application $f\ s$ with $f$ of type $\piType{z}{C}D$ and $s$ of type $C$, then by assumption $B=D\typeEquals B'=\subtype{B}{\lambdaFun{z}{B}\T}$, so it suffices to prove that \[ \concatCtx{\concatCtx{\PhiAppl{\Gamma}}{x,x':\PhiAppl{A}}}{\namedass{xRx'}{\termEqT{A}{x}{x'}}} \dedPT \termEqT{D}{\PhiAppl{f\ s}}{\subst{\PhiAppl{f\ s}}{x}{x'}}.\]
	By induction hypothesis and (\ref{correct:substTerm}) we then have:
	\[
	\concatCtx{\concatCtx{\PhiAppl{\Gamma}}{x,x':\PhiAppl{A}}}{\namedass{xRx'}{\termEqT{A}{x}{x'}}} \dedPT \termEqT{\left(\piType{z}{C}D\right)}{\PhiAppl{f}}{\subst{\PhiAppl{f}}{x}{x'}}
	\] and
	\begin{equation}
	\concatCtx{\concatCtx{\PhiAppl{\Gamma}}{x,x':\PhiAppl{A}}}{\namedass{xRx'}{\termEqT{A}{x}{x'}}} \dedPT \termEqT{C}{\PhiAppl{s}}{\subst{\PhiAppl{s}}{x}{x'}}\label{substLemAppl2}
	.\end{equation}
	By definition (\ref{PTPipred}), we can unpack the former to:
	\begin{equation}
	\concatCtx{\concatCtx{\PhiAppl{\Gamma}}{x,x':\PhiAppl{A}}}{\namedass{xRx'}{\termEqT{A}{x}{x'}}} \dedPT
	\univQuant{z,z'}{\PhiAppl{C}}\termEqT{C}{z}{z'}\impl \termEqT{\left(\piType{z}{C}D\right)}{\PhiAppl{f}\ z}{\subst{\PhiAppl{f}}{x}{x'}\ \subst{z'}{x}{x'}}\label{substLemAppl1}
	\end{equation}
	Using the rules \ruleRef{forallE} and \ruleRef{implE} (using (\ref{substLemAppl1})) to plug in $\PhiAppl{s}$ resp. $\subst{\PhiAppl{s}}{x}{x'}$ for $z, z'$ in (\ref{substLemAppl1}), we yield:
	\[
	\concatCtx{\concatCtx{\PhiAppl{\Gamma}}{x,x':\PhiAppl{A}}}{\namedass{xRx'}{\termEqT{A}{x}{x'}}} \dedPT \termEqT{\left(\piType{z}{C}D\right)}{\PhiAppl{f}\ \PhiAppl{s}}{\subst{\PhiAppl{f}}{x}{x'}\ \subst{\PhiAppl{s}}{x}{x'}}
	\] which is exactly the desired result.
	
	By definition (\ref{PTTpBoolPred}), the typing relation for type $\bool$ is ordinary equality, so the cases of $t$ being an implication or Boolean equality are in fact special cases of (\ref{correct:substTerm}), which is already proven above. 
	It remains to consider the case of $t$ being an equality $s\termEquals{C}s'$ for $C\not\typeEquals\bool$. 
	In this case, the induction hypothesis implies that 
	\begin{equation}
	\concatCtx{\concatCtx{\PhiAppl{\Gamma}}{x,x':\PhiAppl{A}}}{\namedass{xRx'}{\termEqT{A}{x}{x'}}} \dedPT \termEqT{C}{\PhiAppl{s}}{\subst{\PhiAppl{s}}{x}{x'}}\label{substLemEq1}
	\end{equation} and \begin{equation}
	\concatCtx{\concatCtx{\PhiAppl{\Gamma}}{x,x':\PhiAppl{A}}}{\namedass{xRx'}{\termEqT{A}{x}{x'}}} \dedPT \termEqT{C}{\PhiAppl{s'}}{\subst{\PhiAppl{s'}}{x}{x'}}\label{substLemEq2}
	\end{equation}
	We need to prove
	\[
	\concatCtx{\concatCtx{\PhiAppl{\Gamma}}{x,x':\PhiAppl{A}}}{\namedass{xRx'}{\termEqT{A}{x}{x'}}} \dedPT
	\termEqT{C}{\PhiAppl{s}}{\PhiAppl{s'}} \termEqB \termEqT{C}{\subst{\PhiAppl{s}}{x}{x'}}{\subst{\PhiAppl{s'}}{x}{x'}}.\]
	
	If we can show  
	\[\concatCtx{\concatCtx{\concatCtx{\PhiAppl{\Gamma}}{x,x':\PhiAppl{A}}}{\namedass{xRx'}{\termEqT{A}{x}{x'}}}}{\namedass{sRs'}{\termEqT{C}{\PhiAppl{s}}{\PhiAppl{s'}}}} \dedPT \termEqT{C}{\subst{\PhiAppl{s}}{x}{x'}}{\subst{\PhiAppl{s'}}{x}{x'}}\] and similarly also 
	\[\concatCtx{\concatCtx{\concatCtx{\PhiAppl{\Gamma}}{x,x':\PhiAppl{A}}}{\namedass{xRx'}{\termEqT{A}{x}{x'}}}}{\namedass{subrel}{\termEqT{C}{\subst{\PhiAppl{s}}{x}{x'}}{\subst{\PhiAppl{s'}}{x}{x'}}}} \dedPT \termEqT{C}{\PhiAppl{s}}{\PhiAppl{s'}},\]
	then the claim follows by rule \ruleRef{propExt}.
	
	Both follows from the transitivity (\ref{correct:relatTrans}) of the typing relation $\PredPhiName{C}$.

	\paragraph{Well-formedness of theories}
	Well-formedness of \dphol{} theories can be shown using the rules 
	\ruleRef{thyEmpty}, \ruleRef{thyType'}, \ruleRef{thyConst} and \ruleRef{thyAxiom}:

	\subparagraph{\ruleRef{thyEmpty}:}
	\begin{align}
		\NDLine{}{\Thy{\emptyThy}}{\ruleRef{thyEmpty}}\label{thyEmpty1}\\
		\NDLineH{}{\Thy{\PhiAppl{\emptyThy}}}{\ruleRef{thyEmpty}}\nonumber
	\end{align}
	
	\subparagraph{\ruleRef{thyType'}:}
	\begin{align}
	\NDLineT{}{\Ctx{x_1:A_1,\ldots,x_n:A_n}}{by assumption}\label{thyType1}\\
	\NDLinePT{}{\Ctx{x_1:\PhiAppl{A_1},\PredPhi{A_1}{x_1},\ldots,x_n:\PhiAppl{A_n},\typingAss{A_n}{x_n}}}{\IH,(\ref{thyType1})}\label{thyType1a}\\
	\NDLineH{}{\Thy{\PhiAppl{T}}}{\ruleRef{ctxThy},(\ref{thyType1a})}\label{thyType2}\\
	\NDLineH{}{\Thy{\concatThy{\PhiAppl{T}}{a:\type}}}{\ruleRef{thyType},(\ref{thyType2})}\label{thyType3}\\
	\NDLineH{}{
		\Thy{\PhiAppl{\concatThy{T}{a:\piType{x_1}{A_1}\ldots\piType{x_n}{A_n}\type}}}}{\ref{PTTpConstr},(\ref{thyType3})}\nonumber
	\end{align}
	
	\subparagraph{\ruleRef{thyConst}:}
	\begin{align}
	\NDLineT{}{A\;\type}{by assumption}\label{thyConst2}\\
	\NDLinePT{}{\PhiAppl{A}\:\type}{\IH,(\ref{thyConst2})}\label{thyConst4}\\
	\NDLineH{}{\Thy{\concatThy{\PhiAppl{T}}{c:\PhiAppl{A}}}}{\ruleRef{thyConst},(\ref{thyConst4})}\label{thyConst5}\\
	\NDLineH{}{\Thy{\PhiAppl{\concatThy{T}{c:A}}}}{\ref{PTtermDecl},(\ref{thyConst5})}
	\end{align}
	
	\subparagraph{\ruleRef{thyAxiom}:}
	\begin{align}
	\NDLineT{}{F:\bool}{by assumption}\label{thyAxiom1}\\
	\NDLinePT{}{\PhiAppl{F}:\bool}{\IH,(\ref{thyAxiom1})}\label{thyAxiom2}\\
	\NDLineH{}{\Thy{\concatThy{\PhiAppl{T}}{\namedass{ass}{\PhiAppl{F}}}}}{\ruleRef{thyAxiom},(\ref{thyAxiom2})}\label{thyAxiom3}\\
	\NDLineH{}{\Thy{\PhiAppl{\concatThy{T}{\namedax{ax}{F}}}}}{\ref{PTax},(\ref{thyAxiom3})}
	\end{align}
	
	\paragraph{Well-formedness of contexts}
	Well-formedness of contexts can be concluded using the rules \ruleRef{ctxEmpty}, \ruleRef{ctxVar} and \ruleRef{ctxAssume}:
	\subparagraph{\ruleRef{ctxEmpty}:}
	\begin{align}
	\NDLine{}{\Thy{T}}{by assumption}\label{ctxEmpty1}\\
	\NDLineH{}{\Thy{\PhiAppl{T}}}{\IH,(\ref{ctxEmpty1})}\label{ctxEmpty2}\\
	\NDLinePT{}{\Ctx{\emptyCtx}}{\ruleRef{ctxEmpty},(\ref{ctxEmpty2})}\label{ctxEmpty3}\\
	\NDLinePT{}{\Ctx{\PhiAppl{\emptyCtx}}}{\ref{PTemptyCtx},(\ref{ctxEmpty3})}
	\end{align}
	
	\subparagraph{\ruleRef{ctxVar}:}
	\begin{align}
	\NDLineTG{\Type{A}}{\byAss}\label{ctxVar1}\\
	\NDLinePTG{\Type{\PhiAppl{A}}}{\IH,(\ref{ctxVar1})}\label{ctxVar2}
	\\\NDLinePTG{\PredPhiName{A}:\PhiAppl{A}\to\PhiAppl{A}\to \bool}{\IH,(\ref{ctxVar1})}\label{ctxVarPATp}
	\\\NDLinePT{}{\Ctx{\concatCtx{\PhiAppl{\Gamma}}{x:\PhiAppl{A}}}}{\ruleRef{ctxVar},(\ref{ctxVar2})}\label{ctxVar3}
	\\\NDLinePT{\concatCtx{\PhiAppl{\Gamma}}{x:\PhiAppl{A}}}{\PredPhiName{A}:\PhiAppl{A}\to \PhiAppl{A}\to \bool}{\ruleRef{varDed},(\ref{ctxVar2}),(\ref{ctxVarPATp})}\label{ctxVar3a}
	\\\NDLinePT{\concatCtx{\PhiAppl{\Gamma}}{x:\PhiAppl{A}}}{\PredPhi{A}{x}:\bool}{\ruleRef{appl},(\ref{ctxVar3a}),\ruleRef{varS}}\label{ctxVar4}
	\\\NDLinePT{}{\Ctx{\concatCtx{\concatCtx{\PhiAppl{\Gamma}}{x:\PhiAppl{A}}}{\PredPhi{A}{x}}}}{\ruleRef{ctxAssume},(\ref{ctxVar4})}\label{ctxVar5}
	\\\NDLinePT{}{\Ctx{\PhiAppl{\concatCtx{\Gamma}{ x:A}}}}{\ref{PTctxVar},(\ref{ctxVar5})}
	\end{align}
	
	\subparagraph{\ruleRef{ctxAssume}:}
	\begin{align}
	\NDLineTG{F:\bool}{by assumption}\label{ctxAssume1}
	\\\NDLinePTG{\PhiAppl{F}:\bool}{\IH,(\ref{ctxAssume1})}\label{ctxAssume2}
	\\\NDLinePT{}{\Ctx{\concatCtx{\PhiAppl{\Gamma}}{\namedass{ass}{\PhiAppl{F}}}}}{\ruleRef{ctxAssume},(\ref{ctxAssume2})}\label{ctxAssume3}
	\\\NDLinePT{}{\Ctx{\PhiAppl{\concatCtx{\Gamma}{\namedass{ass}{F}}}}}{\ref{PTctxAss},(\ref{ctxAssume3})}
	\end{align}
	
	\paragraph{Well-formedness of types}
	Well-formedness of types can be shown in DHOL using the rules 
	\ruleRef{type'}, \ruleRef{pi} and \ruleRef{psubType}:
	\subparagraph{\ruleRef{type'}:}
	\begin{align}
	&\thyIn{a:\piType{x_1}{A_1}\ldots\piType{x_n}{A_n}\type}{T}&&\text{\byAss}\label{type0}\\
	\NDLineTG{\Ctx{\Gamma}}{\byAss}\label{typeGCtx}\\
	\NDLinePT{}{\Ctx{\PhiAppl{\Gamma}}}{\IH,(\ref{typeGCtx})}\label{typePGCtx}\\
	&\thyIn{a:\type}{\PhiAppl{T}}&&\text{\ref{PTTpConstr},(\ref{type0})}\label{typeT1}\\
	&\thyIn{\PredPhiName{a}:a\to a\to\bool}{\PhiAppl{T}}&&\text{\ref{PTTpConstr},(\ref{type0})}\label{typeT2}\\
	\NDLinePTG{a\;\type}{\ruleRef{type},(\ref{typeT1}),(\ref{typePGCtx})}\nonumber\\
	\NDLinePTG{\PredPhiName{a}:a\to a\to\bool}{\ruleRef{constS},(\ref{typeT2})}\label{typeT3}\\
	\NDLinePTG{\PredPhiName{a}:\PhiAppl{a}\to\PhiAppl{a}\to\bool}{\ref{PTappl},(\ref{typeT3})}\nonumber
	\end{align}
	
	\subparagraph{\ruleRef{pi}:}
	\begin{align}
	\NDLineTG{A\;\type}{\byAss}\label{pi1}\\
	\NDLineT{\concatCtx{\Gamma}{x:A}}{B\;\type}{\byAss}\label{pi2}\\
	\NDLinePTG{\PhiAppl{A}\;\type}{\IH,(\ref{pi1})}\label{pi3}\\
	\concatCtx{\concatCtx{\PhiAppl{\Gamma}}{x:&\PhiAppl{A}}}{\typingAss{A}{x}}\dedPT\PhiAppl{B}\;\type\QQQNegSp&&\QQuad\text{\IH,(\ref{pi2})}\label{pi4pre}\\
	\NDLinePTG{\PhiAppl{B}\;\type}{\QNegSp\hol{} types context independent,(\ref{pi4pre})}\label{pi4}\\
	\NDLinePTG{\PhiAppl{A}\to \PhiAppl{B}\:\type}{\ruleRef{arrow},(\ref{pi3}),(\ref{pi4})}\label{pi5}\\
	\NDLinePTG{\PredPhiName{A}:\PhiAppl{A}\to\PhiAppl{A}\to\bool}{\IH,(\ref{pi1})}\label{pi6}\\
	\NDLinePTG{\PredPhiName{B}:\PhiAppl{B}\to\PhiAppl{B}\to\bool}{\IH,(\ref{pi2})}\label{pi7}\\
	\NDLinePTG{\PhiAppl{\piType{x}{A}B}\;\type}{\ref{PTPitype},(\ref{pi5})}\nonumber\\
	\NDLinePTG{\PredPhi{\left(\piType{x}{A}B\right)}:(\PhiAppl{\piType{x}{A}B})\to(\PhiAppl{\piType{x}{A}B})\to\bool\QQQQNegSp}{\QQQQuad\ref{PTPipred},(\ref{pi6}),(\ref{pi7})}\nonumber
	\end{align}
	
	\subparagraph{\ruleRef{psubType}:}
	\begin{align}
	\NDLineTG{A\;\type}{\byAss}\label{psubType1}\\
	\NDLineTG{p:A\to\bool}{\byAss}\label{psubType2}\\
	\NDLinePTG{\PhiAppl{A}\:\type}{\IH,(\ref{psubType1})}\label{psubType3}\\
	\NDLinePTG{\PhiAppl{\subtype{A}{p}}\;\type}{\ref{PTPStype},(\ref{psubType3})}\nonumber\\
	\NDLinePTG{\PredPhiName{A}:\PhiAppl{A}\to\PhiAppl{A}\to\bool}{\IH,(\ref{psubType1})}\label{psubType5}\\
	\NDLinePTG{\PredPhiName{\left(\subtype{A}{p}\right)}:\PhiAppl{A}\to\PhiAppl{A}\to\bool}{\ref{PTPSpred},(\ref{psubType5})}\nonumber
	\end{align}
	
	\paragraph{Type-equality}
	Type-equality can be shown using the rules \ruleRef{congBase'}, \ruleRef{congPi}, \ruleRef{psubTriv}, \ruleRef{psubTriv'}, \ruleRef{piPsubCod}, \ruleRef{psubQsub} and \ruleRef{psubEq}:
	
	\subparagraph{\ruleRef{psubEq}:}
	\begin{align}
	\NDLineTG{A\typeEquals A'}{\byAss}\label{psubEq1}\\
	\NDLineTG{p\termEquals{\piType{x}{A}\bool}p'}{\byAss}\label{psubEq2}\\
	\NDLinePTG{\PhiAppl{A}\typeEquals\PhiAppl{A'}}{\IH,(\ref{psubEq1})}\label{psubEq3}\\
	\NDLinePTG{\PredPhiName{A}\termEquals{\PhiAppl{A}\to\PhiAppl{A}\to\bool}\PredPhiName{A'}}{\IH,(\ref{psubEq1})}\label{psubEq4}\\
	\NDLinePTG{\PhiAppl{p}\termEquals{\PhiAppl{A}\to\bool}\PhiAppl{p'}}{\IH,(\ref{psubEq2})}\label{psubEq5}\\
	\NDLinePTG{\PredPhiName{(\subtype{A}{p})}\termEquals{\PhiAppl{A}\to\PhiAppl{A}\to\bool}\PredPhiName{(\subtype{A}{p})}}{\ruleRef{refl},\ref{PTPSpred},\ruleRef{eqTyping},(\ref{psubEq4})}\label{psubEq6}\\
	\NDLinePTG{\PredPhiName{(\subtype{A}{p})}\termEquals{\PhiAppl{A}\to\PhiAppl{A}\to\bool}\PredPhiName{(\subtype{A}{p'})}}{\ruleRef{rewrite},(\ref{psubEq6}),(\ref{psubEq5})}\label{psubEq7}\\
	\NDLinePTG{\PredPhiName{(\subtype{A}{p})}\termEquals{\PhiAppl{A}\to\PhiAppl{A}\to\bool}\PredPhiName{(\subtype{A'}{p'})}}{\ref{PTPSpred},\ruleRef{rewrite},(\ref{psubEq7}),(\ref{psubEq4})}\nonumber\\
	\NDLinePTG{\PhiAppl{\subtype{A}{p}}\typeEquals\PhiAppl{\subtype{A'}{p'}}}{\ref{PTPStype},(\ref{psubEq3})}\nonumber
	\end{align}
	
	\subparagraph{\ruleRef{congBase'}:}
	\begin{align}
	&\thyIn{a:\piType{x_1}{A_1}\ldots\piType{x_n}{A_n}\type}{T}&&\text{by assumption}\label{congBase0}\\
	\NDLineTG{s_1\termEquals{A_1}t_1}{by assumption}\label{congBase1}\\
	\vdots\nonumber\\
	\NDLineTG{s_n\termEquals{\subst{A_n}{x_1}{t_1}\ldots\substOp{x_{n-1}}{t_{n-1}}}t_n}{by assumption}\label{congBasen}\\
	\NDLineT{}{\Ctx{\Gamma}}{\byAss}\label{congBaseGCtx}\\
	&\thyIn{a:\type}{\PhiAppl{T}}&&\text{\ref{PTTpConstr},(\ref{congBase0})}\label{congBaseATpl}\\
	&\thyIn{\PredPhiName{a}:\PhiAppl{A_1}\to\ldots \to \PhiAppl{A_n}\to\PhiAppl{a}\to\PhiAppl{a}\to\bool}{\PhiAppl{T}}&&\text{\ref{PTTpConstr},(\ref{congBase0})}\label{congBaseAPred}\\
	\NDLinePTG{\PhiAppl{s_1}\termEquals{\PhiAppl{A_1}}\PhiAppl{t_1}}{\IH,(\ref{congBase1})}\label{congBase1T}\\
	\vdots\nonumber\\	
	\NDLinePTG{\PhiAppl{s_n}\termEquals{\PhiAppl{A_n}}\PhiAppl{t_n}}{\IH,(\ref{congBasen})}\label{congBasenT}\\
	\NDLinePT{}{\Ctx{\PhiAppl{\Gamma}}}{\IH,(\ref{congBaseGCtx})}\label{congBasePGCtx}\\
	\NDLinePTG{a:\;\type}{\ruleRef{type},(\ref{congBaseATpl}),(\ref{congBasePGCtx})}\label{congBaseATp}\\
	\NDLinePTG{a\typeEquals a}{\ruleRef{congBase},(\ref{congBaseATp})}\label{congBaseT}\\
	\NDLinePTG{\PredPhiName{a}\termEquals{\PhiAppl{A_1}\to\ldots \to \PhiAppl{A_n}\to\PhiAppl{a}\to\PhiAppl{a}\to\bool}\PredPhiName{a}}{\ruleRef{refl},\ruleRef{constS},(\ref{congBaseAPred}),(\ref{congBasePGCtx})}\label{congBaseApredEqApred}\\
	\NDLinePTG{\PredPhiName{a}\ \PhiAppl{s_1}\termEquals{\PhiAppl{A_2}\to\ldots \to \PhiAppl{A_n}\to\PhiAppl{a}\to\PhiAppl{a}\to\bool}\PredPhiName{a}\ \PhiAppl{t_1}}{\ruleRef{congAppl},(\ref{congBase1T}),(\ref{congBaseApredEqApred})}\label{congBasePredApp1}\\
	\vdots\nonumber\\	
	\NDLinePTG{\PredPhiName{a}\ \PhiAppl{s_1}\ \ldots\ \PhiAppl{s_n}\termEquals{\PhiAppl{a}\to\PhiAppl{a}\to\bool}\PredPhiName{a}\ \PhiAppl{t_1}\ \ldots\ \PhiAppl{t_n}}{\ruleRef{congAppl},(\ref{congBasenT}),previous line}\label{congBasePredAppn}\\
	\NDLinePTG{\PhiAppl{a\ s_1\ \ldots\ s_n}\typeEquals \PhiAppl{a\ t_1\ \ldots\ t_n}}{\ref{PTTpAppl},(\ref{congBaseT})}\nonumber\\
	\NDLinePTG{\PredPhiName{(a\ s_1\ \ldots\ s_n)}\termEquals{\PhiAppl{a\ s_1\ \ldots\ s_n}\to\PhiAppl{a\ s_1\ \ldots\ s_n}\to \bool}\PredPhiName{(a\ t_1\ \ldots\ t_n)}\QQQNegSp}{\QQQuad\ref{PTTpPredAppl},(\ref{congBasePredAppn})}\nonumber
	\end{align}
	
	\subparagraph{\ruleRef{congPi}:}
	{\small
		\begin{align}
		\NDLineTG{A\typeEquals A'}{\byAss}\label{congPi1}\\
		\NDLineT{\concatCtx{\Gamma}{x:A}}{B\typeEquals B}{\byAss}\label{congPi2}\\
		\NDLinePTG{\PhiAppl{A}\typeEquals\PhiAppl{A'}}{\IH,(\ref{congPi1})}\label{congPi3}\\
		\NDLinePTG{\PredPhiName{A}\termEquals{\PhiAppl{A}\to\PhiAppl{A}\to\bool}\PredPhiName{A'}}{\IH,(\ref{congPi1})}\label{congPi7}\\
		\NDLinePT{\concatCtx{\concatCtx{\PhiAppl{\Gamma}}{x:\PhiAppl{A}}}{\typingAss{A}{x}}}{\PhiAppl{B}\typeEquals\PhiAppl{B'}}{\IH,(\ref{congPi2})}\label{congPi4}
		\intertext{Since $\typeEquals$ is context independent in \hol{}:}
		\NDLinePTG{\PhiAppl{B}\typeEquals\PhiAppl{B'}}{explanation,(\ref{congPi4})}\label{congPi5}\\
		\NDLinePTG{\PhiAppl{A}\to\PhiAppl{B}\typeEquals\PhiAppl{A'}\to\PhiAppl{B'}}{\ruleRef{congTo},(\ref{congPi3}),(\ref{congPi5})}\label{congPi6}\\
		\NDLinePTG{\PhiAppl{\piType{x}{A}B}\typeEquals\PhiAppl{\piType{x}{A'}B'}}{\ref{PTPitype},(\ref{congPi6})}\nonumber
		\end{align}
		\begin{align}
		\NDLinePT{\concatCtx{\concatCtx{\PhiAppl{\Gamma}}{x:\PhiAppl{A}}}{\typingAss{A}{x}}}{\PredPhiName{B}\termEquals{\PhiAppl{B}\to\PhiAppl{B}\to\bool}\PredPhiName{B'}}{\QNegSp\IH,(\ref{congPi2})}\label{congPi8}\\
		\concatCtx{\concatCtx{\PhiAppl{\Gamma}}{f:\PhiAppl{A}\to\PhiAppl{B}}}{x:&\PhiAppl{A}}\dedPT\PredPhi{A}{x}\impl \PredPhi{B}{(f\ x)}&&\nonumber\\&\termEqB \PredPhi{A'}{x}\impl \PredPhi{B}{(f\ x)}&&\text{\ruleRef{rewrite},\ruleRef{refl},(\ref{congPi7})}\label{congPi9}\\
		\concatCtx{\concatCtx{\PhiAppl{\Gamma}}{f:\PhiAppl{A}\to\PhiAppl{B}}}{x:&\PhiAppl{A}}\dedPT\PredPhi{A}{x}\impl \PredPhi{B}{(f\ x)}&&\nonumber\\&\termEqB \PredPhi{A'}{x}\impl \PredPhi{B'}{(f\ x)}&&\text{\ruleRef{rewrite},(\ref{congPi9}),(\ref{congPi8})}\label{congPi10}
		\end{align}
		\begin{align}
		\NDLinePT{\concatCtx{\PhiAppl{\Gamma}}{f:\PhiAppl{A}\to\PhiAppl{B}}}{\univQuant{x}{\PhiAppl{A}} \PredPhi{A}{x}\impl (\PredPhi{B}{(f\ x)})\termEqB&&\nonumber\\&
			\univQuant{x}{\PhiAppl{A'}} \PredPhi{A'}{x} \impl (\PredPhi{B'}{(f\ x)})}{\ruleRef{forallCong},(\ref{congPi3}),(\ref{congPi10})}\label{congPi11}\\
		\NDLinePTG{\PredPhiName{\left(\piType{x}{A}B\right)}\termEquals{\PhiAppl{A}\to\PhiAppl{A}\to\bool}\PredPhiName{\left(\piType{x}{A'}B'\right)}}{\ref{PTLam},\ruleRef{congLam},(\ref{congPi11})}\nonumber
		\end{align}
	}
	\subparagraph{\ruleRef{psubTriv}:}
	{\small
		\begin{align}
		\NDLineTG{\Type{A}}{\byAss}\label{psubTriv1}\\
		\NDLinePTG{\Type{\PhiAppl{A}}}{\IH,(\ref{psubTriv1})}\label{psubTriv2}\\
		\NDLinePTG{\PredPhiName{A}:\PhiAppl{A}\to\PhiAppl{A}\to\bool}{\IH,(\ref{psubTriv1})}\label{psubTriv3}\\
		\NDLinePTG{\PhiAppl{A}\typeEquals\PhiAppl{A}}{\ruleRef{tpEqRefl},(\ref{psubTriv2})}\label{psubTriv4}\\
		\NDLinePTG{\PhiAppl{A}\typeEquals\PhiAppl{\subtype{A}{\lambdaFun{x}{A}\T}}}{(\ref{PTPStype}),(\ref{psubTriv4})}\nonumber\\
		\NDLinePT{\concatCtx{\PhiAppl{\Gamma}}{x:\PhiAppl{A}}}{\PhiAppl{\lambdaFun{x}{A}\T}\ x}{Lemma~\ref{lem:eqTrue},\ruleRef{beta}}\label{psubTriv6}\\
		\NDLinePT{\concatCtx{\concatCtx{\PhiAppl{\Gamma}}{x:\PhiAppl{A}}}{\typingAss{A}{x}}}{\PredPhi{A}{x}\land \PhiAppl{\lambdaFun{x}{A}\T}\ x}{\ruleRef{andI},\ruleRef{assume},(\ref{psubTriv6})}\label{psubTriv7}\\
		\concatCtx{\concatCtx{\PhiAppl{\Gamma}}{x:\PhiAppl{A}}}{\typingAss{A}{x}\land& \PhiAppl{\lambdaFun{x}{A}\T}\ x}\dedPT\PredPhi{A}{x}&&\text{\ruleRef{andEl},\ruleRef{assume}}\label{psubTriv8}\\
		\NDLinePT{\concatCtx{\PhiAppl{\Gamma}}{x:\PhiAppl{A}}}{\PredPhi{A}{x}\termEqB \PredPhi{\left(\subtype{A}{\lambdaFun{x}{A}\T}\right)}{x}}{\ruleRef{propExt},(\ref{psubTriv7}),(\ref{psubTriv8})}\nonumber
		\end{align}}
	\subparagraph{\ruleRef{psubTriv'}:}
	{\small
		\begin{align}
		\NDLineTG{\Type{A}}{\byAss}\label{psubTrivp1}\\
		\NDLinePTG{\Type{\PhiAppl{A}}}{\IH,(\ref{psubTrivp1})}\label{psubTrivp2}\\
		\NDLinePTG{\PredPhiName{A}:\PhiAppl{A}\to\PhiAppl{A}\to\bool}{\IH,(\ref{psubTrivp1})}\label{psubTrivp3}\\
		\NDLinePTG{\PhiAppl{A}\typeEquals\PhiAppl{A}}{\ruleRef{tpEqRefl},(\ref{psubTrivp2})}\label{psubTrivp4}\\
		\NDLinePTG{\PhiAppl{\subtype{A}{\lambdaFun{x}{A}\T}}\typeEquals\PhiAppl{A}}{(\ref{PTPStype}),(\ref{psubTrivp4})}\nonumber\\
		\NDLinePT{\concatCtx{\PhiAppl{\Gamma}}{x:\PhiAppl{A}}}{\PhiAppl{\lambdaFun{x}{A}\T}\ x}{Lemma~\ref{lem:eqTrue},\ruleRef{beta}}\label{psubTrivp6}\\
		\NDLinePT{\concatCtx{\concatCtx{\PhiAppl{\Gamma}}{x:\PhiAppl{A}}}{\typingAss{A}{x}}}{\PredPhi{A}{x}\land \PhiAppl{\lambdaFun{x}{A}\T}\ x}{\ruleRef{andI},\ruleRef{assume},(\ref{psubTrivp6})}\label{psubTrivp7}\\
		\concatCtx{\concatCtx{\PhiAppl{\Gamma}}{x:\PhiAppl{A}}}{\typingAss{A}{x}\land& \PhiAppl{\lambdaFun{x}{A}\T}\ x}\dedPT\PredPhi{A}{x}&&\text{\ruleRef{andEl},\ruleRef{assume}}\label{psubTrivp8}\\
		\NDLinePT{\concatCtx{\PhiAppl{\Gamma}}{x:\PhiAppl{A}}}{\PredPhi{\left(\subtype{A}{\lambdaFun{x}{A}\T}\right)}{x}\termEqB \PredPhi{A}{x} }{\ruleRef{propExt},(\ref{psubTrivp8}),(\ref{psubTrivp7})}\nonumber
		\end{align}}
	
	\subparagraph{\ruleRef{piPsubCod}:}
	\begin{align}
		\NDLineTG{\Type{A}}{\byAss}\label{piPsubCod1}\\
		\NDLineT{\concatCtx{\Gamma}{x:A}}{\Type{B}}{\byAss}\label{piPsubCod1b}\\
		\NDLineT{\concatCtx{\Gamma}{x:A}}{p:\piType{y}{B}\bool}{\byAss}\label{piPsubCod2}\\
		\NDLinePTG{\Type{\PhiAppl{A}}}{\IH,(\ref{piPsubCod1})}\label{piPsubCod3}\\
		\NDLinePTG{\PredPhiName{A}:\PhiAppl{A}\to\PhiAppl{A}\to\bool}{\IH,(\ref{piPsubCod1})}\label{piPsubCod4}\\
		\concatCtx{\concatCtx{\PhiAppl{\Gamma}}{x:&\PhiAppl{A}}}{\typingAss{A}{x}}\dedPT\Type{\PhiAppl{B}}&&\text{\IH,(\ref{piPsubCod1b})}\label{piPsubCod3bpre}\\
		\NDLinePTG{\Type{\PhiAppl{B}}}{\QQNegSp\hol{} types context independent,(\ref{piPsubCod3bpre})}\label{piPsubCod3b}\\
		\concatCtx{\concatCtx{\PhiAppl{\Gamma}}{x:&\PhiAppl{A}}}{\typingAss{A}{x}}\dedPT\PredPhiName{B}:\PhiAppl{B}\to\PhiAppl{B}\to\bool\QQQNegSp&&\QQQuad\text{\IH,(\ref{piPsubCod1b})}\label{piPsubCod4bpre}\\
		\NDLinePTG{\PredPhiName{B}:\PhiAppl{B}\to\PhiAppl{B}\to\bool}{\QQNegSp\hol{} types context independent,(\ref{piPsubCod4bpre})}\label{piPsubCod4b}\\
		\NDLinePTG{\Type{\PhiAppl{A}\to\PhiAppl{B}}}{\ruleRef{arrow},(\ref{piPsubCod3}),(\ref{piPsubCod3b})}\label{piPsubCod8}\\
		\NDLinePTG{\PhiAppl{A}\to\PhiAppl{B}\typeEquals \PhiAppl{A}\to\PhiAppl{B}}{\ruleRef{congBase},(\ref{piPsubCod8})}\label{piPsubCod9}
		\end{align}
		\begin{align}
		\NDLinePTG{\PhiAppl{\piType{x}{A}\subtype{B}{p}}\typeEquals \PhiAppl{\subtype{(\piType{x}{A}B)}{\lambdaFun{f}{\piType{x}{A}B}\univQuant{x}{A}p\ \left(f\ x\right)}}}{(\ref{PTPStype}),(\ref{PTPitype}),(\ref{piPsubCod9})}\nonumber
	\end{align}
	Denote $F:=\univQuant{x,x'}{\PhiAppl{A}}\termEqT{A}{x}{x'}\impl
	\termEqT{B}{(f\ x)}{(f\ x')}\land  \univQuant{x,x'}{\PhiAppl{A}}\termEqT{A}{x}{x'}\impl \PhiAppl{p}\ (f\ x)$, $G:=\univQuant{x,x'}{\PhiAppl{A}}\termEqT{A}{x}{x'}\impl \left(\termEqT{B}{(f\ x)}{(f\ x')}\land \PhiAppl{p}\ (f\ x)\right)$.
	Further denote $\Delta:=\concatCtx{\PhiAppl{\Gamma}}{f:\PhiAppl{A}\to\PhiAppl{B}},$   $$\Theta:=\concatCtx{\concatCtx{\concatCtx{\Delta}{\namedass{ass_F}{F}}}{x,x':\PhiAppl{A}}}{\namedass{xRx'}{\termEqT{A}{x}{x'}}}$$ and $$\Lambda:=\concatCtx{\concatCtx{\concatCtx{\Delta}{\namedass{ass_G}{G}}}{x,x':\PhiAppl{A}}}{\namedass{xRx'}{\termEqT{A}{x}{x'}}}.$$
	\begin{align}
	\NDLinePT{\Theta}{\termEqT{B}{(f\ x)}{(f\ x')}}{\ruleRef{forallE},\ruleRef{implE},\ruleRef{andEl},\ruleRef{assume},\ruleRef{assume}}\label{piPsubCod13}\\
	\NDLinePT{\Theta}{\PhiAppl{p}\ (f\ x)}{\ruleRef{forallE},\ruleRef{implE},\ruleRef{andEr},\ruleRef{assume},\ruleRef{assume}}\label{piPsubCod14}\\
	\NDLinePT{\Theta}{\termEqT{B}{(f\ x)}{(f\ x')}\land\PhiAppl{p}\ (f\ x)\QQNegSp}{\QQuad\ruleRef{andI},(\ref{piPsubCod13}),(\ref{piPsubCod14})}\label{piPsubCod15}\\
	\NDLinePT{\concatCtx{\Delta}{\namedass{ass_F}{F}}}{G}{\ruleRef{forallI},\ruleRef{forallI},\ruleRef{implI},(\ref{piPsubCod15})}\label{piPsubCod16}
	\end{align}
	\begin{align}
		\NDLinePT{\Lambda}{\termEqT{B}{(f\ x)}{(f\ x')}\land\PhiAppl{p}\ (f\ x)\QNegSp}{\ruleRef{implE},\ruleRef{forallE},\ruleRef{forallE},\ruleRef{assume},\ruleRef{assume}}\label{piPsubCod17}\\
		\NDLinePT{\Lambda}{\termEqT{B}{(f\ x)}{(f\ x')}}{\ruleRef{andEl},(\ref{piPsubCod17})}\label{piPsubCod18}\\
		\NDLinePT{\Lambda}{\PhiAppl{p}\ (f\ y)}{\ruleRef{andEl},(\ref{piPsubCod17})}\label{piPsubCod19}\\
		\NDLinePT{\concatCtx{\Delta}{\namedass{ass_G}{G}}}{\univQuant{x,x'}{\PhiAppl{A}}\termEqT{A}{x}{x'}\impl
		\termEqT{B}{(f\ x)}{(f\ x')}\QQQNegSp}{\QQuad\quad\ruleRef{forallI},\ruleRef{forallI},\ruleRef{implI},(\ref{piPsubCod17})}\label{piPsubCod20}\\
		\NDLinePT{\concatCtx{\Delta}{\namedass{ass_G}{G}}}{\univQuant{x,x'}{\PhiAppl{A}}\termEqT{A}{x}{x'}\impl
		\PhiAppl{p}\ (f\ x)\QQQNegSp}{\QQuad\quad\ruleRef{forallI},\ruleRef{forallI},\ruleRef{implI},(\ref{piPsubCod18})}\label{piPsubCod21}\\
		\NDLinePT{\concatCtx{\Delta}{\namedass{ass_G}{G}}}{F}{\ruleRef{andI},(\ref{piPsubCod20}),(\ref{piPsubCod21})}\label{piPsubCod22}\\
		\NDLinePTD{F\termEqB G}{\ruleRef{propExt},(\ref{piPsubCod22}),(\ref{piPsubCod22})}\label{piPsubCod23}\\
		\NDLinePTD{\PhiAppl{\lambdaFun{f}{(\piType{x}{A}B)}\univQuant{x}{A}p\ (f\ x)}\ f\termEqB\nonumber\QQQNegSp&&\\& \univQuant{x,x'}{\PhiAppl{A}}\termEqT{A}{x}{x'}\impl \PhiAppl{p}\ (f\ x)\QQQNegSp}{\QQuad\quad\ruleRef{beta}}\label{piPsubCod24}
	\end{align}
	\begin{align}
		\NDLinePTD{\PredPhi{\subtype{(\piType{x}{A}B)}{\lambdaFun{f}{(\piType{x}{A}B)}\univQuant{x}{A}p\ (f\ x)}}{f}\termEqB F}{\ruleRef{rewrite},\ruleRef{refl},(\ref{piPsubCod24})}\label{piPsubCod25}\\
		\NDLinePTD{\PredPhi{\subtype{(\piType{x}{A}B)}{\lambdaFun{f}{(\piType{x}{A}B)}\univQuant{x}{A}p\ (f\ x)}}{f}\termEqB G}{\ruleRef{trans},(\ref{piPsubCod25}),(\ref{piPsubCod23})}\label{piPsubCod26}\\
		\NDLinePTD{\PredPhi{(\piType{x}{A}\subtype{B}{p})}{f}\termEqB\nonumber\\& \PredPhi{\big(\subtype{(\piType{x}{A}B)}{\lambdaFun{f}{(\piType{x}{A}B)}\univQuant{x}{A}p\ (f\ x)}\big)}{f}}{(\ref{PTPipred}),(\ref{PTPSpred}),(\ref{piPsubCod26})}\nonumber
	\end{align}
	
	\subparagraph{\ruleRef{psubQsub}:}
	\begin{align}
		\NDLineTG{\Type{A}}{\byAss}\label{psubQsub1}\\
		\NDLineTG{p:\piType{y}{B}\bool}{\byAss}\label{psubQsub2}\\
		\NDLineTG{q:\piType{y}{B}\bool}{\byAss}\label{psubQsub2b}\\
		\NDLinePTG{\Type{\PhiAppl{A}}}{\IH,(\ref{psubQsub1})}\label{psubQsub3}\\
		\NDLinePTG{\PredPhiName{A}:\PhiAppl{A}\to\PhiAppl{A}\to\bool}{\IH,(\ref{psubQsub1})}\label{psubQsub4}\\
		\NDLinePTG{\PhiAppl{p}:\PhiAppl{B}\to\bool}{\IH,(\ref{psubQsub2})}\label{psubQsub5}\\
		\NDLinePTG{\PhiAppl{q}:\PhiAppl{B}\to\bool}{\IH,(\ref{psubQsub2})}\label{psubQsub5b}\\
		\NDLinePTG{\PhiAppl{A}\typeEquals\PhiAppl{A}}{\ruleRef{congBase},(\ref{psubQsub3})}\label{psubQsub6}\\
		\NDLinePTG{\PhiAppl{\subtype{\subtype{A}{p}}{q}}\typeEquals\PhiAppl{\subtype{A}{\lambdaFun{x}{A}p\ x\land q\ x}}}{(\ref{PTPStype}),(\ref{psubQsub6})}\nonumber\\
		\NDLinePT{\concatCtx{\PhiAppl{\Gamma}}{x:\PhiAppl{A}}}{\PredPhi{A}{x}\land \PhiAppl{p}\ x\land \PhiAppl{q}\ x:\bool}{\ruleRef{appl} repeatedly,\ruleRef{varS},(\ref{psubQsub5}),(\ref{psubQsub5b})}\label{psubQsub8}\\
		\NDLinePT{\concatCtx{\PhiAppl{\Gamma}}{x:\PhiAppl{A}}}{\PredPhi{A}{x}\land \PhiAppl{p}\ x\land \PhiAppl{q}\ x\termEqB\nonumber\\&\PredPhi{A}{x}\land \PhiAppl{p}\ x\land \PhiAppl{q}\ x}{\ruleRef{refl},(\ref{psubQsub8})}\label{psubQsub8a}\\
		\NDLinePT{\concatCtx{\PhiAppl{\Gamma}}{x:\PhiAppl{A}}}{\left(\lambdaFun{x}{\PhiAppl{A}}\PhiAppl{p}\ x\land \PhiAppl{q}\ x\right)\ x:\bool}{\ruleRef{appl},\ruleRef{lambda},(\ref{psubQsub8}),\ruleRef{varS}}\label{psubQsub9}\\
		\NDLinePT{\concatCtx{\PhiAppl{\Gamma}}{x:\PhiAppl{A}}}{\left(\lambdaFun{x}{\PhiAppl{A}}\PhiAppl{p}\ x\land \PhiAppl{q}\ x\right)\ x\nonumber\\&\termEqB\PhiAppl{p}\ x\land \PhiAppl{q}\ x}{\ruleRef{beta},(\ref{psubQsub9})}\label{psubQsub10}\\
		\NDLinePT{\concatCtx{\PhiAppl{\Gamma}}{x:\PhiAppl{A}}}{\PredPhi{\left(\subtype{\subtype{A}{p}}{q}\right)}{x}\termEqB\PredPhi{\left(\subtype{A}{\lambdaFun{x}{A}p\ x\land q\ x}\right)}{x}}{\ruleRef{rewrite},(\ref{psubQsub8a}),(\ref{psubQsub10})}\nonumber		
	\end{align}
	
	\paragraph{Subtyping}
	Subtyping can be shown using the rules \ruleRef{subtPsub}, \ruleRef{subtI}, \ruleRef{subtPi} and \ruleRef{subtPsubCong}.
	
	\subparagraph{\ruleRef{subtPsub}:}
	\begin{align}
		\NDLineTG{A\subtyping A'}{\byAss}\label{subtPsub1}\\
		\NDLinePTG{\PhiAppl{A}\typeEquals \PhiAppl{A'}}{\IH,(\ref{subtPsub1})}\nonumber\\
		\NDLinePT{\concatCtx{\PhiAppl{\Gamma}}{x,y:\PhiAppl{A}}}{\termEqT{A}{x}{y}\impl \termEqT{A'}{x}{y}}{\IH,(\ref{subtPsub1})}\label{subtPsub3}\\
		\concatCtx{\concatCtx{\PhiAppl{\Gamma}}{x,y:\PhiAppl{A}}}{\namedass{xRy}{&\termEqT{\left(\subtype{A}{p}\right)}{x}{y}}}\dedPT\PredPhi{A'}{x}{y}&&\text{\ruleRef{implE},(\ref{subtPsub3}),\ruleRef{andEl},\ref{PTPSpred},\ruleRef{assume}}\label{subtPsub4}\\
		\NDLinePT{\concatCtx{\PhiAppl{\Gamma}}{x,y:\PhiAppl{A}}}{\termEqT{\left(\subtype{A}{p}\right)}{x}{y}\impl \termEqT{A'}{x}{y}}{\ruleRef{implI},(\ref{subtPsub4})}\nonumber
	\end{align}
	\subparagraph{\ruleRef{subtI}:}
	\begin{align}
		\NDLineTG{A\typeEquals A'}{\byAss}\label{subtI1}\\
		\NDLinePTG{\PhiAppl{A}\typeEquals\PhiAppl{A'}}{\IH,(\ref{subtI1})}\label{subtI2}\\
		\NDLinePTG{\PredPhiName{A}\termEquals{\PhiAppl{A}\to\PhiAppl{A}\to\bool}\PredPhiName{A'}}{\IH,(\ref{subtI1})}\label{subtI3}\\
		\concatCtx{\concatCtx{\PhiAppl{\Gamma}}{x,y:\PhiAppl{A}}}{\namedass{xRy}{&\termEqT{A}{x}{y}}}\dedPT\termEqT{A'}{x}{y}&&\text{\ruleRef{rewrite},\ruleRef{assume},(\ref{subtI3})}\label{subtI4}\\
		\concatCtx{\PhiAppl{\Gamma}}{x,y:\PhiAppl{A}}\dedPT&\termEqT{A}{x}{y}\impl \termEqT{A'}{x}{y}&&\text{\ruleRef{implI},(\ref{subtI4})}\nonumber
	\end{align}
	\subparagraph{\ruleRef{subtPi}:}
	{\small
	\begin{align}
		\NDLineTG{A'\subtyping A}{\byAss}\label{subtPi1}\\
		\NDLineT{\concatCtx{\Gamma}{x:A'}}{B\subtyping B'}{\byAss} \label{subtPi2}\\
		\NDLinePTG{\PhiAppl{A'}\typeEquals \PhiAppl{A}}{\IH,(\ref{subtPi1})}\label{subtPi3}
		\end{align}
		Since \hol{} types cannot depend on terms:	
		\begin{align}
		\NDLinePTG{\PhiAppl{B'}\typeEquals \PhiAppl{B}}{\IH,(\ref{subtPi2})}\label{subtPi4}\\
		\NDLinePTG{\PhiAppl{A}\to\PhiAppl{B}\typeEquals\PhiAppl{A'}\to\PhiAppl{B'}}{\ruleRef{congTo},(\ref{subtPi3}),(\ref{subtPi4})}\nonumber\\
		\NDLinePT{\concatCtx{\PhiAppl{\Gamma}}{x,y:\PhiAppl{A}}}{\termEqT{A'}{x}{y}\impl \termEqT{A}{x}{y}}{\IH,(\ref{subtPi1})}\label{subtPi6}\\
		\concatCtx{\concatCtx{\concatCtx{\PhiAppl{\Gamma}}{x,y:\PhiAppl{A'}}}{\namedass{xRy}{\termEqT{A'}{x}{y}}}}{z,z':\PhiAppl{B'}}\dedPT\QQQNegSp&\QQQuad\termEqT{B}{z}{z'}\impl\termEqT{B'}{z}{z'}\QQQNegSp&&\QQuad\text{\IH,(\ref{subtPi2})}\label{subtPi7}\\
		\NDLinePT{\concatCtx{\concatCtx{\PhiAppl{\Gamma}}{x,y:\PhiAppl{A'}}}{\typingAss{A'}{x}}}{\univQuant{z,z'}{\PhiAppl{B'}}\termEqT{B}{z}{z'}\impl\termEqT{B'}{z}{z'}\QQNegSp}{\QQuad \ruleRef{forallI},\ruleRef{forallI},(\ref{subtPi7})}\label{subtPi8}
		\end{align}
		\begin{align}
		\concatCtx{\concatCtx{\concatCtx{\PhiAppl{\Gamma}}{f:\PhiAppl{\piType{x}{A}B}}}{x,y:\PhiAppl{A'}}}{\namedass{xRy}{\termEqT{A'}{x}{y}}}\dedPT \QQQQQNegSp\nonumber\\
		&\termEqT{B}{\left(f\ x\right)}{\left(f\ y\right)}\impl\termEqT{B'}{\left(f\ x\right)}{\left(f\ y\right)} &&\text{\ruleRef{forallE},\ruleRef{forallE},(\ref{subtPi8})} \label{subtPiPre3}\\
		\concatCtx{\concatCtx{\PhiAppl{\Gamma}}{f:\PhiAppl{\piType{x}{A}B}}}{x,y:\PhiAppl{A'}}\dedPT \left(\termEqT{A}{x}{y}\impl \termEqT{B}{\left(f\ x\right)}{\left(f\ y\right)}\impl\right)\QQQQQNegSp\nonumber\\
		&\left(\termEqT{A'}{x}{y}
		\impl \termEqT{B'}{\left(f\ x\right)}{\left(f\ y\right)}\right)
		&&\text{\ruleRef{implFunctorial},(\ref{subtPi6}),(\ref{subtPiPre3})}\label{subtPiPre2}\\
		\NDLinePT{\concatCtx{\PhiAppl{\Gamma}}{f:\PhiAppl{\piType{x}{A}B}}}{\left(\univQuant{x,y}{\PhiAppl{A'}}\termEqT{A}{x}{y}\impl \termEqT{B}{\left(f\ x\right)}{\left(f\ y\right)}\right) \impl \QQQNegSp\nonumber\\&
		\left(\univQuant{x,y}{\PhiAppl{A'}}\termEqT{A'}{x}{y}
		\impl \termEqT{B'}{\left(f\ x\right)}{\left(f\ y\right)}\right)}{\ruleRef{forallImpl},\ruleRef{forallImpl},(\ref{subtPiPre2})}\label{subtPiPre1}\\
		\NDLinePT{\concatCtx{\PhiAppl{\Gamma}}{f:\PhiAppl{\piType{x}{A}B}}}{\PredPhi{\left(\piType{x}{A}B\right)}{f}\impl \PredPhi{\left(\piType{x}{A'}B'\right)}{f}}{\ref{PTPipred},(\ref{subtPiPre1})}\nonumber
	\end{align}
	}
	\subparagraph{\ruleRef{subtPsubCong}:}
	{\small
	\begin{align}
		\NDLineTG{A\subtyping A'}{\byAss}\label{subtPsubCong1}\\
		\NDLineT{\concatCtx{\Gamma}{x:A}}{p\ x\impl p'\ x}{\byAss}\label{subtPsubCong2}\\
		\NDLinePTG{\PhiAppl{A}\typeEquals\PhiAppl{A'}}{\IH,(\ref{subtPsubCong1})}\nonumber\\
		\NDLinePT{\concatCtx{\PhiAppl{\Gamma}}{x,y:\PhiAppl{A}}}{\termEqT{A}{x}{y}\impl \termEqT{A'}{x}{y}}{\IH,(\ref{subtPsubCong1})}\label{subtPsubCong4}\\
		\NDLinePT{\concatCtx{\concatCtx{\PhiAppl{\Gamma}}{x:\PhiAppl{A}}}{\typingAss{A}{x}}}{\PhiAppl{p}\ x\impl \PhiAppl{p'}\ x}{\IH,(\ref{subtPsubCong2})}\label{subtPsubCong5}\\
		\NDLinePTG{\univQuant{x}{\PhiAppl{A}}\PredPhi{A}{x}\impl \PhiAppl{p}\ x\impl \PhiAppl{p'}\ x\QQQNegSp}{\QQuad\quad\ruleRef{forallI},\ruleRef{implI},(\ref{subtPsubCong5})}\label{subtPsubCong5a}\\
		\NDLinePT{\concatCtx{\concatCtx{\PhiAppl{\Gamma}}{x,y:\PhiAppl{A'}}}{\namedass{xRy}{\termEqT{(\subtype{A}{p})}{x}{y}}}}{\PhiAppl{p}\ x\land \PhiAppl{p}\ y}{\ruleRef{andEr},\ruleRef{assume}}\label{subtPsubCong6}\\
		\NDLinePT{\concatCtx{\concatCtx{\PhiAppl{\Gamma}}{x,y:\PhiAppl{A'}}}{\namedass{xRy}{\termEqT{(\subtype{A}{p})}{x}{y}}}}{\termEqT{A}{x}{y}}{\ruleRef{andEl},\ruleRef{assume}}\label{subtPsubCong7}\\
		\NDLinePT{\concatCtx{\concatCtx{\PhiAppl{\Gamma}}{x,y:\PhiAppl{A'}}}{\namedass{xRy}{\termEqT{(\subtype{A}{p})}{x}{y}}}}{\termEqT{A}{x}{y}\termEqB\PredPhi{A}{x}}{(\ref{correct:substRelatTerms}),\ref{PTTpBoolPred},\ruleRef{refl},(\ref{subtPsubCong7})}\label{subtPsubCong8pre6}\\
		\NDLinePT{\concatCtx{\concatCtx{\PhiAppl{\Gamma}}{x,y:\PhiAppl{A'}}}{\namedass{xRy}{\termEqT{(\subtype{A}{p})}{x}{y}}}}{\termEqT{A}{x}{y}\termEqB\PredPhi{A}{y}}{(\ref{correct:substRelatTerms}),\ref{PTTpBoolPred},\ruleRef{refl},(\ref{subtPsubCong7})}\label{subtPsubCong8pre5}\\
		\NDLinePT{\concatCtx{\concatCtx{\PhiAppl{\Gamma}}{x,y:\PhiAppl{A'}}}{\namedass{xRy}{\termEqT{(\subtype{A}{p})}{x}{y}}}}{\PredPhi{A}{x}}{\ruleRef{congDed},(\ref{subtPsubCong8pre6}),(\ref{subtPsubCong7})}\label{subtPsubCong8pre4}\\
		\NDLinePT{\concatCtx{\concatCtx{\PhiAppl{\Gamma}}{x,y:\PhiAppl{A'}}}{\namedass{xRy}{\termEqT{(\subtype{A}{p})}{x}{y}}}}{\PredPhi{A}{y}}{\ruleRef{congDed},(\ref{subtPsubCong8pre5}),(\ref{subtPsubCong7})}\label{subtPsubCong8pre3}\\
		\concatCtx{\concatCtx{\concatCtx{\PhiAppl{\Gamma}}{x,y:\PhiAppl{A'}}}{\namedass{xRy}{\termEqT{(\subtype{A}{p})}{x}{y}}}}{\namedass{ass}{\PhiAppl{p}\ x\land \PhiAppl{p}\ y}}\dedPT\QQQNegSp\QNegSp&\QQQuad\Quad\PhiAppl{p'}\ x&&\text{\ruleRef{implE},\ruleRef{forallE},(\ref{subtPsubCong5})}\label{subtPsubCong8pre2}\\
		\concatCtx{\concatCtx{\concatCtx{\PhiAppl{\Gamma}}{x,y:\PhiAppl{A'}}}{\namedass{xRy}{\termEqT{(\subtype{A}{p})}{x}{y}}}}{\namedass{ass}{\PhiAppl{p}\ x\land \PhiAppl{p}\ y}}
		\dedPT\QQQNegSp\QNegSp&\QQQuad\Quad\PhiAppl{p'}\ y&&\text{\ruleRef{implE},\ruleRef{forallE},(\ref{subtPsubCong5})}\label{subtPsubCong8pre1}\\
		\NDLinePT{\concatCtx{\concatCtx{\PhiAppl{\Gamma}}{x,y:\PhiAppl{A'}}}{\namedass{xRy}{\termEqT{(\subtype{A}{p})}{x}{y}}}}{\PhiAppl{p}\ x\land \PhiAppl{p}\ y\impl \PhiAppl{p'}\ x\land \PhiAppl{p'}\ y\QQNegSp}{\Quad\ruleRef{implI},\ruleRef{andI},(\ref{subtPsubCong8pre2}),(\ref{subtPsubCong8pre1})}\label{subtPsubCong8}\\
		\NDLinePT{\concatCtx{\concatCtx{\PhiAppl{\Gamma}}{x,y:\PhiAppl{A'}}}{\PredPhi{\left(\subtype{A}{p}\right)}{x}}}{\PhiAppl{p'}\ x\land \PhiAppl{p'}\ y}{\ruleRef{implE},(\ref{subtPsubCong8}),(\ref{subtPsubCong6})}\label{subtPsubCong9}\\
		\NDLinePT{\concatCtx{\concatCtx{\PhiAppl{\Gamma}}{x,y:\PhiAppl{A'}}}{\namedass{xRy}{\termEqT{(\subtype{A}{p})}{x}{y}}}}{\termEqT{A'}{x}{y}}{\ruleRef{implE},(\ref{subtPsubCong4}),(\ref{subtPsubCong7})}\label{subtPsubCong10}\\	
		\NDLinePT{\concatCtx{\concatCtx{\PhiAppl{\Gamma}}{x,y:\PhiAppl{A'}}}{\namedass{xRy}{\termEqT{(\subtype{A}{p})}{x}{y}}}}{\PredPhi{\left(\subtype{A'}{p'}\right)}{x}}{\ruleRef{andI},\ruleRef{andI},(\ref{subtPsubCong10}),(\ref{subtPsubCong9}),(\ref{subtPsubCong9})}\label{subtPsubCong11}\\
		\NDLinePT{\concatCtx{\PhiAppl{\Gamma}}{x,y:\PhiAppl{\subtype{A'}{p'}}}}{\termEqT{\left(\subtype{A}{p}\right)}{x}{y} \impl \termEqT{\left(\subtype{A'}{p'}\right)}{x}{y}\QQNegSp}{\QQuad\ref{PTPStype},\ruleRef{implI},(\ref{subtPsubCong11})}
		\nonumber
	\end{align}}
	
	\paragraph{Typing}
	Typing can be shown using the rules \ruleRef{psubI}, \ruleRef{psubE}, \ruleRef{lambda'}, \ruleRef{appl'}, \ruleRef{implType'}, \ruleRef{const}, \ruleRef{congColon},\ruleRef{var}, \ruleRef{eqType}:
	
	\subparagraph{\ruleRef{psubI}:}
	\begin{align}
	\NDLineTG{t:A}{\byAss}\label{psubI1}\\
	\NDLineTG{p\ t}{\byAss}\label{psubI2}\\
	\NDLinePTG{\PhiAppl{t}:\PhiAppl{A}}{\IH,(\ref{psubI1})}\label{psubI3}\\
	\NDLinePTG{\PredPhi{A}\ \PhiAppl{t}}{\IH,(\ref{psubI1})}\label{psubI4}\\
	\NDLinePTG{\PhiAppl{p}\ \PhiAppl{t}}{\IH,(\ref{psubI2})}\label{psubI5}\\
	\NDLinePTG{\PredPhi{\left(\subtype{A}{p}\right)}\ \PhiAppl{t}}{\ref{PTPSpred},\ruleRef{andI},(\ref{psubI4}),\ruleRef{andI},(\ref{psubI5}),(\ref{psubI5})}\nonumber\\
	\NDLinePTG{\PhiAppl{t}:\PhiAppl{\subtype{A}{p}}}{\ref{PTPStype},(\ref{psubI3})}\nonumber
	\end{align}
	\subparagraph{\ruleRef{psubE}:}
	\begin{align}
	\NDLineTG{t:A}{\byAss}\label{subtE1}\\
	\NDLineTG{A\subtyping A'}{\byAss}\label{subtE2}\\
	\NDLinePTG{\PhiAppl{t}:\PhiAppl{A}}{\IH,(\ref{subtE1})}\label{subtE3}\\
	\NDLinePTG{\PhiAppl{A}\typeEquals\PhiAppl{A'}}{\IH,(\ref{subtE2})}\label{subtE4}\\
	\NDLinePTG{\PhiAppl{t}:\PhiAppl{A'}}{\ruleRef{congColon},(\ref{subtE3}),(\ref{subtE4})}\nonumber\\
	\NDLinePTG{\PredPhi{A}{\PhiAppl{t}}}{\IH,(\ref{subtE3})}\label{subtE6}\\
	\NDLinePT{\concatCtx{\PhiAppl{\Gamma}}{x:\PhiAppl{A'}}}{\PredPhi{A}{x}\impl \PredPhi{(A')}{x}}{\IH,(\ref{subtE2})}\label{subtE7}\\
	\NDLinePTG{\univQuant{x}{\PhiAppl{A'}} \PredPhi{A}{x}\impl \PredPhi{(A')}{x}}{\IH,(\ref{subtE2})}\label{subtE8}\\
	\NDLinePTG{\PredPhi{(A')}{t}}{\ruleRef{implE},\ruleRef{forallE},\ref{subtE8},(\ref{subtE6})}\nonumber
	\end{align}
	
	\subparagraph{\ruleRef{lambda'}:}	
	{\small
	\begin{align}
	\NDLineT{\concatCtx{\Gamma}{x:\PhiAppl{A}}}{t:B}{\byAss}\label{lambdap1}\\
	\NDLinePT{\concatCtx{\concatCtx{\Gamma}{x:\PhiAppl{A}}}{\typingAss{A}{x}}}{\PhiAppl{t}:\PhiAppl{B}}{\IH,\ref{PTctxVar},(\ref{lambdap1})}\label{lambdap2}\\
	\NDLinePT{\concatCtx{\concatCtx{\Gamma}{x:\PhiAppl{A}}}{\typingAss{A}{x}}}{\PredPhi{B}{\PhiAppl{t}}}{\IH,\ref{PTctxVar},(\ref{lambdap1})}\label{lambdap5}\\
	\NDLinePT{\concatCtx{\concatCtx{\Gamma}{x,y:\PhiAppl{A}}}{\namedass{xRy}{\termEqT{A}{x}{y}}}}{\termEqT{B}{\PhiAppl{t}}{\subst{\PhiAppl{t}}{x}{y}}}{(\ref{correct:substRelatTerms}),(\ref{lambdap5})}\label{lambdap5b}\\
	\Gamma\dedT \univQuant{x,y}{\PhiAppl{A}}&\termEqT{A}{x}{y}\impl \termEqT{B}{\PhiAppl{t}}{\subst{\PhiAppl{t}}{x}{y}} \QQQNegSp&&\QQuad \text{\ruleRef{implI},\ruleRef{forallI},(\ref{lambdap5})}\label{lambdap6}\\
	\NDLinePT{\concatCtx{\Gamma}{x:A}}{\PhiAppl{t}:\PhiAppl{B}}{\QQNegSp\QNegSp typing independent of assumptions,(\ref{lambdap2})}\label{lambdap3}\\
	\NDLinePTG{(\lambdaFun{x}{\PhiAppl{A}}\PhiAppl{t}):\PhiAppl{A}\to\PhiAppl{B}}{\ruleRef{lambda},(\ref{lambdap3})}\label{lambdap4}\\
	\NDLinePTG{\PhiAppl{\lambdaFun{x}{A}t}:\PhiAppl{\piType{x}{A}B}\QNegSp}{\QQuad\ref{PTLam},\ref{PTPitype},(\ref{lambdap4})}\nonumber\\
	\NDLinePTG{\PredPhi{(\piType{x}{A}B)}\ \PhiAppl{\lambdaFun{x}{A}t}\QQNegSp}{\QQuad\ref{PTPipred},(\ref{lambdap5b})}\nonumber
	\end{align}
	}
	
	\subparagraph{\ruleRef{appl'}:}
	\begin{align}
	\NDLineTG{f:\piType{x}{A}B}{\byAss}\label{applp1}\\
	\NDLineTG{t:A}{\byAss}\label{applp2}\\
	\NDLinePTG{\PhiAppl{f}:\PhiAppl{A}\to\PhiAppl{B}}{\IH,\ref{PTPitype},(\ref{applp1})}\label{applp3}\\
	\NDLinePTG{\PredPhi{(\piType{x}{A}B)}{\PhiAppl{f}}}{\IH,\ref{PTPitype},(\ref{applp1})}\label{applp6}\\
	\NDLinePTG{\univQuant{x}{\PhiAppl{A}}\univQuant{y}{\PhiAppl{A}}\nonumber\\&
		\termEqT{A}{x}{y}\impl \termEqT{B}{(\PhiAppl{f}\ x)}{(\PhiAppl{f}\ y)}}{\ref{PTPipred},(\ref{applp6})}\label{applp7}\\
	\NDLinePTG{\PhiAppl{t}:\PhiAppl{A}}{\IH,(\ref{applp2})}\label{applp4}\\
	\NDLinePTG{\PredPhi{A}{\PhiAppl{t}}}{\IH,(\ref{applp2})}\label{applp8}\\
	\NDLinePTG{\PredPhi{A}{\PhiAppl{t}}\impl \PredPhi{B}{(\PhiAppl{f}\ \PhiAppl{t})}}{\ruleRef{forallE},\ruleRef{forallE},(\ref{applp7}),(\ref{applp4}),(\ref{applp4})}\label{applp9}\\
	\NDLinePTG{\PredPhi{B}{(\PhiAppl{f}\ \PhiAppl{t})}}{\ruleRef{implE},(\ref{applp9}),(\ref{applp8})}\label{applp10}\\
	\NDLinePTG{\PhiAppl{f}\ \PhiAppl{t}:\PhiAppl{B}}{\ruleRef{appl},(\ref{applp3}),(\ref{applp4})}\label{applp5}\\
	\NDLinePTG{\PhiAppl{f\ t}:\PhiAppl{B}}{\ref{PTappl},(\ref{applp5})}\nonumber\\
	\NDLinePTG{\PredPhi{B}{\PhiAppl{f\ t}}}{\ref{PTappl},(\ref{applp9})}\nonumber
	\end{align}
	
	\subparagraph{\ruleRef{implType'}:}
	\begin{align}
	\NDLineTG{F:\bool}{\byAss}\label{implType1}\\
	\NDLineT{\concatCtx{\Gamma}{\namedass{ass}{F}}}{G:\bool}{\byAss}\label{implType2}\\
	\NDLinePTG{\PhiAppl{F}:\bool}{\IH,(\ref{implType1})}\label{implType3}\\
	\NDLinePT{\concatCtx{\PhiAppl{\Gamma}}{\PhiAppl{F}}}{\PhiAppl{G}:\bool}{\IH,(\ref{implType2})}\label{implType4}\\
	\NDLinePTG{\PhiAppl{G}:\bool}{typing is independent of assumptions,(\ref{implType4})}\label{implType5}\\
	\NDLinePTG{\PhiAppl{F}\impl \PhiAppl{G}:\bool}{\ruleRef{implType},(\ref{implType3}),(\ref{implType5})}\label{implType6}\\
	\NDLinePTG{\PhiAppl{F\impl G}:\bool}{\ref{PTImpl},(\ref{implType6})}\label{implType7}\\
	\NDLinePTG{\boolPred{\PhiAppl{F \impl G}}}{(\ref{PTTpBoolPred}),\ruleRef{refl},(\ref{implType7})}\nonumber
	\end{align}
	
	\subparagraph{\ruleRef{const''}:}
	\begin{align}
	&\thyIn{c:A'}{T}&&\text{\byAss}\label{const1}\\
	\NDLineTG{A'\subtyping A}{\byAss}\label{const2}\\
	&\thyIn{c:\PhiAppl{A'}}{\PhiAppl{T}}&&\text{\ref{PTtermDecl},(\ref{const1})}\label{const3}\\
	&\thyIn{\typingAx{A'}{c}}{\PhiAppl{T}}&&\text{\ref{PTtermDecl},(\ref{const1})}\label{const6}\\
	\NDLinePTG{\PhiAppl{A'}\typeEquals\PhiAppl{A}}{\IH,(\ref{const2})}\label{const4}\\
	\NDLinePT{\concatCtx{\PhiAppl{\Gamma}}{x,y:\PhiAppl{A'}}}{\termEqT{A'}{x}{y}\impl \termEqT{A}{x}{y}}{\IH,(\ref{const2})}\label{const4a}\\
	\NDLinePTG{\univQuant{x,y}{\PhiAppl{A}}\termEqT{A'}{x}{y}\impl \termEqT{A}{x}{y}}{\ruleRef{forallI},\ruleRef{forallI},(\ref{const4a})}\label{const4b}\\
	\NDLinePTG{c:\PhiAppl{A}}{\ruleRef{const},(\ref{const3}),(\ref{const4})}\label{const5}\\
	\NDLinePTG{\PhiAppl{c}:\PhiAppl{A}}{\ref{PTtermDecl},(\ref{const5})}\nonumber\\
	\NDLinePTG{\PredPhi{A'}{ \PhiAppl{c}}}{\ref{PTtermDecl},\ruleRef{axiom},(\ref{const6})}\label{const7}\\
	\NDLinePTG{\PredPhi{A}{\PhiAppl{c}}}{\ruleRef{forallE},\ruleRef{forallE},\ruleRef{implE},(\ref{const4b}),(\ref{const7})}\nonumber
	\end{align}
	
	\subparagraph{\ruleRef{var''}:}
	\begin{align}
	&\ctxIn{x:A'}{\Gamma}&&\text{\byAss}\label{var1}\\
	\NDLineTG{A'\subtyping A}{\byAss}\label{var2}\\
	&\ctxIn{c:\PhiAppl{A'}}{\PhiAppl{\Gamma}}&&\text{\ref{PTtermDecl},(\ref{var1})}\label{var3}\\
	&\ctxIn{\typingAss{A'}{c}}{\PhiAppl{\Gamma}}&&\text{\ref{PTtermDecl},(\ref{var1})}\label{var6}\\
	\NDLinePTG{\PhiAppl{A'}\typeEquals\PhiAppl{A}}{\IH,(\ref{var2})}\label{var4}\\
	\NDLinePT{\concatCtx{\PhiAppl{\Gamma}}{x,y:\PhiAppl{A'}}}{\termEqT{A'}{x}{y}\impl \termEqT{A}{x}{y}}{\IH,(\ref{var2})}\label{var4a}\\
	\NDLinePTG{\univQuant{x,y}{\PhiAppl{A}}\termEqT{A'}{x}{y}\impl \termEqT{A}{x}{y}}{\ruleRef{forallI},\ruleRef{forallI},(\ref{var4a})}\label{var4b}\\
	\NDLinePTG{c:\PhiAppl{A}}{\ruleRef{var},(\ref{var3}),(\ref{var4})}\label{var5}\\
	\NDLinePTG{\PhiAppl{c}:\PhiAppl{A}}{\ref{PTtermDecl},(\ref{var5})}\nonumber\\
	\NDLinePTG{\PredPhi{A'}{ \PhiAppl{c}}}{\ref{PTtermDecl},\ruleRef{assume},(\ref{var6})}\label{var7}\\
	\NDLinePTG{\PredPhi{A}{\PhiAppl{c}}}{\ruleRef{forallE},\ruleRef{forallE},\ruleRef{implE},(\ref{var4b}),(\ref{var7})}\nonumber
	\end{align}
		
	\subparagraph{\ruleRef{eqType}:}
	\begin{align}
	\NDLineTG{s:A}{\byAss}\label{eqType1}\\
	\NDLineTG{t:A}{\byAss}\label{eqType2}\\
	\NDLinePTG{\PhiAppl{s}:\PhiAppl{A}}{\IH,(\ref{eqType1})}\label{eqType3}\\
	\NDLinePTG{\PhiAppl{t}:\PhiAppl{A}}{\IH,(\ref{eqType2})}\label{eqType4}\\
	\NDLinePTG{\PredPhi{A}{\PhiAppl{s}}}{\IH,(\ref{eqType1})}\label{eqType5}\\	
	\NDLineTG{\Type{A}}{\ruleRef{typingTp},(\ref{eqType1})}\label{eqType11}
	\intertext{Since we are proving completeness w.r.t. the individual judgements one by one, we may assume here that the translation is complete w.r.t. those judgements, for which we already proved completeness. We are therefore allowed to use completeness w.r.t. well-formedness of types here. }
	\NDLinePTG{\PredPhiName{A}:\PhiAppl{A}\to\PhiAppl{A}\to\bool}{completeness w.r.t. well-formedness of $A$,(\ref{eqType11})}\label{eqType12}\\
	\NDLinePTG{\termEqT{A}{\PhiAppl{s}}{\PhiAppl{t}}:\bool}{\ruleRef{appl},\ruleRef{appl},(\ref{eqType12}),(\ref{eqType3}),(\ref{eqType4})}\label{eqType13}\\
	\NDLinePTG{\boolPred{\termEqT{A}{\PhiAppl{s}}{\PhiAppl{t}}}}{(\ref{PTTpBoolPred}),\ruleRef{refl},(\ref{eqType13})}\nonumber
	\end{align}
	\paragraph{Term equality}\label{par:soundnessTmEq}
	Fix a context. 
	By rule \ruleRef{rewrite}, if we can show for two \dphol{} terms $s,t:A$ that $\PhiAppl{s}\termEquals{\PhiAppl{A}}\PhiAppl{t}$ and additionally that $\PredPhi{A}{\PhiAppl{s}}$, then $\PredPhi{A}{\PhiAppl{t}}$ and $\termEqT{A}{\PhiAppl{s}}{\PhiAppl{t}}$ follow. 
	By rule \ruleRef{eqTyping} and rule \ruleRef{sym} we further yield $\PhiAppl{s}:\PhiAppl{A}$ and $\PhiAppl{t}:\PhiAppl{A}$.
	This reduces the completeness claim for a term-equality $s\termEquals{A}t$ to showing $\PhiAppl{s}\termEquals{\PhiAppl{A}}\PhiAppl{t}$ and $\PredPhi{A}{\PhiAppl{s}}$.
	
	Term equality can be shown using the rules  \ruleRef{congAppl'}, \ruleRef{congLam'}, \ruleRef{etaPi}, \ruleRef{refl}, \ruleRef{sym} and \ruleRef{beta} in \dphol{}.
	
	\subparagraph{\ruleRef{congAppl'}}
	\begin{align}
	\NDLineTG{t\termEquals{A}t'}{\byAss}\label{congApplp1}\\
	\NDLineTG{f\termEquals{\piType{x}{A}B}f'}{\byAss}\label{congApplp2}\\
	\NDLinePTG{\termEqT{A}{\PhiAppl{t}}{\PhiAppl{t'}}}{\IH,(\ref{congApplp1})}\label{congApplp3}\\
	\NDLinePTG{\univQuant{x}{\PhiAppl{A}}\univQuant{y}{\PhiAppl{A}}\termEqT{A}{x}{y}\impl\nonumber\\& \termEqT{(\piType{z}{A}B)}{\PhiAppl{f}\ x}{\PhiAppl{f'}\ x}}{\IH,(\ref{congApplp2})}\label{congApplp4}\\
	\NDLinePTG{\PhiAppl{t}:\PhiAppl{A}}{\IH,(\ref{congApplp1})}\label{congApplp5}\\
	\NDLinePTG{\PhiAppl{t'}:\PhiAppl{A}}{\IH,(\ref{congApplp1})}\label{congApplp6}\\
	\NDLinePTG{\termEqT{A}{\PhiAppl{t}}{\PhiAppl{t'}}\impl \termEqT{(\piType{z}{A}B)}{\PhiAppl{f}\ \PhiAppl{t}}{\PhiAppl{f'}\ \PhiAppl{t'}}}{\ruleRef{forallE},\ruleRef{forallE},(\ref{congApplp4}),(\ref{congApplp5}),(\ref{congApplp6})}\label{congApplp7}\\
	\NDLinePTG{\PhiAppl{f}:\PhiAppl{A}\to\PhiAppl{B}}{\IH,(\ref{congApplp2})}\label{congApplp8}\\
	\NDLinePTG{\PhiAppl{f'}:\PhiAppl{A}\to\PhiAppl{B}}{\IH,(\ref{congApplp2})}\label{congApplp9}\\
	\NDLinePTG{\termEqT{(\piType{z}{A}B)}{\PhiAppl{f}\ \PhiAppl{t}}{\PhiAppl{f'}\ \PhiAppl{t'}}}{\ruleRef{implE},(\ref{congApplp7}),(\ref{congApplp3})}\nonumber\\
	\NDLinePTG{\PhiAppl{f}\ \PhiAppl{t}:\PhiAppl{B}}{\ruleRef{appl},(\ref{congApplp8}),(\ref{congApplp5})}\nonumber\\
	\NDLinePTG{\PhiAppl{f'}\ \PhiAppl{t'}:\PhiAppl{B}}{\ruleRef{appl},(\ref{congApplp8}),(\ref{congApplp5})}\nonumber
	\end{align}
	
	\subparagraph{\ruleRef{congLam'}}
	This case will use (\ref{correct:substRelatTerms}). 
	
	{\small
	\begin{align}
	\NDLineTG{A\typeEquals A'}{\byAss}\label{congLamp1}\\
	\NDLineT{\concatCtx{\Gamma}{x:A}}{t\termEquals{B}t'}{\byAss}\label{congLamp2}\\
	\NDLinePTG{\PhiAppl{A}\typeEquals\PhiAppl{A'}}{\IH,(\ref{congLamp1})}\label{congLamp3}\\	
	\NDLinePT{\concatCtx{\concatCtx{\PhiAppl{\Gamma}}{x:\PhiAppl{A}}}{\typingAss{A}{x}}}{\termEqT{B}{\PhiAppl{t}}{\PhiAppl{t'}}}{\IH,(\ref{congLamp2})}\label{congLamp4}\\
	\NDLinePT{\concatCtx{\concatCtx{\PhiAppl{\Gamma}}{z:\PhiAppl{A}}}{\typingAss{A}{z}}}{\termEqT{B}{\subst{\PhiAppl{t}}{x}{z}}{\subst{\PhiAppl{t'}}{x}{z}}}{$\alpha$-renaming,(\ref{congLamp4})}\label{congLamp5}\\
	\NDLinePT{\concatCtx{\PhiAppl{\Gamma}}{z:\PhiAppl{A}}}{\PredPhi{A}{z}\impl\nonumber\\&\termEqT{B}{\subst{\PhiAppl{t}}{x}{z}}{\subst{\PhiAppl{t'}}{x}{z}}}{\ruleRef{implI},(\ref{congLamp5})}\label{congLamp5a}\\
	\NDLinePTG{\univQuant{z}{\PhiAppl{A}} \PredPhi{A}{z}\impl\nonumber\\&\termEqT{B}{\subst{\PhiAppl{t}}{x}{z}}{\subst{\PhiAppl{t'}}{x}{z}}}{\ruleRef{forallI},(\ref{congLamp5a})}\label{congLamp5b}\\
	\NDLinePT{\concatCtx{\PhiAppl{\Gamma}}{x,y:\PhiAppl{A}}}{\univQuant{z}{\PhiAppl{A}} \PredPhi{A}{z}\impl\nonumber\\&\termEqT{B}{\subst{\PhiAppl{t}}{x}{z}}{\subst{\PhiAppl{t'}}{x}{z}}}{\ruleRef{varDed},\ruleRef{varDed},(\ref{congLamp5b})}\label{congLamp5c}\\
	\NDLinePT{\concatCtx{\PhiAppl{\Gamma}}{x,y:\PhiAppl{A}}}{\PredPhi{A}{x}\impl\termEqT{B}{\PhiAppl{t}}{\PhiAppl{t'}}}{\ruleRef{varDed},\ruleRef{varDed},(\ref{congLamp5c})}\label{congLamp5d}\\
	\NDLinePT{\concatCtx{\concatCtx{\PhiAppl{\Gamma}}{x,y:\PhiAppl{A}}}{\namedass{xRy}{\termEqT{A}{x}{y}}}}{\PredPhi{A}{x}\impl\termEqT{B}{\PhiAppl{t}}{\PhiAppl{t'}}}{\ruleRef{assDed},(\ref{congLamp5d})}\label{congLamp5e}\\
	\NDLinePT{\concatCtx{\concatCtx{\PhiAppl{\Gamma}}{x,y:\PhiAppl{A}}}{\namedass{xRy}{\termEqT{A}{x}{y}}}}{\PredPhi{A}{x}}{(\ref{correct:substRelatTerms}),\ruleRef{assume},\ruleRef{assume}}\label{congLamp6}\\
	\NDLinePT{\concatCtx{\concatCtx{\PhiAppl{\Gamma}}{x,y:\PhiAppl{A}}}{\namedass{xRy}{\termEqT{A}{x}{y}}}}{\termEqT{B}{\PhiAppl{t}}{\PhiAppl{t'}}}{\ruleRef{implE},(\ref{congLamp5e}),(\ref{congLamp6})}\label{congLamp7}\\
	\NDLinePT{\concatCtx{\concatCtx{\PhiAppl{\Gamma}}{x,y:\PhiAppl{A}}}{\namedass{xRy}{\termEqT{A}{x}{y}}}}{\termEqT{B}{\PhiAppl{t}}{\subst{\PhiAppl{t'}}{x}{y}}}{(\ref{correct:substRelatTerms}),(\ref{congLamp7}),\ruleRef{assume}}\label{congLamp8}\\
	\NDLinePT{\concatCtx{\PhiAppl{\Gamma}}{x,y:\PhiAppl{A}}}{\termEqT{A}{x}{y}\impl \termEqT{B}{\PhiAppl{t}}{\subst{\PhiAppl{t'}}{x}{y}}}{\ruleRef{implI},(\ref{congLamp8})}\label{congLamp9}\\
	\NDLinePTG{\univQuant{x}{\PhiAppl{A}}\univQuant{y}{\PhiAppl{A}}\termEqT{A}{x}{y}\nonumber\\&\impl \termEqT{B}{\PhiAppl{t}}{\subst{\PhiAppl{t'}}{x}{y}}}{\ruleRef{forallI},\ruleRef{forallI},(\ref{congLamp9})}\label{congLamp10}\\
	\NDLinePT{\concatCtx{\concatCtx{\PhiAppl{\Gamma}}{x:\PhiAppl{A}}}{\typingAss{A}{x}}}{\PhiAppl{t}:\PhiAppl{B}}{\IH,(\ref{congApplp2})}\label{congApplp11}\\
	\NDLinePT{\concatCtx{\concatCtx{\PhiAppl{\Gamma}}{x:\PhiAppl{A}}}{\typingAss{A}{x}}}{\PhiAppl{t'}:\PhiAppl{B}}{\IH,(\ref{congApplp2})}\label{congApplp12}\\
	\intertext{Since in \hol{} typing is independent of context assumptions:}
	\NDLinePT{\concatCtx{\PhiAppl{\Gamma}}{x:\PhiAppl{A}}}{\PhiAppl{t}:\PhiAppl{B}}{explanation,(\ref{congApplp11})}\label{congApplp13}\\
	\NDLinePT{\concatCtx{\PhiAppl{\Gamma}}{x:\PhiAppl{A}}}{\PhiAppl{t'}:\PhiAppl{B}}{explanation,(\ref{congApplp12})}\label{congApplp14}\\
	\NDLinePTG{\lambdaFun{x}{\PhiAppl{A}}\PhiAppl{t}:\PhiAppl{A}\to\PhiAppl{B}}{\ruleRef{lambda},(\ref{congApplp13})}\label{congApplp15}\\
	\NDLinePTG{\lambdaFun{x}{\PhiAppl{A}}\PhiAppl{t'}:\PhiAppl{A}\to\PhiAppl{B}}{\ruleRef{lambda},(\ref{congApplp14})}\label{congApplp16}\\
	\NDLinePTG{\PhiAppl{\lambdaFun{x}{A}t\termEquals{\piType{x}{A}B}\lambdaFun{x}{A'}t'}\QNegSp}{\Quad(\ref{PTEq}),(\ref{congLamp10})}\nonumber\\
	\NDLinePTG{\PhiAppl{\lambdaFun{x}{A}t}:\PhiAppl{\piType{x}{A}B}}{\Quad(\ref{PTPitype}),(\ref{PTLam}),(\ref{congApplp15})}\nonumber\\
	\NDLinePTG{\PhiAppl{\lambdaFun{x}{A}t'}:\PhiAppl{\piType{x}{A}B}}{\Quad(\ref{PTPitype}),(\ref{PTLam}),(\ref{congApplp16})}\nonumber
	\end{align}
	}
	
	\subparagraph{\ruleRef{etaPi}}
	\begin{align}
	\NDLineTG{t:\piType{x}{A}B}{\byAss}\label{etaPi1}\\
	\NDLinePTG{\PhiAppl{t}:\PhiAppl{A}\to\PhiAppl{B}}{\ref{PTPitype},\IH,(\ref{etaPi1})}\label{etaPi2}\\
	\NDLinePTG{\PhiAppl{t}\termEquals{\PhiAppl{A}\to\PhiAppl{B}}\lambdaFun{x}{\PhiAppl{A}}\PhiAppl{t}\ x}{\ruleRef{eta},(\ref{etaPi2})}\label{etaPiConc}\\
	\NDLinePTG{\PhiAppl{t}\termEquals{\PhiAppl{\piType{x}{A}B}}\PhiAppl{\lambdaFun{x}{A}t\ x}}{\ref{PTLam},\ref{PTPitype},(\ref{etaPiConc})}\nonumber\\
	\NDLinePTG{\PredPhi{\left(\piType{x}{A}B\right)}{\PhiAppl{t}}}{\IH,(\ref{etaPi2})}\nonumber
	\end{align}
	
	\subparagraph{\ruleRef{refl}}
	\begin{align}
	\NDLineTG{t:A}{\byAss}\label{refl1}\\
	\NDLinePTG{\PhiAppl{t}:\PhiAppl{A}}{\IH,(\ref{refl1})}\label{refl2}\\
	\NDLinePTG{\PhiAppl{t}\termEquals{\PhiAppl{A}}\PhiAppl{t}}{\ruleRef{refl},(\ref{refl2})}\nonumber\\
	\NDLinePTG{\PredPhi{A}{\PhiAppl{t}}}{\IH,(\ref{refl1})}\nonumber
	\end{align}
	
	\subparagraph{\ruleRef{sym}}
	\begin{align}
	\NDLineTG{s\termEquals{A}t}{\byAss}\label{sym1}\\
	\NDLinePTG{\termEqT{A}{\PhiAppl{s}}{\PhiAppl{t}}}{\IH,(\ref{sym1})}\label{sym2}\\
	\NDLinePTG{\PhiAppl{t}:\PhiAppl{A}}{\IH,(\ref{sym1})}\label{sym3}\\
	\NDLinePTG{\PhiAppl{s}:\PhiAppl{A}}{\IH,(\ref{sym1})}\label{sym4}\\
	\NDLinePTG{\termEqT{A}{\PhiAppl{t}}{\PhiAppl{s}}}{\ruleRef{forallE},\ruleRef{forallE},\ruleRef{implE},(\ref{correct:relatSym}),(\ref{sym2}),(\ref{sym3}),(\ref{sym4})}\nonumber
	\end{align}
		
	\subparagraph{\ruleRef{beta}}
	\begin{align}
	\NDLineTG{(\lambdaFun{x}{A}s)\ t:B}{\byAss}\label{beta1}\\
	\NDLinePTG{(\lambdaFun{x}{\PhiAppl{A}}\PhiAppl{s})\ \PhiAppl{t}:\PhiAppl{B}}{\IH,\ref{PTLam},(\ref{beta1})}\label{beta2}\\
	\NDLinePTG{(\lambdaFun{x}{\PhiAppl{A}}\PhiAppl{s})\ \PhiAppl{t}\termEquals{\PhiAppl{B}}\subst{\PhiAppl{s}}{x}{\PhiAppl{t}}}{\ruleRef{beta},(\ref{beta2})}\label{betaPConc}\\
	\NDLinePTG{(\lambdaFun{x}{\PhiAppl{A}}\PhiAppl{s})\ \PhiAppl{t}\termEquals{\PhiAppl{B}}\PhiAppl{\subst{s}{x}{t}}}{(\ref{correct:substTerm}),(\ref{betaPConc})}\label{betaConc}\\
	\NDLinePTG{\PhiAppl{(\lambdaFun{x}{A}s)\ t}\termEquals{\PhiAppl{B}}\PhiAppl{\subst{s}{x}{t}}\QQNegSp}{\QQuad\ref{PTLam},\ref{PTappl},(\ref{betaConc})}\nonumber\\
	\NDLinePTG{\PredPhi{\left(\piType{x}{A}B\right)}{\left((\lambdaFun{x}{A}s)\ t\right)}\QQNegSp}{\QQuad\IH,(\ref{beta1})}\nonumber
	\end{align}
	
	\paragraph{Validity}
	Validity can be shown using the rules \ruleRef{psubE}, \ruleRef{axiom}, \ruleRef{assume}, \ruleRef{congDed}, \ruleRef{implI} and \ruleRef{implE}.
	
	\subparagraph{\ruleRef{psubE}}
	\begin{align}
	\NDLineTG{t:\subtype{A}{p}}{\byAss}\label{psubE1}\\
	\NDLinePTG{\PredPhi{\left(\subtype{A}{p}\right)}{\PhiAppl{t}}}{\IH,(\ref{psubE1})}\label{psubE2a}\\		
	\NDLinePTG{\PhiAppl{p}\ \PhiAppl{t}}{\ruleRef{andEr},\ruleRef{andEr},\ref{PTPSpred},(\ref{psubE2a})}\label{psubE2}\\
	\NDLinePTG{\PhiAppl{p\ t}}{\ref{PTappl},(\ref{psubE2})}\nonumber
	\end{align}
	
	\subparagraph{\ruleRef{axiom}}
	\begin{align}
	&\thyIn{\namedax{ax}{F}}{T}&&\text{\byAss}\label{axiom1}\\
	\NDLineT{}{\Ctx{\Gamma}}{\byAss}\label{axiom2}\\
	&\thyIn{\namedax{ax}{\PhiAppl{F}}}{\PhiAppl{T}}&&\text{\ref{PTax},(\ref{axiom1}}\label{axiom3}\\
	\NDLinePT{}{\Ctx{\PhiAppl{\Gamma}}}{\IH,\ref{axiom2}}\label{axiom4}\\
	\NDLinePTG{\PhiAppl{F}}{\ruleRef{axiom},(\ref{axiom3}),(\ref{axiom4})}\nonumber
	\end{align}
	
	\subparagraph{\ruleRef{assume}}
	\begin{align}
	&\ctxIn{\namedass{ass}{F}}{\Gamma}&&\text{\byAss}\label{assume1}\\
	\NDLineT{}{\Ctx{\Gamma}}{\byAss}\label{assume2}\\
	&\ctxIn{\namedass{ass}{\PhiAppl{F}}}{\PhiAppl{\Gamma}}&&\text{\ref{PTctxAss},(\ref{assume1}}\label{assume3}\\
	\NDLinePT{}{\Ctx{\PhiAppl{\Gamma}}}{\IH,\ref{assume2}}\label{assume4}\\
	\NDLinePTG{\PhiAppl{F}}{\ruleRef{assume},(\ref{assume3}),(\ref{assume4})}\nonumber
	\end{align}
	
	\subparagraph{\ruleRef{congDed}}
	\begin{align}
	\NDLineTG{F\termEqB F'}{\byAss}\label{congDed1}\\
	\NDLineTG{F'}{\byAss}\label{congDed2}\\
	\NDLinePTG{\PhiAppl{F}\termEqB\PhiAppl{F'}}{(\ref{PTTpBoolPred}),\IH,(\ref{congDed1})}\label{congDed3}\\
	\NDLinePTG{\PhiAppl{F'}}{\IH,(\ref{congDed2})}\label{congDed4}\\
	\NDLinePTG{\PhiAppl{F}}{\ruleRef{congDed},(\ref{congDed3}),(\ref{congDed4})}\nonumber
	\end{align}
	
	\subparagraph{\ruleRef{implI}}
	\begin{align}
	\NDLineTG{F:\bool}{\byAss}\label{implI1}\\
	\NDLineT{\concatCtx{\Gamma}{\namedass{ass}{F}}}{G}{\byAss}\label{implI2}\\
	\NDLinePTG{\PhiAppl{F}:\bool}{\IH,(\ref{implI1})}\label{implI3}\\
	\NDLinePT{\concatCtx{\PhiAppl{\Gamma}}{\PhiAppl{F}}}{\PhiAppl{G}}{\IH,\ref{PTctxAss},(\ref{implI2})}\label{implI4}\\
	\NDLinePTG{\PhiAppl{F}\impl\PhiAppl{G}}{\ruleRef{implI},(\ref{implI3}),(\ref{implI4})}\label{implI5}\\
	\NDLinePTG{\PhiAppl{F\impl G}}{\ref{PTImpl},(\ref{implI5})}\nonumber
	\end{align}
	
	\subparagraph{\ruleRef{implE}}
	\begin{align}
	\NDLineTG{F\impl G}{\byAss}\label{implE1}\\
	\NDLineTG{F}{\byAss}\label{implE2}\\
	\NDLinePTG{\PhiAppl{F}\impl \PhiAppl{G}}{\IH,\ref{PTImpl},(\ref{implE1})}\label{implE3}\\
	\NDLinePTG{\PhiAppl{F}}{\IH,(\ref{implE2})}\label{implE4}\\
	\NDLinePTG{\PhiAppl{G}}{\ruleRef{implE},(\ref{implE3}),(\ref{implE4})}\nonumber
	\end{align}
	\subparagraph{\ruleRef{boolExt}}
	\begin{align}
	\NDLineTG{p\ \T}{\byAss}\label{boolExt1}\\
	\NDLineTG{p\ \F}{\byAss}\label{boolExt2}\\
	\NDLinePTG{\PhiAppl{p}\ \T}{\IH,\ref{PTappl},(\ref{boolExt1})}\label{boolExt3}\\
	\NDLinePTG{\PhiAppl{p}\ \F}{\IH,\ref{PTappl},(\ref{boolExt2})}\label{boolExt4}\\
	\NDLinePTG{\univQuant{z}{\bool}\PhiAppl{p}\ z}{\ruleRef{boolExt},(\ref{boolExt3}),(\ref{boolExt4})}\label{boolExt5}\\
	\NDLinePT{\concatCtx{\PhiAppl{\Gamma}}{x:\bool}}{\univQuant{z}{\bool}\PhiAppl{p}\ z}{\ruleRef{varDed},(\ref{boolExt5})}\label{boolExt6}\\
	\NDLinePT{\concatCtx{\PhiAppl{\Gamma}}{x:\bool}}{\PhiAppl{p}\ x}{\ruleRef{forallE},(\ref{boolExt6}),\ruleRef{assume}}\label{boolExt7}\\
	\concatCtx{\concatCtx{\PhiAppl{\Gamma}}{x,y:&\bool}}{x \termEqB y}\dedT \PhiAppl{p}\ x&&\text{\ruleRef{assDed},\ruleRef{varDed},\ruleRef{boolExt7}}\label{boolExt9}\\	
	\concatCtx{\concatCtx{\PhiAppl{\Gamma}}{x,y:&\bool}}{x \termEqB y}\dedT \PhiAppl{p}\ y&&\text{\ruleRef{rewrite},\ruleRef{boolExt9},\ruleRef{assume}}\label{boolExt10}\\	
	\concatCtx{\PhiAppl{\Gamma}}{x,y:&\bool}\dedT \termEqT{\bool}{x}{y}\impl\PhiAppl{p}\ y&&\text{(\ref{PTTpBoolPred}),\ruleRef{implI},\ruleRef{boolExt10}}\label{boolExt11}\\
	\NDLinePTG{\univQuant{x}{\bool}\univQuant{y}{\bool} \nonumber\\&\termEqT{\bool}{x}{y}\impl\PhiAppl{p}\ y}{\ruleRef{forallI},\ruleRef{forallI},(\ref{boolExt11})}\label{boolExtPre}\\
	\NDLinePTG{\PhiAppl{\univQuant{x}{\bool}p\ x}}{(\ref{PTEq}),(\ref{PTPipred}),(\ref{boolExtPre})}\nonumber
	\end{align}
\end{proof}

\section{Soundness proof}\label{sec:soundness}

The idea of the soundness proof is to transform HOL-proofs into DHOL-proofs.
The proof is very involved, and we proceed in multiple steps.

\subsection{Type-wise injectivity of the translation}\label{sec:injectivity}

\begin{definition}
	Let $t$ be an ill-typed \dphol{} term with well-typed image $\PhiAppl{t}$ in \hol. In this case we will say that $\PhiAppl{t}$ is a \emph{spurious} term. 
	A term $\PhiAppl{s}$ in \hol{} that is the image of a well-typed term $s$, will be called \emph{proper}. A term $tm$ in \hol{} that is not the image of any (well-typed or not) term is said to be \emph{improper}. 
%
%
%
%
\end{definition}

\begin{lemma}\label{lem:termwise-inj}
	Let $\Delta$ be a \dphol{} context and let $\Gamma$ denote its translation.
	Given two \dphol{} terms $s, t$ of type $A$ and assuming $s$ and $t$ are not identical, it follows that $\PhiAppl{s}$ and $\PhiAppl{t}$ are not identical.
\end{lemma}

\begin{proof}[Proof of Lemma~\ref{lem:termwise-inj}]
	We prove this by induction on the shape of the types both equalities are over\,---\,in case both terms are equalities\,---\,and by subinduction on the shape of the two translated terms otherwise.
	We observe that terms created using a different top-level production are non-identical and will remain that way in the image. 
	So we can go over the productions one by one and assuming type-wise injectivity for subterms show injectivity of applying them.
	Different constants are mapped to different constants and different variables to different variables, so in those cases there is nothing to prove. 
	If two function applications or implications differ in \dphol{} then one of the two pairs of corresponding arguments must differ as well. By induction hypothesis so will the images of the terms in that pair. Since function application and implication both commute with the translation, it follows that the images of the function applications or implications also differ. 
	Since the translations of the terms on both sides of an equality also show up in the translation, the same argument also works for two equalities over the same type. Similarly for lambda functions of same type.
	
	Consider now two equalities over different types that get identified by dependency-erasure. 
	
	In case of equalities over different base types, the typing relations that are applied in the images are different, so the images of the equalities differ.
	For equalities over different $\Pi$-types either the domain type or the codomain type must differ by rule \ruleRef{congPi}. 
	If the domain types differ then the typing assumption after the two universal quantifiers of the translated equalities will differ. 
	If the codomain types are different then the applications of the typing relations on the right of the $\impl$ of the translated equalities are the translations of the equalities yielded by applying the functions on both sides of the equalities to a freshly bound variable of the domain type.
	The translations of the equalities are only identical if those "inner  equalities" are identical. Furthermore, the inner equalities are over types that are the codomain of the type the equalities are over. The claim then follows from the induction hypothesis. 
	
	Finally it remains to consider the case of equalities $s\termEquals{\subtype{A}{p}}t$ and $s'\termEquals{\subtype{A'}{p'}}t'$ over non-identical predicate subtypes $\subtype{A}{p}$ and $\subtype{A'}{p'}$ where not both $A=A'$ and $p=p'$. 
	If $p\neq p'$, then the translations have different subterms $\PhiAppl{p}\ s$ and $\PhiAppl{p'}\ s'$ and thus differ.
	If $A\neq A'$, then the first conjuncts in the translated equalities are the translations of equalities over the types $A$ and $A'$ respectively, which by the induction hypothesis have different translations.
	So in any case, the equalities have different images.
\end{proof}

\subsection{Quasi-preimages for terms and validity statements in admissible \hol{} derivations}\label{sec:quasi-preimages}
Firstly, we will consider the preimage of a typing relations $\PredPhiName{A}$ to be the equality symbol $\lambdaFun{x}{A}\lambdaFun{y}{A} x\termEquals{A}y$ (if equality is treated as a (parametric) binary predicate rather than a production of the grammar this eta reduces to the symbol $\termEquals{A}$). 

Using this convention, we define the normalization of an improper \hol{} term, which is either a proper term or a spurious term.
The normalization of an improper \hol{} term is defined by:
\begin{definition}\label{defn:normalization}
	Let $t$ be an improper \hol{} term. Then we define the normalization $\norm{t}$ of $t$ by induction on the shape of $t$:
	
	\begin{align*}
	\norm{\PhiAppl{t}}&:=t\plabel{PTnormProper}\\
	\norm{\norm{s}}&:=\norm{s}\plabel{PTnormNorm}\\	
	\norm{\PredPhiName{A}\ s}&:=\lambdaFun{y}{\PhiAppl{A}} \termEqT{A}{s}{y}\plabel{normRelAppl}\\
	\norm{\PredPhiName{A}}&:=\lambdaFun{x}{\PhiAppl{A}}\lambdaFun{y}{\PhiAppl{A}} \termEqT{A}{x}{y}\plabel{normRel}\\
	\norm{c}&:=c\plabel{normConst}\\
	\norm{x}&:=x\plabel{normVar}\\
	\norm{f\ t}&:=\norm{f}\ \norm{t}\plabel{normAppl}\\
	\norm{\lambdaFun{x}{C}t}&:=\lambdaFun{x}{C}\norm{t}\plabel{normLam}
	\intertext{If $F$ not of shape $\PredPhiName{A} \_\ \_\impl \_$ or $\univQuant{x'}{\PhiAppl{A}}\termEqT{A}{x}{x'}\impl\_$:}
	\norm{\univQuant{x}{\PhiAppl{A}}F}&:=\norm{\univQuant{x}{\PhiAppl{A}}\PredPhi{A}{x}F}\plabel{normUniv}\\
	\norm{\univQuant{x}{\PhiAppl{A}}\PredPhi{A}{x}\impl G}&:=\univQuant{x,x'}{\PhiAppl{A}}\termEqT{A}{x}{x'}\impl G\plabel{normAltTransUniv}\\
	\norm{s\termEquals{\PhiAppl{A}}t}&:=\termEqT{A}{s}{t} \plabel{normEq}\\
	\norm{s\impl t}&:=\norm{s}\impl \norm{t} \plabel{normImpl}
	\end{align*}
	For proper and spurious terms $t$, we define the normalization of $t$ be be $t$ itself.
\end{definition}

\begin{definition}	
	Assume a well-formed \dphol{} theory $T$. 
	
	We say that an \hol{} context $\Delta$ is \emph{proper} (relative to $\PhiAppl{T}$), iff there exists a well-formed \hol{} context  $\Theta$ (relative to $\PhiAppl{T}$), s.t. there is a well-formed \dphol{} context $\Gamma$ (relative to $T$) with $\PhiAppl{\Gamma}=\Theta$ and $\Theta$ can be obtained from $\Delta$ by adding well-typed typing assumptions.
	In this case, $\Gamma$ is called a \emph{quasi-preimage} of $\Delta$.
	Inspecting the translation, it becomes clear that $\Gamma$ is uniquely determined by the choices of the preimages of the types of variables without a typing assumption in $\Delta$.
	
	Given a proper \hol{} context $\Delta$ and a well-typed \hol{} formula $\phi$ over $\Delta$, we say that $\phi$ is \emph{quasi-proper} iff $\norm{\phi}=\PhiAppl{F}$ for $\Gamma\dedT F:\bool$ 
	and $\Gamma$ is a quasi-preimage of $\Delta$.
	In that case, we call $F$ a \emph{quasi-preimage} of $\phi$.
	\par
	
	Finally, we call a validity judgement $\Delta\dedPT \phi$ in \hol{} \emph{proper} iff\begin{enumerate}
		\item $\Delta$ is proper,
		\item $\phi$ is quasi-proper in context $\Delta$
	\end{enumerate}
	In this case, we will call $\PhiAppl{\Gamma}\dedPT \PhiAppl{F}$ a relativization of $\Delta\dedPT \phi$ and $\Gamma\dedT F$ a \emph{quasi-preimage} of the statement $\Delta\dedPT \phi$, where $\Gamma$ is a quasi-preimage of $\Delta$ and $F$ a quasi-preimage of $\phi$. 
\end{definition}

\subsection{Transforming \hol{} derivations into admissible \hol{} derivations}\label{sec:prfTransform}
It will be useful to distinguish between two different kinds of improper terms.
\begin{definition}
	An improper term is called \emph{almost proper} iff its normalization isn't spurious and contains no spurious subterms, otherwise it is said to be \emph{unnormalizably spurious}.
	This means that improper terms are almost proper iff they have a well-typed quasi-preimage (i.e. if they are quasi-proper).
	We consider proper terms to be \emph{almost proper} as well.
\end{definition}

\begin{definition}\label{defn:admissibleDerivation}
	A valid \hol{} derivation is called \emph{admissible} iff all terms occuring in it are almost proper.
\end{definition}
In the following we describe a proof transformation which maps \hol{} derivations to admissible \hol{} derivations. 

\begin{definition}\label{defn:statement-transformation}
	A \emph{statement transformation} in a given logic is a map that maps statements in the logic to statements in the logic.
\end{definition}
\begin{definition}\label{defn:macro-step}
	A \emph{macro-step} $M$ \emph{for} a statement transformation \emph{$T$} \emph{replacing} a step \emph{$S$} in a derivation is a sequence of steps $S_1,\ldots, S_n$ (called \emph{micro-steps} of $M$) s.t. the assumptions of the $S_i$ that are not concluded by $S_j$ with $j<i$ are results of applying $T$ to assumptions of step $S$ and furthermore the conclusion of step $S_n$ is the result of applying $T$ to the conclusion of $S$. The assumptions of those $S_j$ that are not concluded by previous micro-steps of $M$ are called the \emph{assumptions of macro-step} $M$ and the conclusion of the last micro-step $S_n$ of $M$ is called the \emph{conclusion of macro-step} $M$. 
\end{definition}
\begin{definition}\label{defn:normalizingPrfTransform}
	A \emph{normalizing statement transformation} $\sRed{\cdot}$ is defined to be a transformation that replaces terms in statements (including their contexts) as described below. 
	The definition of the transformation of a term additionally depends on a \dphol-type $A$ (called the \emph{preimage type}) for each term $t$.
	We will write those types as indices to the \hol{} terms, so for instance $t_A$ would indicate an \hol{} term $t$ of type $\PhiAppl{A}$ and preimage type $A$. 
	
	These indices are used to effectively associate to each term a type of a possible quasi-preimage, which is useful as for $\lambda$-functions there are quasi-preimages of potentially many different types.
	We require that for an indexed term $t_A$, term $t$ has type $\PhiAppl{A}$ and that for almost proper terms $t_A$ with unique quasi-preimage the quasi-preimage has type $A$. 
	
	Intuitively the transformation will do two things (in this order) in order to "normalize" unnormalizably spurious terms to almost proper ones:
	\begin{enumerate}
		\item apply beta and eta reductions and in case this doesn't yield almost proper terms
		\item replace unnormalizably spurious function applications of type $B$ by the "default terms" $\defaultTerm{B}$ of type $B$ which is proper and whose existence is assumed for all \hol{} types.
	\end{enumerate}
	Additionally we will choose the index types to ensure we yield almost proper terms of same index on both sides of equalities.
	
	As we are assuming a valid derivation, we will only define this transformation on well-typed \hol{} terms.
	We can then define the transformation of $t_A$ (denoted by $\sRed{t_A}$) by induction on the shape of $t_A$ as follows:
	{\begin{align}
		&\sRed{t_A}&:=&t_A&&\text{if $t$ has quasi-preimage of type $A$}\sredlabel{sredP}\\
		&\sRed{f_{\piType{x}{A}B}\ t_{A}}\negSp\!\!&:=&\sRed{\sRed{f_{\piType{x}{A}B}}\ \sRed{t_{A}}}\QQQQNegSp&&\nonumber\\& &&
		&&\text{if $f_{\piType{x}{A}B}\ t_{A}$ not beta or eta reducible}\sredlabel{sredApplP}
		\end{align}}
	In the following cases, we assume that the term $t_A$ in $\sRed{\cdot}$ on the left of $:=$ isn't almost proper with a quasi-preimage of type $A$: 
	{\small
	\begin{align}
	&\sRed{t_A}&:=&\sRed{\betaEtaRed{t_A}} \quad\negSp\text{if $t$ is beta or eta reducible}\sredlabel{sredBetaEtaRed}\\
	&\sRed{s_A\termEquals{\PhiAppl{A}} t_{A'}}&:=&\sRed{s_A}\termEquals{\PhiAppl{A}}\sRed{t_{A}}\sredlabel{sredEq}\\
	&\sRed{F_\bool\impl G_\bool}&:=&\sRed{F_\bool}\impl \sRed{G_\bool}\QQNegSp\QQNegSp\QQNegSp\sredlabel{sredImpl}\\
	&\sRed{\lambdaFun{x}{A}s_B}&:=&\lambdaFun{x}{A}\sRed{s_B}\sredlabel{sredLam}\\
	&\sRed{\left(\sRed{f_{\piType{x}{A}B}}_{\piType{x}{A}B}\ \sRed{t_{A'}}_{A'}\right)_{B'}}&:=&\defaultTerm{\PhiAppl{B}} \qquad\qquad\negSp\text{if $A\neq A'$ or $B\neq B'$} \sredlabel{sredAppl}
	\end{align}}
\end{definition}

\begin{lemma}\label{lem:normalizingPrfTransform}
	Assume a well-typed \dphol{} theory $T$ and a conjecture $\Gamma\dedT \phi$ with $\Gamma$ well-formed and $\phi$ well-typed. Assume a valid \hol{} derivation of $\PhiAppl{\Gamma}\dedPT \PhiAppl{\phi}$. 
	Then, we can index the terms in the derivation s.t.
	any steps $S$ in the derivation can be replaced by a macro-step for the normalizing statement transformation replacing step $S$ s.t. after replacing all steps by their macro-steps:
	\begin{itemize}
		\item the resulting derivation is valid,
		\item all terms occuring in the derivation are almost proper.
	\end{itemize}
\end{lemma}
\begin{proof}[Proof of Lemma~\ref{lem:normalizingPrfTransform}]
	We will show this by induction on the inference rules.
	
	Firstly, we observe that there are no dependent types in \hol{} and the context and axioms contain no spurious subterms. Hence, well-formedness (of theories, contexts, types) and type-equality judgements are unaffected by the transformation. 
	So there is nothing to prove for the well-formedness and type-equality rules.
	
	\paragraph{Regarding the indices for the terms:}  
	We start indexing the term at the end of the derivation and go up line by line. 
	Whenever we need to pick an index for a term we already encountered (in a later step) in the derivation we will pick the same index. 
	Whenever we need to pick an index for a constant or variable that is the translation of a variable in the context of the \dphol{} conjecture, we pick as the index the type of the (unique) preimage of the constant or variable.
	Whenever we need to pick indices for terms on both sides of an equality, we pick the same index.
	Whenever we need to pick an index for a non-atomic term we pick indices for the atomic subterms in the term and then choose a type of the quasi-preimage (given the type by termwise-injectivity (see Lemma~\ref{lem:termwise-inj}) the quasi-preimage is unique) for which the quasi-preimages of the subterm have the types they are indexed with. 
	If we can choose indices in a way such that there exists a well-typed quasi-preimage of that type, we do. In particular, this means for an equality between $\lambda$-functions that we will pick the same types for both $\lambda$-functions if possible.
	We choose the indices for the normalizing statement transformation of a term indexed by $A$ to also be $A$, unless we have already picked a different index.
	If we have to pick an index for a variable that is part of a $\lambda$-function, we pick the index consistently with the domain of the index of the entire $\lambda$-function if an index is already chosen for it. In that case, we also pick the codomain of the index of the entire $\lambda$-function as the index of the body of the $\lambda$-function. 
	If no index is chosen for the $\lambda$-function, but it is applied to an argument, we pick the index of that argument for the variable bound by the $\lambda$ (this choice ensures that our transformation changes as little as necessary and is also consistent with the indexing of variables in proper terms). 
	Otherwise, we pick the index arbitrarily (but satisfying the already stated requirements). 
	
	It remains to consider the typing and validity rules and to construct macro steps for the steps in the derivation using them for the normalizing statement transformation. 
	
	Since terms indexed by a type $A$ have type $\PhiAppl{A}$ it is easy to see from Definition~\ref{defn:normalizingPrfTransform} that the normalizing statement transformation replaces terms of type $\PhiAppl{A}$ by terms of type $\PhiAppl{A}$.
	
	\paragraph{\ruleRef{const}:}
	Since constants are proper terms, there is nothing to prove.
	
	\paragraph{\ruleRef{var}:}
	Since context variables are proper terms, there is nothing to prove.
	
	\paragraph{\ruleRef{eqType}:}
	\begin{align}
	\NDLinePTD{\sRed{s}_A:\PhiAppl{A}}{\byAss}\label{eqTypeSR1}\\
	\NDLinePTD{\sRed{t}_A:\PhiAppl{A}}{\byAss}\label{eqTypeSR2}\\
	\NDLinePTD{\sRed{s}_A\termEquals{\PhiAppl{A}}\sRed{t}_A:\bool}{\ruleRef{eqType},(\ref{eqTypeSR1}),(\ref{eqTypeSR2})}\label{eqTypeSR3}\\
	\NDLinePTD{\sRed{s_A\termEquals{\PhiAppl{A}}t_A}_\bool:\bool}{\ref{sredEq},(\ref{eqTypeSR3})}\nonumber
	\end{align}
	
	\paragraph{\ruleRef{lambda}:}
	\begin{align}
	\NDLinePT{\concatCtx{\Delta}{x_A:\PhiAppl{A}}}{\sRed{t_B}_B:\PhiAppl{B}}{\byAss}\label{lambdaSR1}\\
	\NDLinePTD{\left(\lambdaFun{x_A}{\PhiAppl{A}}\sRed{t_B}_B\right):\PhiAppl{A}\to \PhiAppl{B}}{\ruleRef{lambda},(\ref{lambdaSR1})}\label{lambdaSR2}
	\intertext{If $\sRed{t}_B$ isn't an unnormalizably spurious function application $\sRed{f_{\piType{y}{A'}B}}\ x_A$ for which $x$ doesn't appear in $f$:}
	\NDLinePTD{\sRed{\lambdaFun{x_A}{\PhiAppl{A}}\sRed{t_B}_B}:\PhiAppl{A}\to \PhiAppl{B}}{\ref{sredLam},(\ref{lambdaSR2})}\nonumber
	\intertext{Else by (\ref{sredBetaEtaRed}) we have  $\sRed{\lambdaFun{x_A}{\PhiAppl{A}}\sRed{t_B}_B}=\sRed{f_{\piType{y}{A}B}}$. By the remark about the type of $\sRed{\cdot}$ it follows that $\sRed{f_{\piType{y}{A}B}}$ has type $\PhiAppl{\piType{y}{A}B}=\PhiAppl{A}\to\PhiAppl{B}$. }
	\NDLinePTD{\sRed{f_{\piType{y}{A'}B}}:\PhiAppl{A}\to \PhiAppl{B}}{see above}\nonumber
	\end{align}
	
	\paragraph{\ruleRef{appl}:}
	\begin{align}
	\NDLinePTD{\sRed{f_{\piType{x}{A}B}}:\PhiAppl{A}\to \PhiAppl{B}}{\byAss}\label{applSR1}\\
	\NDLinePTD{\sRed{t_{A'}}:\PhiAppl{A}}{\byAss}\label{applSR2}
	\intertext{Unless $\sRed{f_{\piType{x}{A}B}}\ \sRed{t_A'}$ beta reducible or not satisfying $A \typeEquals A'$:}
	\NDLinePTD{\sRed{f_{\piType{x}{A}B}\ t_{A'}}:\PhiAppl{B}}{\ref{sredP},\ruleRef{lambda},(\ref{applSR1}),(\ref{applSR2})}\nonumber
	\intertext{If $\sRed{f_{\piType{x}{A}B}}\ \sRed{t_A'}$ not beta reducible but also not satisfying $A \typeEquals A'$, then by (\ref{sredApplP}) and (\ref{sredAppl}) we have $$\sRed{f_{\piType{x}{A}B}\ t_A'}=\sRed{\sRed{f_{\piType{x}{A}B}}_{\piType{x}{A}B}\ \sRed{\sRed{t_A'}_A'}}=\defaultTerm{\PhiAppl{B}}.$$ By the axiom schema asserting the existence of $\defaultTerm{\PhiAppl{B}}$ we have $\defaultTerm{\PhiAppl{B}}:\PhiAppl{B}$:}
	\NDLinePTD{\defaultTerm{\PhiAppl{B}}:\PhiAppl{B}}{assumption}\nonumber
	\intertext{Otherwise if $\sRed{f_{\piType{x}{A}B}}\ \sRed{t_A'}$ beta reducible with $f_{\piType{x}{A}B}=\lambdaFun{x_A}{\PhiAppl{A}}\sRed{s}_B$ and not satisfying $A \typeEquals A'$, it follows that $$\sRed{f_{\piType{x}{A}B}\ t_{A'}}=\sRed{\subst{\sRed{s_B}}{x_A}{\sRed{t_{A'}}}}=\subst{\sRed{s_B}}{x_A}{\sRed{t_{A'}}}.$$ Observe that $\concatCtx{\Delta}{x:\PhiAppl{A}}\dedPT \sRed{s}:\PhiAppl{B}$ must be derivable (otherwise $f$ not well-typed) and thus:}
	\NDLinePTD{\subst{\sRed{s}_B}{x_A}{\sRed{t}_{A'}}:\PhiAppl{B}}{\ruleRef{rewriteTyping},assumption,(\ref{applSR2})}\nonumber
	\end{align}
	Observe that by induction hypothesis the derivations of $\Delta\dedPT \sRed{f_{\piType{x}{A}B}}:\PhiAppl{A}\to\PhiAppl{B}$ (and thus of $\Delta\dedPT \sRed{s}_B:\PhiAppl{B}$) and of $\Delta\dedPT \sRed{t}_A:\PhiAppl{A}$ are almost proper with quasi-preimages for terms of indexed types. Consequently, the steps in the derivation obtained by plugging in the proof of rule \ruleRef{rewriteTyping} into this case will also be almost proper with quasi-preimages of indexed types (as the terms in it for corresponding steps have the same types and indices as in the derivation for $\Delta\dedPT \sRed{s}_B:\PhiAppl{B}$ and steps in the derivation of $\Delta\dedPT \sRed{t}_A:\PhiAppl{A}$ which occur in this derivation are unchanged).
	Therefore treating rule \ruleRef{rewriteTyping} like a primitive rule is harmless here (we can replace a step using the rule by the steps in the derivation of the rule) and we are done.
	
	\paragraph{\ruleRef{implType}:}
	\begin{align}
	\NDLinePTD{\sRed{F}_\bool:\bool}{\byAss}\label{implTypeSR1}\\
	\NDLinePTD{\sRed{G}_\bool:\bool}{\byAss}\label{implTypeSR2}\\
	\NDLinePTD{\sRed{F}_\bool\impl \sRed{G}_\bool:\bool}{\ruleRef{implType},(\ref{implTypeSR1}),(\ref{implTypeSR2})}\label{implTypeSR3}\\
	\NDLinePTD{\sRed{F_\bool\impl F_\bool}:\bool}{\ref{sredImpl},(\ref{implTypeSR3})}\nonumber
	\end{align}
	
	\paragraph{\ruleRef{axiom}:}
	Since translations of axioms to \hol{} are always proper terms and the additionally generated axioms are almost proper, there is nothing to prove here.
	
	\paragraph{\ruleRef{assume}:}
	If the axiom is a typing axiom generated by the translation, it follows that it is almost proper. Similarly, if it is an axiom for a base type it will Otherwise:
	\begin{align}
	\namedass{ass}{\sRed{F_\bool}}&\text{ in $\Delta$}&&\text{\byAss}\label{assumeSR1}\\
	\NDLinePTD{\sRed{F_\bool}}{\ruleRef{assume},(\ref{assumeSR1})}\nonumber
	\end{align}
	By assumption $\sRed{F}$ almost proper (with a quasi-preimage of type $\bool$), so the conclusion of the rule is almost proper and there is nothing ot prove here.
	
	\paragraph{\ruleRef{congLam}:}
	\begin{align}
	\NDLinePTD{A\typeEquals A'}{\byAss}\label{congLamSR1}\\
	\NDLinePT{\concatCtx{\Delta}{x_A:\PhiAppl{A}}}{\sRed{t_B\termEquals{\PhiAppl{B}}t'_B}_\bool}{\byAss}\label{congLamSR2}\\
	\NDLinePT{\concatCtx{\Delta}{x_A:\PhiAppl{A}}}{\sRed{t}_B\termEquals{\PhiAppl{B}}\sRed{t'}_B}{\ref{sredEq},(\ref{congLamSR2})}\label{congLamSR2a}\\
	\NDLinePTD{\lambdaFun{x_A}{\PhiAppl{A}}\sRed{t}_B\termEquals{\PhiAppl{A}\to \PhiAppl{B}}\nonumber\\&\lambdaFun{x_A}{\PhiAppl{A}}\sRed{t'}_B}{\ruleRef{congLam},(\ref{congLamSR1}),(\ref{congLamSR2a})}\label{congLamSR3}
	\end{align}	
	By assumption $\sRed{t}_B\termEquals{\PhiAppl{B}}\sRed{t'}_B$ almost proper with quasi-preimage consistent with type indices and $A\typeEquals A'$, thus also $\lambdaFun{x_A}{\PhiAppl{A}}\sRed{t}_B\termEquals{\PhiAppl{A}\to \PhiAppl{B}}\lambdaFun{x_A}{\PhiAppl{A}}\sRed{t'}_B$ almost proper with quasi-preimage consistent with type indices. 
	\begin{align}
	\NDLinePTD{\sRed{\lambdaFun{x_A}{\PhiAppl{A}}\sRed{t}_B\termEquals{\PhiAppl{A}\to \PhiAppl{B}}\lambdaFun{x_A}{\PhiAppl{A}}\sRed{t'}_B}\QQNegSp}{\qquad\ref{sredLam},\ref{sredEq},(\ref{congLamSR3})}\nonumber
	\end{align} 
	
	\paragraph{\ruleRef{congAppl}:}
	\begin{align}
	\NDLinePTD{\sRed{t_{A}\termEquals{\PhiAppl{A}}t'_{A}}}{\byAss}\label{congApplSR1}\\
	\NDLinePTD{\sRed{t}_{A}\termEquals{\PhiAppl{A}}\sRed{t'}_{A}}{\ref{sredEq},(\ref{congApplSR1})}\label{congApplSR1a}\\
	\NDLinePTD{\sRed{f_{\piType{x}{A'}B}\termEquals{\PhiAppl{A}\to \PhiAppl{B}}f'_{\piType{x}{A'}B}}}{\byAss}\label{congApplSR2}\\
	\NDLinePTD{\sRed{f}_{\piType{x}{A'}B}\termEquals{\PhiAppl{A}\to \PhiAppl{B}}\sRed{f'}_{\piType{x}{A'}B}}{\ref{sredEq},(\ref{congApplSR2})}\label{congApplSR2a}
	\intertext{
		Assume that $A\not\typeEquals A'$. Our choice of type indices then implies $\sRed{f}$ and $\sRed{f'}$ are not $\lambda$-functions. 
		Consequently, the applications $\sRed{f}\ \sRed{s}$ and $\sRed{f'}\ \sRed{s'}$ are not beta or eta reducible. Thus, $\sRed{\sRed{f}_{\piType{x}{A}B}\ \sRed{t}_{A'}}=\defaultTerm{\PhiAppl{B}}$ \ and \breakintertext$\mathsf{sRed}\big(\sRed{f'}_{\piType{x}{A}B} \sRed{t'}_{A'}\big)=\defaultTerm{\PhiAppl{B}}$ \ and we yield:}
	\NDLinePTD{\sRed{\sRed{f}_{\piType{x}{A'}B}\ \sRed{t}_{A}}\termEquals{\PhiAppl{B}} \sRed{\sRed{f'}_{\piType{x}{A'}B}\ \sRed{t'}_{A}}\QQQNegSp\QQQNegSp}{\qqqquad\qquad\ruleRef{refl}}\nonumber
	\end{align}
	Otherwise the applications $\sRed{f}_{\piType{x}{A}B}\ \sRed{t}_{A'}$ and $\sRed{f'}_{\piType{x}{A}B}\ \sRed{t'}_{A'}$ are almost proper with quasi-preimages consistent with type indices. It follows: \[\sRed{\sRed{f}_{\piType{x}{A}B}\ \sRed{t}_{A'}}=\sRed{f}_{\piType{x}{A}B}\ \sRed{t}_{A'}\] and \[\sRed{\sRed{f'}_{\piType{x}{A}B}\ \sRed{t'}_{A'}}=\sRed{f'}_{\piType{x}{A}B}\ \sRed{t'}_{A'}\] and thus:
	\begin{align}
	\NDLinePTD{\sRed{\sRed{f}_{\piType{x}{A}B}\ \sRed{t}_{A'}}\termEquals{\PhiAppl{B}}\nonumber\\
		&\sRed{\sRed{f'}_{\piType{x}{A}B}\ \sRed{t'}_{A'}}}{\ruleRef{congAppl},(\ref{congApplSR1a}),(\ref{congApplSR2a})}\label{congApplSR3}\\
	\NDLinePTD{\sRed{\sRed{f}_{\piType{x}{A}B}\ \sRed{t}_{A'}\termEquals{\PhiAppl{B}}\sRed{f'}_{\piType{x}{A}B}\ \sRed{t'}_{A'}}\QQQQNegSp}{\QQQQuad\ref{sredEq},(\ref{congApplSR3})}\nonumber
	\end{align}
	
	\paragraph{\ruleRef{refl}:}
	\begin{align}
	\NDLinePTD{\sRed{t}_A:\PhiAppl{A}}{\byAss}\label{reflSR1}\\
	\NDLinePTD{\sRed{t}_A\termEquals{\PhiAppl{A}}\sRed{t}_A}{\ruleRef{refl},(\ref{reflSR1})}\label{reflSR2}\\
	\NDLinePTD{\sRed{\sRed{t}_A\termEquals{\PhiAppl{A}}\sRed{t}_A}}{\ref{sredEq},(\ref{reflSR2})}\nonumber
	\end{align}
	
	\paragraph{\ruleRef{sym}:}
	\begin{align}
	\NDLinePTD{\sRed{t_A\termEquals{\PhiAppl{A}}s_A}}{\byAss}\label{symSR1}\\
	\NDLinePTD{\sRed{t_A}_A\termEquals{\PhiAppl{A}}\sRed{s_A}_A}{\ref{sredEq},(\ref{symSR1})}\label{symSR2}\\
	\NDLinePTD{\sRed{s_A}_A\termEquals{\PhiAppl{A}}\sRed{t_A}_A}{\ruleRef{sym},(\ref{symSR2})}\label{symSR3}\\
	\NDLinePTD{\sRed{s_A\termEquals{\PhiAppl{A}}t_A}}{\ref{sredEq},(\ref{symSR3})}\nonumber
	\end{align}
	
	\paragraph{\ruleRef{beta}:}
	\begin{align}
	\NDLinePTD{\sRed{\left(\lambdaFun{x_A}{\PhiAppl{A}}s_B\right)\ t_A'}_{B'}:\PhiAppl{B}}{\byAss}\label{betaSR1}
	\intertext{
	By (\ref{sredLam}) and choice of type indeces, it follows that $\sRed{\left(\lambdaFun{x_A}{\PhiAppl{A}}s_B\right)}=\lambdaFun{x_A}{\PhiAppl{A}}\sRed{s_B}_B$.
	If $\left(\lambdaFun{x_A}{\PhiAppl{A}}\sRed{s_B}_B\right)\ t_{A'}$ is almost proper with quasi-preimage of type $B\typeEquals B'$, then we yield:}
	\NDLinePTD{\left(\lambdaFun{x_A}{\PhiAppl{A}}\sRed{s_B}\right)\ \sRed{t_A}:\PhiAppl{B}}{\ref{sredApplP},(\ref{betaSR1})}\label{betaSR2}\\
	\NDLinePTD{\left(\lambdaFun{x_A}{\PhiAppl{A}}\sRed{s_B}\right)\ \sRed{t_A}\termEquals{\PhiAppl{B}}\nonumber\\& \subst{\sRed{s_B}}{x_A}{\sRed{t_A}}}{\ruleRef{beta},(\ref{betaSR2})}\label{betaSR3}\\
	\NDLinePTD{\sRed{\left(\lambdaFun{x_A}{\PhiAppl{A}}\sRed{s_B}\right)\ \sRed{t_A}\termEquals{\PhiAppl{B}} \subst{\sRed{s_B}}{x_A}{\sRed{t_A}}}\QQQNegSp}{\QQQuad\ref{sredEq},\ref{sredP},(\ref{betaSR3})}\nonumber
	\end{align}
	Otherwise, $$\sRed{\left(\lambdaFun{x_A}{\PhiAppl{A}}s_B\right)\ t_A'}=\sRed{\subst{\sRed{s_B}}{x_A}{\sRed{t_A'}}}=\subst{\sRed{s_B}}{x_A}{\sRed{t_A'}}$$ and we yield:
	\begin{align}
		\NDLinePTD{\subst{\sRed{s_B}}{x_A}{\sRed{t_A}}\termEquals{\PhiAppl{B}}\subst{\sRed{s_B}}{x_A}{\sRed{t_A}}}{\ruleRef{refl}}\label{betaSR5}
		\intertext{By induction hypothesis, we yield that $\subst{\sRed{s_B}}{x_A}{\sRed{t_A}}$ is has a quasi-preimage of type $B'$. By choice of indeces it then follows that $B\typeEquals B'$. Therefore, we can conclude:}
		\NDLinePTD{\sRed{\left(\lambdaFun{x_A}{\PhiAppl{A}}\sRed{s_B}\right)\ \sRed{t_A}\termEquals{\PhiAppl{B}} \subst{\sRed{s_B}}{x_A}{\sRed{t_A}}}\QQQNegSp}{\QQQuad\ref{sredEq},(\ref{betaSR5})}\nonumber
	\end{align}
	
	\paragraph{\ruleRef{eta}:}
	{\small\begin{align}
	\NDLinePTD{\sRed{t_{\piType{x}{A}B}}:\PhiAppl{A}\to \PhiAppl{B}}{\byAss}\label{etaSR1}\\
	&\text{$x$ not in $\Delta$}&&\text{\byAss}\label{etaSR2}
	\intertext{$\lambdaFun{x_A}{\PhiAppl{A}}\sRed{t_{\piType{x}{A}B}}\ x_A$ is by choice of indeces (determined by the equality concluded in this step of the proof) almost proper with quasi-preimage of type $\piType{x}{A}B$. It follows:}
	\NDLinePTD{\sRed{t_{\piType{x}{A}B}}\termEquals{\PhiAppl{A}\to \PhiAppl{B}}\lambdaFun{x_A}{\PhiAppl{A}}\sRed{t_{\piType{x}{A}B}}\ x_A}{\ruleRef{eta},(\ref{etaSR1}),(\ref{etaSR2})}\label{etaSR3}\\
	\NDLinePTD{\sRed{\sRed{t_{\piType{x}{A}B}}\termEquals{\PhiAppl{A}\to \PhiAppl{B}}\lambdaFun{x_A}{\PhiAppl{A}}\sRed{t_{\piType{x}{A}B}}\ x_A}}{\ref{sredApplP},\ref{sredLam},\ref{sredEq},(\ref{etaSR3})}\nonumber
	\end{align}}
	
	\paragraph{\ruleRef{congDed}:}
	\begin{align}
	\NDLinePTD{\sRed{F_\bool\termEqB F'_\bool}}{\byAss}\label{congDedSR1}\\
	\NDLinePTD{\sRed{F'_\bool}}{\byAss}\label{congDedSR3}\\
	\NDLinePTD{\sRed{F_\bool}\termEqB\sRed{F'_\bool}}{\ref{sredEq},(\ref{congDedSR1})}\label{congDedSR2}\\
	\NDLinePTD{\sRed{F_\bool}}{\ruleRef{congDed},(\ref{congDedSR2}),(\ref{congDedSR3})}\nonumber
	\end{align}
	
	\paragraph{\ruleRef{implI}:}
	\begin{align}
	\NDLinePTD{\sRed{F_\bool}:\bool}{\byAss}\label{implISR1}\\
	\NDLinePT{\concatCtx{\Delta}{\namedass{ass_F}{\sRed{F_\bool}}}}{\sRed{G_\bool}}{\byAss}\label{implISR2}\\
	\NDLinePTD{\sRed{F_\bool}\impl \sRed{G_\bool}}{\ruleRef{implI},(\ref{implISR1}),(\ref{implISR2})}\label{implISR3}\\
	\NDLinePTD{\sRed{F_\bool\impl G_\bool}}{\ref{sredImpl},(\ref{implISR3})}\nonumber
	\end{align}
	
	\paragraph{\ruleRef{implE}:}
	\begin{align}
	\NDLinePTD{\sRed{F_\bool\impl G_\bool}}{\byAss}\label{implESR1}\\
	\NDLinePTD{\sRed{F_\bool}}{\byAss}\label{implESR3}\\
	\NDLinePTD{\sRed{F_\bool}\impl\sRed{G_\bool}}{\ref{sredImpl},(\ref{implESR1})}\label{implESR2}\\
	\NDLinePTD{\sRed{G_\bool}}{\ruleRef{implE},(\ref{implESR2}),(\ref{implESR3})}\nonumber
	\end{align}
	
	\paragraph{\ruleRef{boolExt}:}
	\begin{align}
	\NDLinePTD{\sRed{p_{\bool\to\bool}\ \T_\bool}}{\byAss}\label{boolExtSR1}\\
	\NDLinePTD{\sRed{p_{\bool\to\bool}\ \F_\bool}}{\byAss}\label{boolExtSR2}\\
	\NDLinePTD{\sRed{p_{\bool\to\bool}}\ \T}{(\ref{sredApplP}),(\ref{sredP}),(\ref{boolExtSR1})}\label{boolExtSR3}\\
	\NDLinePTD{\sRed{p_{\bool\to\bool}}\ \F}{(\ref{sredApplP}),(\ref{sredP}),(\ref{boolExtSR2})}\label{boolExtSR4}\\
	\NDLinePTD{\univQuant{x}{\bool}\sRed{p_{\bool\to\bool}}\ x}{\ruleRef{boolExt},(\ref{boolExtSR3}),(\ref{boolExtSR4})}\label{boolExtSR5}\\
	\NDLinePTD{\sRed{\univQuant{x}{\bool}\sRed{p_{\bool\to\bool}}\ x}}{\ref{sredEq},(\ref{sredApplP}),(\ref{sredP}),(\ref{boolExtSR5})}\nonumber
	\end{align}
\end{proof}

\subsection{Lifting admissible \hol{} derivations of validity statements to \dphol{}}\label{sec:sound}
We finally have all required results to prove the soundness of the translation from \dphol{} to \hol{}.
\begin{proof}[Proof of Theorem~\ref{thm:soundPaper}]
	As shown in Lemma~\ref{lem:normalizingPrfTransform}, we may assume that the proof of $\PhiAppl{\Gamma}\dedPT \PhiAppl{F}$ is admissible, so it only contains almost-proper terms. 
	Consequently, whenever an equality $s\termEquals{\PhiAppl{A}}t$ is derivable in \hol{} and $s', t'$ are the quasi-preimages of $s, t$ respectively, it follows that it's quasi-preimage $s'\termEquals{A}t$ is well-typed in \dphol{} and thus $s':A$ and $t':A$. 
	Without loss of generality (adding extra assumptions throughout the proof) we may assume that the context of the (final) conclusion is the translation of a \dphol{} context.
	By Lemma~\ref{lem:termwise-inj} the translation is term-wise injective. 
	
	Therefore, the translated conjecture is a proper validity statement with unique (quasi)-preimage in \dphol{}. 
	If we can lift a derivation of the translated conjecture to a valid \dphol{} derivation of its quasi-preimage, the resulting derivation is a valid derivation of the original conjecture. 
	This means, that it suffices to prove that we can lift admissible derivations of a proper validity statement $S$ in \hol{} to a derivation of a quasi-preimage of $S$.
	
	We prove this claim by induction on the validity rules of \hol{} as follows:
	
	Given a validity rule $R$ with assumptions $A_1,\ldots,A_n$, validity assumptions (assumptions that are validity statements) $V_1,\ldots,V_m$, non-judgement assumptions (meaning assumptions that something occurs in a context or theory) $N_1,\ldots,N_p$ and conclusion $C$ we will show the following:
	\begin{claim}\label{soundness-inductive-claim}
		Assuming that the $A_i$ and the $N_j$ hold.
		\begin{enumerate}
			\item Assume that the conclusion $C$ is proper with quasi-preimage $C^{-1}$. Then the contexts $C_i$ of the $V_i$ are proper 
			and the quasi-preimages of the $V_i$ are well-formed. 
			\item Assume that whenever an $V_i$ is proper its quasi-preimage (where we choose the same preimages for identical terms and types with several possible preimages) holds in \dphol{} and that the conclusion $C$ is proper with quasi-preimage $C^{-1}$. 
			Then, $C^{-1}$ holds in \dphol{}.
		\end{enumerate}
	\end{claim}
	
	Consider the first part of this claim, namely that if $C$ is proper then the $V_i$ are proper. 
	Since all formulae appearing in the derivation are almost proper, this implies that the $V_i$ themselves are proper and by construction (choice of quasi-preimage) the contexts of their quasi-preimages fit together with the context of $C^{-1}$.
	
	The translation clearly implies that if an $N_j$ holds in \hol{}, the corresponding non-judgement assumption ${N_j}^{-1}$ holds in \dphol{} (e.g. if $\PhiAppl{F}$ is an axiom in $\PhiAppl{T}$, then $F$ must be an axiom in $T$). 
	
	Since the validity judgement being derived is proper, it follows from this first part of the claim that the validity assumptions of all validity rules in the derivation are proper. 
	
	By induction on the validity rules, if given an arbitrary validity rule $R$ whose assumptions hold and whose validity assumptions all satisfy a property $P$ we can show that $P$ holds on the conclusion of $R$, then all derivable validity judgments have property $P$.
	Since all the validity assumptions and conclusions of validity rules in the derivation are proper, the property of having a derivable quasi-preimage is such a property.
	By this induction principle, it suffices to prove the claim for the validity rules in \hol. 
	
	We will therefore consider the validity rules one by one. For each rule we first prove the first part of the claim. 
	Sometimes we also need that the quasi-preimages of some non-validity (typically typing) assumptions hold, so we will prove that this also follows from the conclusion being proper.
	Then the assumption of the second part, combined with the first part implies that the quasi-preimages of the $V_i$ hold in \dphol{} and it is easy to prove that also $C^{-1}$ holds in \dphol{}.
	
	Throughout this proof we will use the notation $\quasiImage{t}$ to denote that $t$ is some quasi-preimage of $\quasiImage{t}$. 
	Since the translation is surjective on type-level we will only need this notation on term-level.
	
	Validity can be shown using the rules \ruleRef{congLam}, \ruleRef{eta}, \ruleRef{congAppl}, \ruleRef{congDed}, \ruleRef{beta}, \ruleRef{refl}, \ruleRef{sym}, \ruleRef{assume}, \ruleRef{axiom}, \ruleRef{implI}, \ruleRef{implE} and \ruleRef{boolExt}.
	\subparagraph{\ruleRef{congLam}:}

	Since the conclusion is proper, it follows that the preimage
	\[
	\Gamma\dedT \lambdaFun{x}{A}t\termEquals{\piType{x}{A}B}\lambdaFun{x}{A}t'
	\] of the normalization \[\PhiAppl{\Gamma}\dedPT 
	\univQuant{x}{\PhiAppl{A}}\univQuant{y}{\PhiAppl{A}}\termEqT{A}{x}{y}\impl \termEqT{B}{\quasiImage{t}\ x}{\quasiImage{t'}\ y}\] of the conclusion is well-formed. 
	By rule \ruleRef{eqTyping} and rule \ruleRef{sym} we obtain $\Gamma\dedT\lambdaFun{x}{A}t:\piType{x}{A}B$ and $\Gamma\dedT\lambdaFun{x}{A}t':\piType{x}{A}B$ in \dphol. 
	\begin{align}
		\NDLineTG{\lambdaFun{x}{A}t:\piType{x}{A}B}{see above}\label{congApplP1.1}\\
		\NDLineTG{\lambdaFun{x}{A}t':\piType{x}{A}B}{see above}\label{congApplP1.2}\\
		\NDLineT{\concatCtx{\Gamma}{y:A}}{\left(\lambdaFun{x}{A}t\right)\ y:B}{\ruleRef{appl},\ruleRef{varDed},(\ref{congApplP1.1}),\ruleRef{assume}}\label{congApplP1.3}\\
		\NDLineT{\concatCtx{\Gamma}{y:A}}{\left(\lambdaFun{x}{A}t'\right)\ y:B}{\ruleRef{appl},\ruleRef{varDed},(\ref{congApplP1.2}),\ruleRef{assume}}\label{congApplP1.4}\\
		\NDLineT{\concatCtx{\Gamma}{y:A}}{\left(\lambdaFun{x}{A}t\right)\ y\termEquals{B}\subst{t}{x}{y}}{\ruleRef{beta},(\ref{congApplP1.3})}\label{congApplP1.5}\\
		\NDLineT{\concatCtx{\Gamma}{y:A}}{\left(\lambdaFun{x}{A}t'\right)\ y\termEquals{B}\subst{t'}{x}{y}}{\ruleRef{beta},(\ref{congApplP1.4})}\label{congApplP1.6}\\		
		\NDLineT{\concatCtx{\Gamma}{x:A}}{t:B}{$\alpha$-renaming,\ruleRef{congColon},(\ref{congApplP1.5})}\label{congApplP1.7}\\
		\NDLineT{\concatCtx{\Gamma}{x:A}}{t':B}{$\alpha$-renaming,\ruleRef{congColon},(\ref{congApplP1.6})}\label{congApplP1.8}\\
		\NDLineT{\concatCtx{\Gamma}{x:A}}{t\termEquals{B}t}{\ruleRef{eqType},(\ref{congApplP1.7}),(\ref{congApplP1.8})}\nonumber
	\end{align}
	Clearly, $\concatCtx{\Gamma}{x:A}\dedT t\termEquals{B}t'$ is a quasi-preimage of the validity assumption, so this proves the first part of the claim. 
	
	Regarding the second part:
	\begin{align}
	\NDLineT{\concatCtx{\Gamma}{x:A}}{t\termEquals{B}t'}{\byAss}\label{congLamC3}\\
	\NDLineTG{A\typeEquals A}{\ruleRef{tpEqRefl},\ruleRef{typingTp},(\ref{congLamC3})}\label{congLamC4}\\
	\NDLineTG{\lambdaFun{x}{A}t\termEquals{\piType{x}{A}B}\lambdaFun{x}{A}t'}{\ruleRef{congLam'},(\ref{congLamC4})}\label{congLamC5}
	\end{align}
	
	\subparagraph{\ruleRef{eta}:}
	Since the rule has no validity assumption, the first part of the claim holds. 
	
	For the second part, we still need the quasi-preimage of the assumption to hold, so we will show that it follows from the conclusion being proper.
	
	Since the conclusion is proper, it follows that the preimage 
	\[
	\Gamma\dedT t\termEquals{\piType{x}{A}B}\lambdaFun{x}{A}t\ x
	\] of the normalization \[\PhiAppl{\Gamma}\dedPT 
	\univQuant{x}{\PhiAppl{A}}\univQuant{y}{\PhiAppl{A}}\termEqT{A}{x}{y}\impl \termEqT{B}{\quasiImage{t}\ x}{\left(\lambdaFun{x}{\PhiAppl{A}}\quasiImage{t}\ x\right)\ y}\] of the conclusion is well-formed. 
	By rule \ruleRef{eqTyping} and rule \ruleRef{sym} we obtain $\Gamma\dedT t:\piType{x}{A}B$ and $\Gamma\dedT\lambdaFun{x}{A}t\ x:\piType{x}{A}B$ in \dphol.
	Clearly, $\Gamma\dedT t:\piType{x}{A}B$ is a quasi-preimage of the validity assumption, so this proves the quasi-preimage of the assumption of the rule. 
	
	Regarding the second part:
	\begin{align}
		\NDLineTG{t:\piType{x}{A}B}{see above}\label{etaC1}\\
		\NDLineTG{t\termEquals{\piType{x}{A}B}\lambdaFun{x}{A}t\ x}{\ruleRef{etaPi},(\ref{etaC1})}\nonumber
	\end{align}
	\subparagraph{\ruleRef{congAppl}:}
	Since the conclusion is proper, it follows that the preimage 
	\[
	\Gamma\dedT f\ t\termEquals{B}f'\ t'
	\] of the normalization \[\termEqT{B}{\quasiImage{f}\ \quasiImage{t}}{\quasiImage{f'}\ \quasiImage{t'}}\] of the conclusion is well-formed. 
	By rule \ruleRef{eqTyping} and rule \ruleRef{sym} we obtain $\Gamma\dedT f\ t:B$ and $\Gamma\dedT f'\ t':B$ in \dphol.
	Obviously, $\Gamma\dedT t\termEquals{A}t'$ and $\Gamma\dedT f\termEquals{\piType{x}{A}B}f'$ are quasi-preimages of the validity assumptions. 
	
	Since the validity assumptions use the same context as the conclusion, it follows that they are both proper with uniquely determined context. As observed in the beginning of the proof if a proper assumption of a rule is an equality over a type $\PhiAppl{A}$, the induction hypothesis implies that the quasi-preimage of that assumption in which the equality is over type $A$ must be well-formed.
	Hence both $\Gamma\dedT t\termEquals{A}t'$ and $\Gamma\dedT f\termEquals{\piType{x}{A}B}f'$ are well-formed in \dphol, so we have proven the first part of the claim.
	
	Regarding the second part of the claim:
	\begin{align}
		\NDLineTG{t\termEquals{A}t'}{\byAss}
		\label{congApplC7}\\
		\NDLineTG{f\termEquals{\piType{x}{A}B}f'}{\byAss}
		\label{congApplC8}\\
		\NDLineTG{f\ t\termEquals{B}f'\ t'}{\ruleRef{congAppl},(\ref{congApplC7}),(\ref{congApplC8})}
	\end{align}
	This is what we had to show.
	
	\subparagraph{\ruleRef{congDed}:}
	Since the conclusion is proper, it follows that the preimage 
	\[
	\Gamma\dedT F
	\] of the normalization \[\PhiAppl{\Gamma}\dedPT \quasiImage{F}\] of the conclusion is well-formed. 
	Thus we have $\Gamma\dedT F:\bool$. 
	Since the validity assumptions use the same context as the conclusion, it follows that they are both proper with uniquely determined. As observed in the beginning of the proof if a proper assumption of a rule is an equality over a type $\PhiAppl{A}$ (here $A=\PhiAppl{A}=\bool$), the induction hypothesis implies that the quasi-preimage of that assumption in which the equality is over type $\bool$ must be well-formed.
	Clearly, $\Gamma\dedT F'\termEqB F$ and $\Gamma\dedT F$ are the quasi-preimages of the two validity assumptions. 
	Since the former is a validity statement about the quasi-preimage of an equality, it follows that $\Gamma\dedT F'\termEqB F$ is well-formed.
	We have already seen that $\Gamma\dedT F$ is well-typed.
	This shows the first part of the claim.

	Regarding the second part:
	\begin{align}
		\NDLineTG{F'\termEqB F}{\byAss}\label{congDedC1}\\
		\NDLineTG{F'}{\byAss}\label{congDedC2}\\
		\NDLineTG{F}{\ruleRef{congDed},(\ref{congDedC1}),(\ref{congDedC2})}\nonumber
	\end{align}
	\subparagraph{\ruleRef{beta}:}
	Since the rule has no validity assumptions, the first part of the claim trivially holds.
	
	Since the conclusion is proper, it follows that the preimage 
	\[
	\Gamma\dedT \big(\lambdaFun{x}{A}s\big)\ t\termEquals{\piType{x}{A}B}\subst{s}{x}{t}
	\] of the normalization \[\PhiAppl{\Gamma}\dedPT 
	\termEqT{B}{\big(\lambdaFun{x}{\PhiAppl{A}}\quasiImage{s}\big)\ \quasiImage{t}}{\subst{\quasiImage{s}}{x}{\quasiImage{t}}}\] of the conclusion is well-formed. 
	By rule \ruleRef{eqTyping}, we obtain $\Gamma\dedT \big(\lambdaFun{x}{A}s\big)\ t:B$ in \dphol.
	Clearly, $\Gamma\dedT \big(\lambdaFun{x}{A}s\big)\ t:B$ is a quasi-preimage of the assumption of the rule, so we have proven that the quasi-preimage of the assumption of the rule holds in \dphol{}.
	
	Regarding the second part:
	\begin{align}
	\NDLineTG{\Gamma\dedT \big(\lambdaFun{x}{A}s\big)\ t:B}{see above}\label{betaC1}\\
	\NDLineTG{\Gamma\dedT \big(\lambdaFun{x}{A}s\big)\ t\termEquals{\piType{x}{A}B}\subst{s}{x}{t}}{\ruleRef{beta},(\ref{betaC1})}\nonumber
	\end{align}
	
	\subparagraph{\ruleRef{refl}:}
	Once again the rule has no validity assumptions, so the first part of the claim trivially holds.
	
	Since the conclusion is proper, it follows that the preimage 
	\[
	\Gamma\dedT t\termEquals{A} t'
	\] of the normalization \[\PhiAppl{\Gamma}\dedPT \termEqT{A}{\quasiImage{t}}{\quasiImage{t}}\] of the conclusion is well-formed. 
	By Lemma~\ref{lem:termwise-inj} it follows that $t$ and $t'$ are identical so the quasi-preimage is $\Gamma\dedT \termEquals{A} t$. 
	By rule \ruleRef{eqTyping}, we obtain $\Gamma\dedT t:A$ in \dphol{}, the quasi-preimage of the assumption of the rule.
	
	Regarding the second part of the claim:
	\begin{align}
		\NDLineTG{t:A}{see above}\label{reflC1}\\
		\NDLineTG{t\termEquals{A}t}{\ruleRef{refl},(\ref{reflC1})}\nonumber
	\end{align}
	
	\subparagraph{\ruleRef{sym}:}
	Since the conclusion is proper, it follows that the preimage 
	\[
	\Gamma\dedT t\termEquals{A}s
	\] of the normalization \[
	\PhiAppl{\Gamma}\dedPT \termEqT{A}{\quasiImage{t}}{\quasiImage{s}}\] of the conclusion is well-formed. 
	By the rules \ruleRef{eqTyping} and \ruleRef{sym} both $\Gamma\dedT t:A$ and $\Gamma\dedT s:A$ follow. 
	By rule \ruleRef{eqType} it follows that $\Gamma\dedT s\termEquals{A}t$ is well-formed. 
	Clearly, $\Gamma\dedT s\termEquals{A}t$ is the quasi-preimages of the validity assumption, so we have proven the first part of the claim.
	
	Regarding the second part:
	\begin{align}
		\NDLineTG{t\termEquals{A}s}{\byAss}\label{symC1}\\
		\NDLineTG{s\termEquals{A}t}{\ruleRef{sym},(\ref{symC1})}
	\end{align}
	
	\subparagraph{\ruleRef{assume}:}
	Once again, there are no validity assumption, so the first part of the claim is trivial.
	
	Since the conclusion is proper, it follows that the preimage 
	\[
	\Gamma\dedT F
	\] of the normalization \[
	\PhiAppl{\Gamma}\dedPT \quasiImage{F}\] of the conclusion is well-formed and thus $\Gamma\dedT F:\bool$.
	\begin{align}
		\NDLineTG{F:\bool}{see above}\label{assumeP1.1}\\
		\NDLineTG{\Type{\bool}}{\ruleRef{typingTp},(\ref{assumeP1.1})}\label{assumeP1.2}\\
		\NDLineT{}{\Ctx{\Gamma}}{\ruleRef{tpCtx},(\ref{assumeP1.2})}\label{assumeP1.3}		
	\end{align}
	
	The context assumption may be the translation of a context assumption in \dphol{} or a typing assumption added by the translation. 
	In the latter case, $F$ is of the form $F=\PredPhi{A}{x}$ for $\ctxIn{x:A}{\Gamma}$. 
	In that case, the second part of the claim $\Gamma\dedT F$ can be concluded as follows:
	\begin{align}
		\NDLineTG{x:A}{\ruleRef{var''},\ruleRef{subtI}}\label{assumeCT1}\\
		\NDLineTG{x\termEquals{A}t}{\ruleRef{refl},(\ref{assumeCT1})}\label{assumeCT2}\\
		\NDLineTG{F}{\byAss{} $F=\PredPhi{A}{x}$,(\ref{assumeCT2})}\nonumber
	\end{align}
	
	Otherwise:
	\begin{align}
		&\ctxIn{\namedass{ass}{F}}{\Gamma}&&\text{\byAss}\label{assumeC1}\\
		\NDLineTG{F}{\ruleRef{assume},(\ref{assumeC1}),(\ref{assumeP1.3})}\nonumber
	\end{align}
	\subparagraph{\ruleRef{axiom}:}
	Once again, there are no validity assumption, so the first part of the claim is trivial.
	
	Since the conclusion is proper, it follows that the preimage 
	\[
	\Gamma\dedT F
	\] of the normalization \[
	\PhiAppl{\Gamma}\dedPT \quasiImage{F}\] of the conclusion is well-formed and thus $\Gamma\dedT F:\bool$.
	\begin{align}
	\NDLineTG{F:\bool}{see above}\label{axiomP1.1}\\
	\NDLineTG{\Type{\bool}}{\ruleRef{typingTp},(\ref{axiomP1.1})}\label{axiomP1.2}\\
	\NDLineT{}{\Ctx{\Gamma}}{\ruleRef{tpCtx},(\ref{axiomP1.2})}\label{axiomP1.3}		
	\end{align}
	The axiom may be the translation of an axiom in $T$, a typing axiom added by the translation or an axiom added for some base type $A$.
	In the first case, the second part of the claim follows by:
	\begin{align}
	&\thyIn{\namedax{ax}{F}}{T}&&\text{\byAss}\label{axiomC1}\\
	\NDLineTG{F}{\ruleRef{axiom},(\ref{axiomC1}),(\ref{axiomP1.3})}\nonumber
	\end{align}
	If the axiom is a typing axiom then its preimage states that some constant $c$ of type $A$ satisfies $c\termEquals{A}t$ which follows by rule \ruleRef{refl}.
	
	If the axiom is the PER axiom generated for some $A$ type declared in $T$, then it's quasi-preimage states that equality on $A$ implies itself which is obviously true. 

	\subparagraph{\ruleRef{implI}:}
	Since the conclusion is proper, it follows that the preimage 
	\[
	\Gamma\dedT F\impl G
	\] of the normalization \[
	\PhiAppl{\Gamma}\dedPT \quasiImage{F}\impl \quasiImage{G}\] of the conclusion is well-formed and thus $\Gamma\dedT F\impl G:\bool$.
	\begin{align}
		\NDLineTG{F\impl G:\bool}{see above}\label{implIP1.1}\\
		\NDLineTG{F:\bool}{\ruleRef{implTypingL},(\ref{implIP1.1})}\label{implIP1.2}\\
		\NDLineTG{G:\bool}{\ruleRef{implTypingR},(\ref{implIP1.1})}\label{implIP1.3}\\
		\NDLineT{\concatCtx{\Gamma}{\namedass{ass}{F}}}{G:\bool}{\ruleRef{assDed},(\ref{implIP1.3})}\nonumber
	\end{align}
	Obviously $\concatCtx{\Gamma}{\namedass{ass}{F}}\dedT G$ is a quasi-preimage of the validity assumption of the rule, so the first part of the claim is proven.
	
	Regarding the second part:
	\begin{align}
		\NDLineT{\concatCtx{\Gamma}{\namedass{ass}{F}}}{G}{\byAss}\label{implIC1}\\
		\NDLineTG{F\impl G}{\ruleRef{implI},(\ref{implIP1.2}),(\ref{implIC1})}\nonumber
	\end{align}
	
	\subparagraph{\ruleRef{implE}:}
	Since the conclusion is proper, it follows that the preimage 
	\[
	\Gamma\dedT G
	\] of the normalization \[
	\PhiAppl{\Gamma}\dedPT \quasiImage{G}\] of the conclusion is well-formed and thus $\Gamma\dedT G:\bool$.
	
	Since the validity assumptions use the same context as the conclusion, it follows that they are both proper and uniquely determined.
	
	Since the formula $\quasiImage{F}$ (where $\PhiAppl{\Gamma}\dedPT \quasiImage{F}$ is the second validity assumption) must be almost proper, it follows that its preimage $F$ is well-typed i.e. $\Gamma\dedT F:\bool$.
	\begin{align}
		\NDLineTG{F:\bool}{$\quasiImage{F}$ almost proper}\label{implEP1.1}\\
		\NDLineTG{G:\bool}{see above}\label{implEP1.2}\\
		\NDLineTG{F\impl G:\bool}{\ruleRef{implType'},(\ref{implEP1.1}),(\ref{implEP1.2})}\nonumber
	\end{align}
	Clearly, $\Gamma\dedT F\impl G$ and $\Gamma\dedT F$ are quasi-preimages of the two validity assumptions of the rule, so we have proven the first part of the claim.
	
	Regarding the second part:
	\begin{align}
		\NDLineTG{F\impl G}{\byAss}\label{implEC1}\\
		\NDLineTG{F}{\byAss}\label{implEC2}\\
		\NDLineTG{G}{\ruleRef{implE},(\ref{implEC1}),(\ref{implEC2})}\nonumber
	\end{align}
	
	\subparagraph{\ruleRef{boolExt}:}
	Since the conclusion is proper, it follows that the preimage 
	\[
	\Gamma\dedT \univQuant{x}{\bool}p\ x
	\] of the normalization \[
	\PhiAppl{\Gamma}\dedPT \univQuant{x}{\bool}\univQuant{y}{\bool}\termEqT{\bool}{x}{y}\impl \termEqT{\bool}{\left(\lambdaFun{x}{\bool}\T\right)\ x}{\quasiImage{p}\ y}\] of the conclusion is well-formed and thus $\Gamma\dedT \univQuant{x}{\bool}p\ x:\bool$.
	Expanding the definition of $\forall$ yields:
	\begin{align}
		\NDLineTG{\lambdaFun{x}{\bool}\T\termEquals{\piType{x}{\bool}\bool}\lambdaFun{x}{\bool}p\ x:\bool\QNegSp}{\Quad see above}\label{boolExtP1.1}\\
		\NDLineTG{\lambdaFun{x}{\bool}p\ x:\piType{x}{\bool}\bool}{\ruleRef{eqTyping},\ruleRef{sym},(\ref{boolExtP1.1})}\label{boolExtP1.2}\\
		\NDLineTG{\left(\lambdaFun{x}{\bool}p\ x\right)\ \T:\bool}{\ruleRef{appl},(\ref{boolExtP1.2})}\label{boolExtP1.3}\\
		\NDLineTG{\left(\lambdaFun{x}{\bool}p\ x\right)\ \F:\bool}{\ruleRef{appl},(\ref{boolExtP1.2})}\label{boolExtP1.4}\\
		\NDLineTG{p\ \T\termEqB \left(\lambdaFun{x}{\bool}p\ x\right)\ \T}{\ruleRef{sym},\ruleRef{beta},(\ref{boolExtP1.3})}\label{boolExtP1.5}\\
		\NDLineTG{p\ \F\termEqB \left(\lambdaFun{x}{\bool}p\ x\right)\ \F}{\ruleRef{sym},\ruleRef{beta},(\ref{boolExtP1.4})}\label{boolExtP1.6}\\
		\NDLineTG{p\ \T:\bool}{\ruleRef{eqTyping},(\ref{boolExtP1.5})}\nonumber\\
		\NDLineTG{p\ \F:\bool}{\ruleRef{eqTyping},(\ref{boolExtP1.6})}\nonumber
	\end{align}
	Since $\Gamma\dedT p\ \T$ and $\Gamma\dedT p\ \F$ are clearly quasi-preimages of the two validity assumptions of the rule, we have proven the first part of the claim.
	
	Regarding the second part:
	\begin{align}
		\NDLineTG{p\ \T}{\byAss}\label{boolExtC1}\\
		\NDLineTG{p\ \F}{\byAss}\label{boolExtC2}\\
		\NDLineTG{\univQuant{x}{\bool}p\ x}{\ruleRef{boolExt},(\ref{boolExtC1}),(\ref{boolExtC2})}
	\end{align}
\end{proof}

\end{appendix}
\end{document}